\documentclass{amsart}

\usepackage{enumerate}
\usepackage{amsmath,amsthm}
\usepackage{amssymb}
\usepackage{graphicx}
\usepackage{color}
\usepackage{epic}
\usepackage{color}

\def\poisson#1{\left\{#1\right\}}

\def\W{\mathcal{W}}
\def\dt{\Delta t}
\def\nc{m}
\def\manq{\Sigma}
\def\sgram{\Gamma}

\def\E{\mathbb{E}}

\def\R{\mathbb{R}}

\def\T{\mathbb{T}}

\def\mS{\mathcal{S}}
\def\Z{\mathbb{Z}}

\def\tr{{\rm Tr}}
\def\rme{{\rm e}}
\def\ph{\varphi}

\def\one{\mathbf{1}}
\def\dps{\displaystyle}

\newcommand{\op}[1]{{\rm #1}}
\newcommand{\abs}[1]{\left | #1\right |}
\newcommand{\set}[1]{\left\{#1\right\}}
\newcommand{\setbig}[1]{\big\{#1\big\}}
\newcommand{\pare}[1]{ \left(#1\right) }
\newcommand{\norm}[1]{\left\Vert#1\right\Vert}

\newcommand{\Frac}[2]{\frac{\dps #1}{ \dps #2}}

\newcommand{\bmat}{\begin{pmatrix}}
\newcommand{\emat}{\end{pmatrix}}
\newcommand{\EL}{\mathrm{(CL)}}
\newcommand{\NL}{\mathrm{(SCL)}}

\newtheorem{theorem}{Theorem}[section]
\newtheorem{lemma}[theorem]{Lemma}
\newtheorem{proposition}[theorem]{Proposition}

\newtheorem{algorithm}[theorem]{Algorithm}
\theoremstyle{definition}
\newtheorem{definition}[theorem]{Definition}

\theoremstyle{remark}
\newtheorem{remark}[theorem]{Remark}
\numberwithin{equation}{section}

\def\sqw{\hbox{\rlap{\leavevmode\raise.3ex\hbox{$\sqcap$}}$%
\sqcup$}}
\def\cqfd{\ifmmode\sqw\else{\ifhmode\unskip\fi\nobreak\hfil
\penalty50\hskip1em\null\nobreak\hfil\sqw
\parfillskip=0pt\finalhyphendemerits=0\endgraf}\fi}

\begin{document}

\title[Langevin dynamics with constraints]{Langevin dynamics with constraints and computation of free energy differences}

\author{Tony Leli\`evre}
\address{Universit\'e Paris Est, CERMICS and INRIA, MICMAC project-team 
    Ecole des Ponts ParisTech, 6 \& 8 Av. Pascal, 77455 Marne-la-Vall\'ee, France}
\email{lelievre@cermics.enpc.fr}

\author{Mathias Rousset}
\address{INRIA Lille - Nord Europe, Parc Scientifique de la Haute Borne, 40 avenue Halley, Bât.A Park Plaza, 59650 Villeneuve d'Ascq, France}
\email{mathias.rousset@inria.fr}

\author{Gabriel Stoltz}
\address{Universit\'e Paris Est, CERMICS and INRIA, MICMAC project-team 
    Ecole des Ponts ParisTech, 6 \& 8 Av. Pascal, 77455 Marne-la-Vall\'ee, France}
\email{stoltz@cermics.enpc.fr}

\thanks{
  We would like to thank the anonymous referee for a careful 
  reading of the manuscript and useful suggestions.
  This work is supported by the Agence Nationale de la Recherche, 
  under the grant ANR-09-BLAN-0216-01 (MEGAS)
}


\date{June 23, 2010}


\begin{abstract}
In this paper, we consider Langevin processes with mechanical
constraints. The latter are a fundamental tool in molecular dynamics
simulation for sampling purposes and for the computation of free energy
differences. The results of this paper can be divided into three
parts. 
(i) We propose a simple discretization of the constrained
Langevin process based on a splitting strategy. We show how 
to correct the scheme so that it samples {\em exactly} the canonical
measure restricted on a submanifold, using a Metropolis-Hastings correction
in the spirit of the Generalized Hybrid Monte Carlo (GHMC) algorithm.
Moreover, we obtain, in some limiting regime, a consistent discretization
of the overdamped Langevin (Brownian) dynamics on a
submanifold, also sampling exactly the correct canonical measure with constraints. 
(ii) For free energy computation using thermodynamic
integration, we rigorously prove that the longtime average of the
Lagrange multipliers of the constrained Langevin dynamics yields the
gradient of a rigid version of the free energy associated with the
constraints. A second order time discretization using the Lagrange multipliers
is proposed. (iii) The Jarzynski-Crooks fluctuation relation is proved
for Langevin processes with mechanical constraints evolving in time. An original
numerical discretization without time discretization error is proposed, and its overdamped
limit is studied. Numerical illustrations are provided for (ii) and (iii).
\end{abstract}

\maketitle

\section{Introduction and main results}

Free energy is a central concept in thermodynamics and in modern
works on biochemical and physical systems. Typical
examples studied by computer simulations include the solvation free
energies (which is the free energy difference between a molecule {\it
  in vacuo} and its counterpart surrounded by solvent molecules) and the
binding free energy of two molecules (which 
determines whether a new drug can have an efficient action on a given
protein). In many applications, it is actually the
free energy difference profile between the initial and the final state
which is a quantity of paramount importance. It is observed 
by practitioners that free energy barriers are a very important
element to describe transition kinetics from one state to the other.
For instance, the chemical
kinetics of reactions happening in solvent (such as in the cells of our bodies)
are limited by free energy barriers, and can take place only when the
free energy difference between the initial and the final state is
negative, or at least less than the typical thermal energy. It is
therefore very important to accurately compute free energy differences
in order to assess the likelihood of a certain physical event to happen.

Beside these physical motivations to compute free energy differences,
a more abstract motivation is to overcome sampling barriers encountered
when computing canonical averages 
(see the discussion in~\cite[Section~1.3.3]{LelRouStoBook}). 
Indeed,
it is often the case in practice that the trajectories
generated by the numerical methods at hand remain stuck for a long time 
in some region of the phase space, and hop only occasionally to
another region, where they also remain stuck -- a behavior known as 
metastability. 
Chemical and physical intuitions may guide the practitioners of the
field towards the choice of some slowly evolving degree of freedom,
called \emph{reaction coordinate} in the following,
responsible for the metastable behavior of the system. In this case,
free energy techniques can be used to accelerate the 
sampling. 
This viewpoint allows to consider applications which are not motivated
by physical or biological problems, 
such as curing sampling issues in Bayesian statistics~\cite{Chopin}.

\medskip

In this introductory section, we present the main results and give the outline
of the paper, highlighting the three main contributions of this work. We only briefly 
define the concepts we need in 
this general introduction, and refer the reader to the following sections (in particular
Section~\ref{sec:notation}) for further precisions on the mathematical objects at hand.

\subsection{General setting for molecular dynamics with constraints}

We consider mechanical systems with constraints. The configuration of 
a classical $N$-body system is denoted by
$(q,p) \in \R^{6N}$.
The results of the paper can be generalized {\it mutatis mutandis} to periodic boundary conditions ($q \in {\mathbb T}^{3N}$, where $\T  = \R / \Z$ denotes the one dimensional torus), or to systems with positions confined in a domain $q\in \mathcal D \subset \R^{3N}$.
The mass matrix of the system is assumed to be a constant strictly positive symmetric  matrix $M$. One could typically think of a diagonal matrix $M=\op{Diag}(m_1 \op{Id}_3,\cdots,m_N\op{Id}_3)\in \R^{3N\times 3N}$. The interaction potential is a smooth function $V:\R^{3N} \to \R$. 
The Hamiltonian of the system is assumed to be separable:
\[
H(q,p) = \frac{1}{2}p^T M^{-1}p + V(q) .
\]
In the present paper, the focus is on the canonical ensemble, which is the equilibrium probability distribution of microscopic states of a system at fixed temperature (fixed average energy). For systems without constraints, this ensemble is characterized by the probability distribution
\begin{equation}
  \label{eq:canonical}
  \dps  \mu(dq \, dp)  =  Z^{-1} \, {\rm e}^{-\beta H(q,p)} \, dq \, dp,
  \qquad
  Z = \int_{\mathbb{R}^{6N}} {\rm e}^{-\beta H},
\end{equation}
where $Z$ is the normalizing constant\footnote{The potential $V$ is assumed to be such that $Z < \infty$.} ensuring that $\mu$ is indeed a probability distribution,
and $\beta = (k_{\rm B}T)^{-1}$ is proportional to the inverse temperature.
One dynamics which admits the canonical measure~\eqref{eq:canonical} as an invariant measure is 
the Langevin dynamics (see for instance~\cite[Section~2.2.3]{LelRouStoBook} and references therein):
\begin{equation}
\label{eq:Langevin_std}
\left \{ \begin{aligned}
  d q_{t} & = M^{-1}p_t \, dt, &\\
  d p_{t} & = -\nabla V(q_{t}) \, dt -\gamma(q_t) M^{-1}p_t \, dt 
  + \sigma(q_t) \, d W_t, & 
\end{aligned} \right.
\end{equation}
where $W_t$ is a standard $3N$-dimensional Brownian motion,
and $\gamma(q), \sigma(q)$ are $3N \times 3N$ position dependent 
real matrices which are assumed to satisfy the fluctuation-dissipation identity
\begin{equation}
  \label{eq:FDR_constraints}
  \sigma(q)\, \sigma^T(q) = \frac{2}{\beta} \gamma(q).
\end{equation}
The Langevin dynamics can be seen as some modification of the Hamiltonian
dynamics with two added components: a damping term $-\gamma(q_t) M^{-1}p_t \, dt$ 
and a random forcing term $\sigma(q_t) \, d W_t$. The energy dissipation 
due to the damping is compensated by the random forcing in such a way that
the temperature of the system is~$T=(k_B \beta)^{-1}$ (with $k_{\rm B}$
Boltzmann's constant).

\medskip

We will consider positions subject to a $\nc$-dimensional mechanical constraint denoted by
\[
\xi(q) = \pare{ \xi_1(q) , \dots , \xi_\nc(q) }^T  = z \in  \R^{\nc}.
\]
As will become clear below, constrained systems appear in computational statistical physics in two
kinds of contexts (see {\it e.g.} \cite[Chapter~10]{Rap},
and~\cite{Dar07,LelRouStoBook} for applications to the computation of 
free energy differences, and \cite{Arn89,leimkuhler-reich-04} for 
mathematical textbooks dealing with constrained Hamiltonian dynamics):
\begin{enumerate}[(i)]
\item for free energy computations, where
$\xi$ is a given reaction coordinate parameterizing a transition
between ''states'' of interest ; 
\item when the system is subject to
  molecular constraints such as rigid covalent bonds, or rigid bond angles in
  molecular systems. 
\end{enumerate}
In the sequel, $\xi$ may be thought of at first reading as a reaction coordinate (case (i)). 
Section~\ref{TI-sec:MC} explains how to handle additional molecular constraints (case (ii)) 
within the same formalism.
In any case, the position of the system is
constrained onto the submanifold of co-dimension~$\nc$:
\begin{equation}
  \label{eq:sigmaz}
  \manq(z) = \Big \{ q \in \R^{3N} \ \Big | \ \xi(q) = z \Big \},
\end{equation}
and the associated phase space is the cotangent bundle denoted by
\begin{equation}
\label{eq:phasespace}
T^{\ast} \manq(z) = \Big \{ (q,p) \in \R^{6N}  \ \Big | 
\ q \in \manq(z), \, \, \nabla \xi (q) ^T M^{-1} p =0 \Big \}. 
\end{equation}
For a given $q \in \manq(z)$, the set of cotangent 
momenta is denoted by
\begin{equation}
\label{eq:T*qSz}
T^{*}_q\manq(z) = \Big \{ p\in \R^{3N} \ \Big | \ 
  \nabla \xi(q)^T \,  M^{-1}p =0\Big \}.
\end{equation}
The orthogonal projection on $T^{*}_q\manq(z)$ with respect to the scalar product induced by $M^{-1}$ is denoted
\begin{equation}
  \label{eq:proj}
  P_M(q) = \op{Id} -  \nabla \xi(q) \, G_{M}^{-1}(q) \nabla \xi(q)^T M^{-1},
\end{equation}
where $G_M(q)$ is the Gram matrix associated with the constraints
\begin{align}
  \label{eq:Gram}
  G_{M}(q) =  \nabla \xi(q)^T  M^{-1} \, \nabla \xi(q).
\end{align}
Throughout the paper, we assume that $G_M$ is invertible everywhere on $\manq(z)$
(for all~$z$). It is easily checked that $P_M$ satisfies the projector property $P_M(q)^2 =  P_M(q)$, and the orthogonality property $$M^{-1} P_M(q) = P_M(q)^T M^{-1}.$$

\subsection{The constrained Langevin dynamics}
For constrained systems, the associated canonical
distribution is defined by
\begin{equation}
  \label{eq:canonicalconst}
    \dps  \mu_{T^{*}\manq(z)}(dq \, dp)  =  Z_{z,0}^{-1} \, {\rm e}^{-\beta H(q,p)} \,
  \sigma_{T^{*} \manq(z)}(dq \, dp),
\end{equation}
where $\sigma_{T^{*} \manq(z)}(dq \, dp)$ is the phase space
Liouville measure of $T^{*} \manq(z)$, and $Z_{z,0}$ the normalizing constant ($z$ refers to the position constraint, and $0$ to the velocity or momentum constraint, see~\eqref{eq:mu_v} below). 
See Section~\ref{sec:measures} for precise definitions.

A dynamics admitting the constrained canonical measure~\eqref{eq:canonicalconst} 
as an invariant equilibrium measure 
is the following Langevin process (``CL'' stands for ``constrained Langevin''): 
For a given initial condition $(q_0,p_0) \in T^*\manq(z)$,
\begin{equation*}
  \fbox{$\EL \qquad
  \begin{cases}
     \dps  d q_{t} = M^{-1}p_t \, dt, &\\[6pt]
   \dps  d p_{t} = -\nabla V(q_{t}) \, dt -\gamma(q_t) M^{-1}p_t \, dt 
    + \sigma(q_t) \, d W_t + \nabla \xi(q_{t}) \, d \lambda_t, & \\[6pt]
    \dps \xi(q_{t}) = z, & (C_q)
  \end{cases}$}
\end{equation*}
where the $\R^m$-valued adapted\footnote{{\it i.e.} a random variable depending only on the past
  values of the Brownian motion.} process $t \mapsto \lambda_t$ is the Lagrange multiplier associated
with the (vectorial) constraint $(C_q)$, and $\gamma(q), \sigma(q)$ are again assumed to satisfy~\eqref{eq:FDR_constraints}.
Note that $(q_t,p_t) \in T^*\manq(z)$ for all $t \geq 0$. Then, averages of an observable $A:\R^{6N} \to \R$ with respect to the distribution~\eqref{eq:canonicalconst} can be obtained as longtime averages along any trajectory of the dynamics~$\EL$
(when $P_M \gamma P_M^T$ is symmetric positive on~$\manq(z)$):
\begin{equation}
  \label{eq:longtime_cv}
  \lim_{T \to + \infty} \frac{1}{T} \int_0^T A(q_t,p_t) \, dt = \int_{T^{*} \manq(z)} A \, d \mu_{T^{*}\manq(z)} \qquad \mathrm{a.s.}
\end{equation}
This is made precise in Section~\ref{sec:sampling}. Several recent studies
({\it e.g.}
\cite{hartmann-schuette-05-b,hartmann-schuette-05-a,EveCic06,ciccotti-lelievre-vanden-einjden-08})
have analyzed dynamics similar to~$\EL$ and some
appropriate discretization of the process in order to approximate
the left-hand side of~\eqref{eq:longtime_cv}.

\underline{The first contribution} of our work is to propose a simple discretization
of the dynamics~$\EL$ and to highlight its remarkable properties. 
The numerical scheme is based on a splitting strategy
between the Hamiltonian and the thermostat part, see Equations~\eqref{eq:flucdiss1}-\eqref{eq:Verletconst}-\eqref{eq:flucdiss2} below (in the spirit of the scheme proposed in~\cite{BouOwh08} in the unconstrained case). 
The Hamiltonian part is discretized using a Verlet scheme with position and momentum constraints (the so-called RATTLE scheme, see~\cite{LeiSke94}). 
We show that this discretization enjoys the following properties: (i) for some choice of the parameters, 
an Euler discretization of the {\em overdamped} Langevin dynamics (also called Brownian dynamics)
with a projection step associated with the constraints is obtained (see
Equation~\eqref{eq:overdampconst} and
Proposition~\ref{p:langtooverd}); (ii) it can be completed by a
Metropolis-Hastings correction to obtain a Generalized Hybrid Monte Carlo (GHMC)
method sampling {\emph{exactly}} ({\em i.e.} without any bias due to time-discretization) the constrained canonical distribution~\eqref{eq:canonicalconst} (see Algorithm~\ref{a:GHMCconst} below).
The so-obtained numerical scheme is close to the ones proposed 
in~\cite{hartmann-schuette-05-a,hartmann-schuette-05-b,Hartmann08}. 
See also~\cite{DKPR87,Mackenzie89} for
historic references on Hybrid Monte Carlo methods, and~\cite{Hor91} for GHMC. One output of this part is thus a new Metropolization procedure for overdamped Langevin dynamics to sample, without bias, measures with support a submanifold.

\subsection{Free energy computations}
The free energy $F:\R^{\nc}\to \R$ associated with the reaction coordinates $\xi:\R^{3N} \to \R^{\nc}$ is defined as $-\beta^{-1}$ times the log-density of the marginal probability distribution of the reaction coordinates $\xi$ under the canonical distribution~\eqref{eq:canonical}. Explicitly, it 
is defined through the following relation: for any test function $\phi: \R^{\nc}\to \R$,
\begin{equation}
  \label{eq:F}
   \int_{\R^{\nc}} \phi(z) \, \op{e}^{-\beta F(z)} \, dz = \int_{\R^{6N}} \phi(\xi(q)) \, \mu(dq \, dp).
\end{equation}
In words, $\op{e}^{-\beta F(z)} \, dz$ is the image of the measure $\mu$ by $\xi$, and $F$ can be seen as an ``effective potential energy'' associated to $\xi$.

Computing the free energy profile $z \mapsto F(z)$ (up to an additive constant independent of~$z$), or free energy differences between two states $F(z_2)-F(z_1)$ is a way to compare the relative probabilities of different ''states'' parameterized by $\xi$. This is a very important calculation for practical applications, see~\cite{chipot-pohorille-07,LelRouStoBook}.
A state should be understood here as the collection
of all possible microscopic configurations $(q,p)$, distributed according to the canonical
measure~\eqref{eq:canonical}, and 
satisfying the macroscopic constraint $\xi(q) = z$.
Since we only focus on computing free energy differences, $F$ is defined up to an additive constant (independent of $z$, denoted by~${\rm C}$ below, and whose value may vary from line to line) and can be rewritten as:
\begin{equation}
  \label{eq:Fbis}
  \begin{aligned}
    F(z) & = - \frac{1}{\beta} \ln  \int_{\manq(z) \times \mathbb{R}^{3N}} 
    \rme^{-\beta H(q,p) } \, \delta_{\xi(q)-z}(dq) \, dp \\
    & = - \frac{1}{\beta} \ln  \int_{\manq(z)} \rme^{-\beta V(q) } \, \delta_{\xi(q)-z}(dq) + {\rm C},
\end{aligned}
\end{equation}
where $\delta_{\xi(q)-z}$ denotes the conditional measure on $\manq(z)$ verifying the following identity of measures in~$\R^{3N}$: $dq = \delta_{\xi(q)-z}(dq) \, dz$ (see Section~\ref{sec:measures} 
for more precisions on this relation).

However, when using constrained simulations in phase space, the momentum variable of the dynamical system is also constrained, and a modified free energy (called ``rigid free energy'' in the sequel, see Remark~\eqref{rem:highosclang} below for a justification of the term "rigid")
is more naturally computed, see Section~\ref{sec:sampling}. The latter is defined as
\begin{align}\label{eq:Frgd}
F_{\rm rgd}^M(z) 
&=  -\frac{1}{\beta} \ln \int_{T^\ast \manq(z)} {\rm e}^{-\beta H(q,p)} 
\sigma_{T^\ast \manq(z)}(dq\, dp) . 
\end{align}
The superscript $M$ indicates that this free energy depends on the considered mass matrix, even 
though this is not clear at this stage (see~\eqref{eq:explicit_dependence_M} below).
The above two definitions of free energy are related through the identity:
\begin{equation}
  \label{eq:FFtilde}
  F(z) - F_{\rm rgd}^M(z) = -\frac{1}{\beta} \ln \int_{T^{*}\manq(z)}   (\det G_M) ^{-1/2} d\mu_{T^{*}\manq(z)} + {\rm C},
\end{equation}
where $\mu_{T^{*}\manq(z)}$ is the equilibrium distribution with constraints~\eqref{eq:canonicalconst}. The relation~\eqref{eq:FFtilde}, already proposed in~\cite{Dar07} (see also~\cite{otter-00,EVandenEijnden,SchlitterKlahn03,hartmann-schuette-07} for related formulas), is proved at the beginning of Section~\ref{sec:Langti}. 
For any value of the reaction coordinate, the difference $F(z) - F_{\rm rgd}^M(z)$ can then be easily computed with any method sampling the probability distribution $\mu_{T^{*}\manq(z)}$, such as~$\EL$. 

Several methods have been suggested in the literature to compute either $F$ or $F_{\rm rgd}^M$ from the Lagrange multipliers of a constrained process similar to~$\EL$. We refer for instance to~\cite{Dar07} (and references therein) for the Hamiltonian case, and to~\cite{ciccotti-lelievre-vanden-einjden-08} (and references therein) for the overdamped case. 
\underline{The second contribution} of this paper is twofold: (i) we rigorously prove that the longtime average of the Lagrange multipliers in~$\EL$ converges to the gradient of the rigid free energy~\eqref{eq:Frgd} (the so-called \emph{mean force}); and (ii) we then show that the latter mean-force can be computed with second order accuracy (\emph{i.e.} up to $\mathrm{O}(\dt^2)$ error terms, where $\dt$ is the time-step) using the Lagrange multipliers involved in the Hamiltonian part of the splitting scheme.   

More precisely, the first point (i) amounts to showing that
\begin{equation}\label{eq:avlagr}
\lim_{T \to + \infty}\frac{1}{T}\int_0^T d \lambda_t  = \nabla_{z} F_{\rm rgd}^M(z) \qquad {\rm a.s}.
\end{equation}
As compared to \cite{Dar07}, where a formal proof for the Hamiltonian case is proposed, we use an explicit calculation that does not require the use of the Lagrangian structure of the problem, or a change of coordinates. Once $\nabla_{z} F_{\rm rgd}^M(z)$ is obtained, $F_{\rm rgd}^M(z)$ can be computed (up to an additive constant) by integration. This procedure is
known as {\em thermodynamic integration}.
Note that using~\eqref{eq:avlagr} and thermodynamic integration, together with \eqref{eq:FFtilde}, allows to obtain $F(z)$ without computing second order derivatives of~$\xi$. This is a desirable property since computing such high derivatives may be cumbersome for some reaction coordinates used in practice. Straightforward computations of the mean force using analytical expressions (see for instance~\eqref{eq:avfbar}-\eqref{eq:constforceav}) usually involve such high order derivatives.

The second point (ii) is then based on a discretization of a variant of~\eqref{eq:avlagr}, obtained by subtracting the martingale part of the Lagrange multipliers. This amounts to averaging the two Lagrange multipliers involved in the RATTLE part of the scheme, see~\eqref{eq:freeestim}).

We also discuss how these techniques can be generalized to compute the free energy for systems with molecular constraints, see Section~\ref{TI-sec:MC}.

\subsection{Jarzynski-Crooks relations and nonequilibrium computations of the free energy}
The last part of this article is devoted to nonequilibrium methods for free energy computations, based on a Hamiltonian or Langevin dynamics with constraints subject to a predetermined time evolution. Such methods rely on a nonequilibrium fluctuation equality, the so-called Jarzynski-Crooks relation. See~\cite{Jarzynski97PRL} for a pioneering work, as well as~\cite{Crooks98,Crooks99} for an extension. 
They are termed ``nonequilibrium'' since 
the transition from one value of the reaction coordinate $\xi$
to another one is imposed \textit{a priori}, in a finite time$~T$, and with a given smooth deterministic schedule $t \in [0,T] \mapsto z(t) \in \mathbb{R}^\nc$. 
In particular, it may be arbitrarily fast. Therefore, even if
the system starts at equilibrium, it does not remain at equilibrium. 
The out-of-equilibrium Langevin process we consider to this end 
is given by the following equations of motion
(``SCL'' stands for ``switched constrained Langevin''):
\begin{equation*}
  \fbox{$\NL \qquad
  \left\{
  \begin{aligned}
    d q_t & = M^{-1}p_t \, dt,\\[6pt]
    d p_t & = -\nabla V (q_t)  \,dt -\gamma_P(q_t) M^{-1} p_t \, dt+ \sigma_P(q_t)  \,d W_t  + \nabla \xi(q_t) \, d \lambda_t, \\[6pt]
    \xi(q_t) & = z(t), \hspace{5cm} (C_q(t))
  \end{aligned} \right.$}
   \end{equation*}
where $ t \mapsto \lambda_t \in \R^{\nc}$ is an adapted process enforcing the constraints $(C_q(t))$ (the Lagrange multipliers). 
Initial conditions are sampled from the phase-space canonical distribution 
defined by the constraints $\xi(q) = z(0)$ and 
$v_\xi(q,p) = \nabla\xi(q)^T M^{-1} p = \dot{z}(0)$ (see~\eqref{eq:mu_v}).
We restrict ourselves to projected fluctuation-dissipation matrices of the specific form
\begin{equation}
  \label{eq:gamma_P}
(\sigma_P,\gamma_P) := ( P_M \, \sigma , P_M \, \gamma \, P_M^T),
\end{equation}
where $\gamma(q), \sigma(q) \in \mathbb{R}^{3N \times 3N}$ satisfy the 
fluctuation-dissipation identity~\eqref{eq:FDR_constraints}. Note that 
$\gamma_P,\sigma_P$ also verify~\eqref{eq:FDR_constraints}. 
Our analysis also applies to deterministic Hamiltonian dynamics upon choosing $\gamma = 0$.
The dynamics~$\NL$ 
is a natural extension of the constrained Langevin dynamics~$\EL$.
It is different from the dynamics proposed in~\cite{LHS10},
which is a Langevin dynamics associated with a modified Hamiltonian 
with projected momenta, driven by a forcing term along $\nabla \xi$ which acts 
directly on the position variable.
As explained below (see~\eqref{eq:laglangjarz} and the discussion following
this equation), 
the specific choice~\eqref{eq:gamma_P} (rather than considering unprojected matrices 
$\gamma(q), \sigma(q) \in \mathbb{R}^{3N \times 3N}$)
leads to a simpler analysis and more natural numerical schemes,
based again on a splitting procedure.

As explained in Section~\ref{sec:jarz_lang_cons}, 
it is possible to define the work associated with the constraints exerted on the system between time $0$ and $T$ as the displacement multiplied by the constraining force:
\begin{align}\label{eq:work_intro}
    \W_{0,t}\pare{ \set{q_s,p_s}_{0 \leq s \leq t} } & : =  \int_0^t \dot{z}^T(s) \,  d \lambda_s.
\end{align}
\underline{The third contribution} of the present paper is twofold: (i) We derive a new general Crooks-Jarzynski relation (see Theorem~\ref{th:crookslangcons} below) based on the nonequilibrium constrained dynamics~$\NL$ and the associated work defined in~\eqref{eq:work_intro}; and (ii) An original numerical scheme is proposed, which allows to compute free energy differences \emph{without time discretization error} (see Theorem~\ref{eq:jarz_disc} below). More precisely, concerning the first point, the main corollary is given by the following result. Consider the corrector
\begin{equation}\label{eq:corr}
  C(t,q) = \frac{1}{2\beta} \ln \Big( \det G_M(q) \Big) -
  \frac{1}{2} \dot{z}(t)^T G_M^{-1}(q) \dot{z}(t),
\end{equation}
where $\frac{1}{2\beta} \ln \det G_M(q)$ is the so-called Fixman term due to the geometry of the position constraints (see~\eqref{eq:FFtilde} and Remark~\ref{rem:highosclang}), and $\frac{1}{2} \dot{z}(t)^T G_M^{-1}(q) \dot{z}(t)$ is the kinetic energy term due to the  velocity of the switching.
Then, the free energy profile can be computed through the following fluctuation identity
(see~\eqref{eq:identity_to_approximate}):
\begin{equation}
\label{eq:FK_multi_bis_corr}
F(z(T)) - F(z(0)) = -\frac1\beta \ln  \left( \frac{\E\left(\rme^{-\beta \left[\W_{0,T}\pare{ \set{q_t,p_t}_{0 \leq t \leq T} }  +C(T,q_T)   \right]  }\right)}{\E\left(\rme^{-\beta C(0,q_0)   }\right)} \right),
\end{equation}
where the expectation is with respect to canonical (equilibrium) initial conditions
and for all realizations of the dynamics~$\NL$.

The numerical scheme mentioned in the second point (ii) above is based on a modification of the splitting scheme used to discretize the constrained Langevin dynamics~$\EL$. This modification allows to take into account the evolving constraints. Using the symplecticity of the modified RATTLE scheme, we are able to prove a discrete-in-time version of the Crooks relation, and of the associated Jarzynski free energy estimator~\eqref{eq:FK_multi_bis_corr}. Moreover, for some choice of the parameters, the latter scheme yields a Jarzynski-Crooks relation for an Euler discretization of the overdamped Langevin (Brownian) dynamics with a projection step associated with the evolving constraints, without time discretization error. This can be seen as an extension of the scheme formerly proposed in~\cite{lelievre-rousset-stoltz-07-a} (see Equation~\eqref{eq:Eulerconst_jarz} and Proposition~\ref{p:langtooverd_jarz}). We also check the consistency of the various free energy estimators we introduce.

\subsection{Organization of the paper}
We start with an introduction to the mathematical concepts required 
for mechanically constrained systems in
Section~\ref{sec:notation}. 
Section~\ref{sec:sampling} is devoted to the properties and the discretization of mechanically constrained Langevin processes defined by~$\EL$, and the problem of sampling the canonical distribution~\eqref{eq:canonicalconst}. 
Thermodynamic integration with constrained Langevin processes is presented 
in Section~\ref{sec:Langti}. 
Section~\ref{sec:lang_jarz} discusses nonequilibrium constrained Langevin processes~$\NL$ and the associated Jarzynski-Crooks fluctuation identity~\eqref{eq:FK_multi_bis_corr}. 
Finally, some technical lemmas are gathered in Section~\ref{sec:appendix}.

\section{Preliminaries}
\label{sec:notation}

After making precise our notation for matrices and matrix valued functions 
in Section~\ref{sec:matrix_notation}, we introduce some additional concepts required to describe
constrained systems in Section~\ref{sec:const_notation}, and define the phase space measures
with constraints in Section~\ref{sec:measures}.

\subsection{Notation}
\label{sec:matrix_notation}
Throughout the paper, the following notation is used:
\begin{itemize}
\item Vectors and vector fields are by convention of column type. When vectors are written as 
  a line, they should be understood as the corresponding
  column version.
  For instance, $(q,p) \in \mathbb{R}^{6N}$ should be understood as $(q^T,p^T)^T$, where 
  $q,p\in \mathbb{R}^{3N}$ are both column vectors.
\item Gradients in $\R^{3N}$ (or $\R^{6N}$) of $\nc$-dimensional vector fields are by convention $3N \times \nc $-matrices, for instance:
\[
\nabla \xi(q) = \Big( \nabla \xi_1(q),\dots,\nabla \xi_\nc(q) \Big )
\in \mathbb{R}^{3N \times \nc},
\]
where $\nabla \xi_i(q) \in \mathbb{R}^{3N}$ is a column vector
for any $i=1,\dots,\nc$. Gradients in the space of constraints parameters  $z\in\R^{\nc}$ or $\zeta \in \R^{2 \nc}$ are denoted with the associated subscripts, namely $\nabla_z$ and $\nabla_\zeta$.
\item Second order derivatives in $\R^{3N}$ of $\nc$-dimensional vector fields are characterized through the Hessian bilinear form:
\begin{equation}
  \label{eq:hessnot}
  \op{Hess}_q(\xi)\big(v_1,v_2\big) = \bmat v_1^T \nabla^2 \xi_1(q) v_2 \\ \vdots \\
  v_1^T \nabla^2 \xi_\nc(q) v_2 \emat \in \R^\nc,
\end{equation}
where $v_1,v_2 \in \R^{3N}$ are test vectors.
\item The canonical symplectic matrix is denoted by:
\begin{equation}
  \label{eq:J}
  J:=\bmat
  0 & {\rm Id}_{3N} \\
  - {\rm Id}_{3N} & 0 \emat \in \R^{6N \times 6N} .
\end{equation}
For any smooth test functions $\ph_1 \, : \, \R^{6N} \to \R^{n_1}$ and 
$\ph_2 \, : \, \R^{6N} \to \R^{n_2}$, the Poisson bracket is the $n_1\times n_2$ matrix
\begin{equation}
  \label{eq:poisson}
  \poisson{\ph_1 ,\ph_2 } = \big(\nabla \ph_1 \big)^T J \, 
\nabla \ph_2 \in \R^{n_1\times n_2} .
\end{equation}
\item For two matrices $A,B \in \mathbb{R}^{n \times n}$, $A : B = \tr(A^TB)$.
\end{itemize}

\subsection{Constraints}
\label{sec:const_notation}

Contrarily to what
is often done in the literature, we avoid global changes of variables and 
the use of generalized coordinates.
We observe that global changes of variables are not required for the 
proofs of the theoretical results we present, and they are 
definitely to be avoided in practical numerical computations whenever possible.

Two useful concepts to study constrained Hamiltonian systems (in particular, to
use the co-area formula in phase space, as well as the Poisson bracket
formulation of the Liouville equation) are the
effective velocity $v_\xi$ and the effective momentum
$p_\xi$ associated with the constrained degrees of freedom $\xi$:
\begin{equation}
  \label{eq:effv}
  v_\xi(q,p) = \nabla \xi(q)^T M^{-1} p \in \R^{\nc},
\end{equation}
and
\begin{equation}
  \label{eq:effp}
  p_\xi(q,p) = G_{M}^{-1}(q) \, v_\xi(q,p) 
  = G_{M}^{-1}(q) \nabla \xi(q)^T M^{-1} p \in \R^{\nc}.
\end{equation}
The expression of the effective velocity is obtained by deriving the 
constraint $\xi$ along an unconstrained trajectory of the Hamiltonian dynamics
\[
\begin{cases}
  \dps  \frac{d \tilde q_t}{dt} = M^{-1}\tilde p_t, \\[8pt]
  \dps  \frac{d \tilde p_t}{dt} = -\nabla V(\tilde q_t),
\end{cases}
\]
since
$
\frac{d \xi(\tilde q_t)}{dt} = v_\xi(\tilde q_t,\tilde p_t).
$
The term $G_{M}^{-1}(q)$ in the expression~\eqref{eq:effp}
of the effective momentum may be interpreted as the
effective mass of $\xi$.
This can be motivated by a decomposition of the kinetic
energy of the system into tangential and orthogonal parts, 
using the projector \eqref{eq:proj} for a given position 
$q \in \R^{3N}$:
\[
\begin{aligned}
E_{\rm kin}(p) 
& = \frac12 \, p^T M^{-1} p \\
& = \frac{1}{2} \, p^T P_M(q)^T M^{-1} P_M(q) p + \frac{1}{2} \, p^T
\big(\op{Id}-P_M(q)\big)^T M^{-1} \big(\op{Id}-P_M(q)\big) p. 
\end{aligned}
\]
The orthogonal part can be rewritten, for any $(q,p) \in \R^{6N}$, as:
\begin{align*}
  E_{\rm kin}^{\perp}(q,p) 
  &:= \frac{1}{2} \, p^T \Big(\op{Id}-P_M(q) \Big)^T M^{-1} \Big(\op{Id}-P_M(q)\Big) p \\ 
  &= \frac{1}{2} \, v_\xi(q,p)^T G_{M}^{-1}(q) \, 
  v_\xi(q,p) =\frac{1}{2} \, p_\xi(q,p)^T G_{M}(q) \, 
  p_\xi(q,p).
\end{align*}
The last equations allow to consider $G_{M}^{-1}$ as some effective mass.

The constraints on a mechanical system can also be reformulated in the more general form
\begin{equation}
  \label{eq:fullconst}
  \dps \Xi(q,p) = \zeta \in \R^{2 \nc},
\end{equation}
where either (i) the effective momentum is constrained, in which case $ \Xi = (\xi , p_{\xi})$ and $\zeta = (z,p_z)$; or (ii) the effective velocity is constrained, in which case $\Xi = (\xi , v_{\xi})$ and $\zeta = (z,v_z)$. The phase space associated with such constraints is denoted by
\begin{equation}
  \label{eq:phaseXi}
  \manq_{\Xi}(\zeta) = \Big\{ (q,p) \in \R^{6N} \ \Big|
  \ \Xi(q,p) = \zeta \Big\}.
\end{equation}
A position $q\in \Sigma(z)$ being given, the affine space of constrained momenta verifying~\eqref{eq:fullconst} is then denoted by
\begin{equation}
  \label{eq:phasevxi_q}
  \manq_{v_\xi(q,\cdot)}(v_z) = \Big\{ p \in \R^{3N} \ \Big|
  \ v_\xi(q,p) = v_z \Big\}
\end{equation}
in the effective velocity case, and by $\manq_{p_\xi(q,\cdot)}(p_z)$ in the effective momentum case. 
This notation is very important for nonequilibrium
methods where the constraints evolve in time according to a predefined schedule, see
Section~\ref{sec:lang_jarz}. 
Note that the phase space of mechanical constraints, defined by
\eqref{eq:phasespace}, is simply $T^*\manq(z) = \manq_{\xi,v_\xi}(z,0) = \manq_{\xi,p_\xi}(z,0)$.

We can now define the skew-symmetric Gram tensor of dimension $2\nc \times 2\nc$
  associated with the constraints:
  \begin{equation}
    \label{eq:sgram}
    \sgram(q,p) = \lbrace\Xi,\Xi \rbrace (q,p) = \nabla\Xi^T(q,p) \, J \, \nabla \Xi (q,p) \in \R^{2\nc \times 2\nc}.
  \end{equation}
The Gram matrix $\Gamma$ associated with the generalized constraints
\eqref{eq:fullconst} can be explicitly computed
by block. Indeed, for $\Xi = (\xi,p_\xi)^T$,
\begin{equation}
  \label{eq:Gram_pxi}
  \sgram = \bmat 0 & {\rm Id} \\ - {\rm Id} & \ \nabla p_\xi^T J\,
  \nabla p_\xi \ \\ \emat.
\end{equation}
Therefore, ${\rm det } ( \sgram) = 1$ in this case.
In the case $\Xi=(\xi,v_\xi)^T$, the Gram matrix reads
\begin{equation}
  \label{eq:Gram_vxi}
  \sgram = \bmat 0 & G_{M} \\ - G_{M} & \ \nabla v_\xi^T
  \, J \, \nabla v_\xi \ \emat,
\end{equation}
and ${\rm det } ( \sgram) = {\rm det } (G_{M})^2$. Note that in both cases ${\rm det } ( \sgram) > 0$.
The constrained symplectic (skew-symmetric) matrix is now defined by
\begin{equation}
  \label{eq:constJ}
  J_{\Xi}(q,p) = J - J \, \nabla \Xi(q,p) \, \sgram^{-1}(q,p) \, \nabla \Xi^T(q,p) \, J ,
\end{equation}
and the Poisson bracket associated with generalized 
constraints~\eqref{eq:fullconst} by:
\begin{align} 
  \label{eq:poissonconst}
  \poisson{\ph_1,\ph_2}_{\Xi} 
  &= \nabla \ph_1 ^T  J_{\Xi} \nabla \ph_2.
\end{align}
This Poisson bracket is often called the Dirac bracket in the literature
(in reference to the seminal work of Dirac~\cite{Dirac50,MarsdenRatiu03}).

It is easily checked that $\poisson{\cdot,\cdot}_{\Xi}$ 
verifies the characteristic properties of Poisson brackets, namely the skew-symmetry, 
Jacobi's identity, and Leibniz' rule. 
Therefore, the flow associated with the evolution equation
\begin{equation}
  \label{eq:Hconstgen}
  \frac{d }{dt} \bmat q_t \\ p_t \emat = J_{\Xi} \nabla H (q_t,p_t),
\end{equation}
defines a symplectic map. 
Recall that (see~\cite[Section~VII.1.2]{HairerLubichWanner06})
a map $\phi \, : \, \Sigma_\Xi(\zeta) \to \Sigma_\Xi(\zeta)$ is symplectic if
for any $(q,p) \in \Sigma_\Xi(\zeta)$ and $u,v \in T_{(q,p)}\Sigma_\Xi(\zeta)$,
\[
u^T \nabla\phi(q,p)^T J \nabla\phi(q,p) v = u^T J v.
\]
A consequence of the symplectic structure is the divergence formula~\eqref{eq:divsympl} relating the phase space measure
$\sigma_{\manq_{\Xi}(\zeta)}(dq \, dp)$ on $\manq_{\Xi}(\zeta)$ (defined below), and the Poisson bracket~\eqref{eq:poissonconst}. The reader is referred to Chapter~$8$ in~\cite{Arn89}, and Section~VII.$1$ in \cite{HairerLubichWanner06} for more material on constrained systems. 

It will be shown in Proposition~\ref{s:p:Lconstgen} below 
that the Poisson system~\eqref{eq:Hconstgen} is equivalent to~$\EL$
when $(\gamma,\sigma)=(0,0)$.

\subsection{Phase space measures}
\label{sec:measures}

\subsubsection{Definitions}

The phase space measure (also termed Liouville measure) on the phase space $T^*\manq(z)$ (or more generally on $\manq_{\Xi}(\zeta)$)
of constrained mechanical systems is denoted by $\sigma_{T^*\manq(z)}$ (or more generally $\sigma_{\manq_{\Xi}(\zeta)}$). The latter is induced by the symplectic, or skew-symmetric $2$-form on 
$\R^{6N}$ defined by the canonical skew-symmetric matrix $J$ in $\R^{6N}$. More precisely, it can be defined through the volume form $\abs{\op{det}\,\mathcal{G}(u(q,p))}^{1/2}$, where
\[
\mathcal{G}_{a,b}(u) = (u_a)^T J u_b, \qquad a,b =1,\ldots,6N-2m,
\]
and $(u_1(q,p),\ldots,u_{6N-2 m }(q,p))$ is a basis of tangential vectors of the 
submanifold~$T^\ast\manq(z)$ (or~$\manq_{\Xi}(\zeta)$) at a given point $(q,p)$.

Surface measures induced by scalar products associated with general symmetric definite
positive matrices will also be of interest. We denote
by $\sigma^{M}_{\manq(z)}(dq)$ the surface measure on $\manq(z)$ induced by the scalar
product $\langle q,\tilde{q}\rangle_M = q^T M \tilde{q}$ on $\R^{3N}$, and, 
for a given $q \in \manq(z)$, by $\sigma^{M^{-1}}_{\manq_{p_\xi(q,\cdot)}(p_z)}(dp)$ and 
$\sigma^{M^{-1}}_{\manq_{v_\xi(q,\cdot)}(v_z)}(dp)$ the
surface measures on the affine
spaces $\manq_{p_\xi(q,\cdot)}(p_z)$ and $\manq_{v_\xi(q,\cdot)}(v_z)$ respectively, 
induced by the scalar product
$\langle p,\tilde{p}\rangle_{M^{-1}} = p^T M^{-1} \tilde{p}$ on~$\R^{3N}$.
For more precise definitions of these measures, we refer 
to~\cite[Sections~3.2.1 and~3.3.2]{LelRouStoBook} 
and the references therein.

It is now possible to define a generalization of the 
canonical distribution~\eqref{eq:canonicalconst} as follows:
\begin{equation}
  \label{eq:mu_v}
  \begin{cases}
    \dps  \mu_{\manq_{\xi,v_\xi}(z,v_z)}(dq\, dp ) := \frac{ {\rm e}^{-\beta H(q,p)}}{Z_{z,v_z}} \sigma_{\manq_{\xi,v_\xi}(z,v_z)}(dq\, dp ), \\
    \dps Z_{z,v_z} := \int_{\manq_{\xi,v_\xi}(z,v_z)} {\rm e}^{-\beta H} \, d \sigma_{\manq_{\xi,v_\xi}(z,v_z)} .
  \end{cases}
\end{equation} 
The distribution~\eqref{eq:mu_v} is associated with the generalized constraints~$\Xi = \pare{ \xi , v_{\xi} }^T$, and is used in Section~\ref{sec:lang_jarz} for nonequilibrium methods. Note that $\mu_{\manq_{\xi,v_\xi}(z,0)}=\mu_{T^\ast \manq(z)}$ defined in~\eqref{eq:canonicalconst}.

\subsubsection{Co-area decompositions}
The co-area formula (see~\cite{ambrosio-fusco-pallara-00,evans-gariepy-92}) relates the phase space or surface measures, and the conditional measures. Conditional measures are defined in $\R^{6N}$ by the following conditioning formula: for any test function~$\phi: \R^{6N} \to \R$,
\begin{equation}
  \label{eq:conditional_phase_space_measure}
 \int_{\R^{6N}} \phi(q,p) \, dq \, dp = \int_{\R^{2 \nc}} \! \int_{\manq_{\Xi}(\zeta)} \phi(q,p) \, \delta_{\Xi(q,p)-\zeta}(dq \, dp) \, d \zeta .
\end{equation}
In the same way in $\R^{3N}$, conditional measures are defined, for any test function~$\phi: \R^{3N} \to \R$, by
\begin{equation}
  \label{eq:conditional_measure}
 \int_{\R^{3N}} \phi(q) \, dq = \int_{\R^{ \nc}} \! \int_{\manq(z)} \phi(q) \, \delta_{\xi(q)-z}(dq) \, d z.
\end{equation}
A more concise notation for the above equalities is 
$dq \, dp  = \delta_{\Xi(q,p)-\zeta}(dq \, dp) \, d\zeta$ and $dq = \delta_{\xi(q)-z}(dq) \, dz$.

\begin{proposition}[Co-area]
  \label{p:coareasympl}
  Let $\manq(z)$ be the submanifold~\eqref{eq:sigmaz} defined by the constraints $\xi(q) = z$, and assume that $G_M$ defined in~\eqref{eq:Gram} is non-degenerate in a neighborhood of $\manq(z)$. Then, in the sense of measures on $\R^{3N}$:
\begin{equation}
    \label{eq:coarea}
    \delta_{\xi(q)-z}(dq) = 
    \pare{{\rm det  }\, M}^{-1/2} \big| {\rm det  } \, G_M(q) \big|^{-1/2} \,
    \sigma^{M}_{\manq(z)}(dq ).
  \end{equation}
  Let $\manq_{\Xi}(\zeta)$ be the phase space defined by generalized
  constraints~\eqref{eq:fullconst}.
  Assume that~$\sgram$ defined in~\eqref{eq:sgram} is non-degenerate in a neighborhood of
  $\manq_{\Xi}(\zeta)$. Then, in the sense of measures on $\R^{6N}$:
  \begin{equation}
    \label{eq:coareasympl}
    \delta_{\Xi(q,p)-\zeta}(dq \, dp) = \big| {\rm det  } \, \sgram(q,p) \big|^{-1/2} \,
    \sigma_{\manq_{\Xi}(\zeta)}(dq \, dp).
  \end{equation}
\end{proposition}
We refer for example to Chapter~$3$ in~\cite{LelRouStoBook} for an elementary proof. 
An equivalent of \eqref{eq:coarea}-\eqref{eq:conditional_measure} for momenta reads,
for constrained effective momenta:
\begin{equation}
  \label{eq:coarea_effective_momentum}
  dp = \delta_{p_\xi(q,p)-p_z}(dp) \, dp_z = 
  \op{det}(M)^{1/2} \, \big| {\rm det  } \, G_M(q) \big|^{1/2} \,
  \sigma^{M^{-1}}_{\manq_{p_\xi(q,\cdot)}(p_z)}(dp) \, dp_z,
\end{equation}
and for constrained effective velocities:
\[
dp = \delta_{v_\xi(q,p)-v_z}(dp) \, dv_z = 
\op{det}(M)^{1/2} \, \big| {\rm det  } \, G_M(q) \big|^{-1/2}
\, \sigma^{M^{-1}}_{\manq_{v_\xi(q,\cdot)}(v_z)}(dp)\, dv_z.
\]
Using the co-area formulas~\eqref{eq:coarea}-\eqref{eq:coareasympl}, and the expressions of symplectic Gram matrices~\eqref{eq:Gram_pxi}-\eqref{eq:Gram_vxi}, we obtain the following expressions of the phase space measures:
\begin{enumerate}[(i)]
\item The phase space measure on $\manq_{\xi,p_\xi}(z,p_z)$ can be identified 
  with the conditional measure defined in~\eqref{eq:conditional_phase_space_measure}:
\begin{equation}
  \label{eq:deltasigmap}
  \sigma_{\manq_{\xi,p_\xi}(z,p_z)}(dq \, dp) =
  \delta_{(\xi(q)-z,p_\xi(q,p)-p_z)}(dq \, dp),
\end{equation}
while the phase space measure on $\manq_{\xi,v_\xi}(z,v_z)$ is
related to the corresponding conditional measure as
\begin{equation}
  \label{eq:deltasigmav}
  \sigma_{\manq_{\xi,v_\xi}(z,v_z)}(dq \, dp) = {\rm det } (G_{M})
  \, \delta_{(\xi(q)-z,v_\xi(q,p)-v_z)}(dq \, dp).
\end{equation}
\item The phase space measures are given by the product of surface
  measures:
  \begin{equation}
    \label{eq:product_decomposition_pz}
  \sigma_{\manq_{\xi,p_\xi}(z,p_z)}(dq \, dp) =
  \sigma^{M^{-1}}_{\manq_{p_\xi(q,\cdot)}(p_z)} (dp) \, \sigma^{M}_{\manq(z)}(dq),
  \end{equation}
and
\begin{equation}
  \label{eq:surfacemeas}
  \sigma_{\manq_{\xi,v_\xi}(z,v_z)}(dq \, dp) =
  \sigma^{M^{-1}}_{\manq_{v_\xi(q,\cdot)}(v_z)}(dp) \, \sigma^{M}_{\manq(z)}(dq).
 \end{equation}
\end{enumerate}
Equations~\eqref{eq:product_decomposition_pz}-\eqref{eq:surfacemeas} are a consequence of the fact that $$\delta_{(\xi(q)-z,p_\xi(q,p)-p_z)}(dq \, dp)= \delta_{p_\xi(q,p)-p_z}(dp) \delta_{\xi(q)-z}(dq)$$ (and a similar relation for $v_\xi$).

\subsubsection{Divergence formulas}
We end this section with an important formula, which is used to
show the invariance of the canonical measure in the proof of Proposition~\ref{p:revcons}.

\begin{proposition}[Divergence theorem in phase space]
  \label{p:divsympl} Consider the Poisson bracket $\poisson{\cdot,\cdot}_{\Xi}$ defined
  by~\eqref{eq:poissonconst}, and an open neighborhood $\mathcal{O}$ of
  $\manq_{\Xi}(\zeta) \subset \R^{6N}$ where $\Gamma$ is invertible. 
  Then for any smooth test functions
  $\ph_1,\ph_2 \, : \,  \R^{6N} \to \R$ with compact support in~$\mathcal{O}$,
  \begin{equation}
    \label{eq:divsympl}
    \int_{\manq_{\Xi}(\zeta)} \poisson{\ph_1,\ph_2}_{\Xi} \, d\sigma_{\manq_{\Xi}(\zeta)} = 0.
  \end{equation}
\end{proposition}
The divergence formula~\eqref{eq:divsympl} can be proved using Darboux's 
theorem and internal coordinates, or directly using the co-area formula 
(see Section~3.3 in \cite{LelRouStoBook}).

We will also need the classical divergence formula on affine spaces (see for instance Section~3.3 in~\cite{LelRouStoBook}): for a fixed $q \in \R^{3N}$,
for any compactly supported smooth vector field $\phi(q,p) \in \R^{3N}$,
\begin{equation}\label{eq:divpaff}
\int_{\manq_{v_\xi(q,\cdot)}(v_z) }  
\op{div}_p\Big( P_M(q)  \phi(q,p) \Big) \sigma^{M^{-1}}_{\manq_{v_\xi(q,\cdot)}(v_z)}(dp)  = 0.
\end{equation}


\section{Constrained Langevin processes and sampling}
\label{sec:sampling}

We first give some properties of the constrained Langevin equation~$\EL$
in Section~\ref{sec:sampling_gen}, then propose some numerical schemes to discretize it
in Section~\ref{sec:consnum}, and finally consider the overdamped limit
in Section~\ref{sec:ovd_limit_constraints}.

\subsection{Properties of the dynamics}\label{sec:sampling_gen}

We consider the dynamics~$\EL$:
\[
\begin{cases}
  \dps  d q_{t} = M^{-1}p_t \, dt, &\\[6pt]
  \dps  d p_{t} = -\nabla V(q_{t}) \, dt -\gamma(q_t) M^{-1}p_t \, dt 
  + \sigma(q_t) \, d W_t + \nabla \xi(q_{t}) \, d \lambda_t, & \\[6pt]
  \dps \xi(q_{t}) = z. & (C_q)
\end{cases}
\]
By differentiating with respect to time the constraint $\xi(q_t)=z$, 
the Lagrange multipliers can be computed explicitly (see for instance Section~3.3 in~\cite{LelRouStoBook}):
\begin{align}
  \dps d\lambda_t &=  -G_{M}^{-1}(q_t) \Big [ 
    \op{Hess}_{q_t}(\xi)\big(M^{-1}p_t,M^{-1}p_t\big) \, dt \nonumber\\ 
    & \quad + \nabla \xi(q_t)^T M^{-1} \Big( -\nabla V(q_t) \, dt
    - \gamma(q_{t}) M^{-1}p_{t} \,dt +\sigma(q_{t}) \,d W_t\Big) \Big ] \nonumber\\
&=\, f_{\rm rgd}^M(q_t,p_t) \, dt + G_{M}^{-1}(q_t)\nabla \xi(q_t)^T M^{-1} \pare{\gamma(q_{t}) M^{-1}p_{t} \,dt - \sigma(q_{t}) \,d W_t },  \label{eq:first_expression_lagrange}
\end{align}
where the constraining force $f_{\rm rgd}^M \in \R^\nc $ is defined as:
\begin{equation}\label{eq:constforce}
 f_{\rm rgd}^M(q,p) = G_M^{-1}(q) \nabla \xi (q)^T M^{-1}\nabla V (q) - G_M^{-1}(q) \op{Hess}_q(\xi)(M^{-1} p, M^{-1} p).
\end{equation}
Thus, using the fact that $P_M(q)^T M^{-1} p = M^{-1}p$ when $p \in T^*_q\manq(z)$, the dynamics~$\EL$ can be recast in a more explicit form as
\begin{equation}
  \label{eq:Langevinconst2}
  \left\{
\begin{aligned}
    d q_{t} &= \dps M^{-1} p_{t}  \, dt, \\[6pt]
 d p_{t} &= -\nabla V(q_{t}) \, dt + \nabla \xi(q_t) f_{\rm rgd}^M(q_t,p_t) \, dt -\gamma_P(q_t) M^{-1} p_{t}\, dt \\
&\quad + \sigma_P(q_t) \, d W_t,
  \end{aligned}
\right.
\end{equation}
where we introduced the notation $(\sigma_P,\gamma_P) := ( P_M \, \sigma , P_M \, \gamma \, P_M^T)$.
The constraint therefore has two effects: (i) the matrices $\gamma,\sigma$ in
the dissipation and fluctuation terms are replaced by their projected counterparts
$\gamma_P,\sigma_P$, and (ii) an orthogonal constraining force $\nabla \xi f_{\rm rgd}^M$
is introduced.

The generator of this stochastic Langevin dynamics is the operator 
$\mathcal{L}_{\Xi}$ which appears in the Kolmogorov evolution equation:
for $(q_t,p_t)$ satisfying~$\EL$ or~\eqref{eq:Langevinconst2}, and 
for any smooth test function $\ph$,
\[
\frac{d}{dt} \E \pare{\ph(q_t,p_t) } = \E \big ( \mathcal{L}_{\Xi}(\ph)(q_t,p_t) \big ).
\]
The expression of $\mathcal{L}_{\Xi}$ can be obtained using It\^o calculus, 
as made precise in the following proposition.
\begin{proposition}
  \label{s:p:Lconstgen}
  Consider either the effective momentum~\eqref{eq:effp} or the effective velocity~\eqref{eq:effv}, denoted with the general constraints $\Xi = (\xi,v_\xi)$ or $\Xi = (\xi,p_\xi)$ (see~\eqref{eq:fullconst}). The solution of the constrained dynamics~$\EL$ (or equivalently~\eqref{eq:Langevinconst2} with an initial condition $(q_0,p_0) \in \Sigma_\Xi(z,0)$) belongs to $\Sigma_\Xi(z,0) =  T^*\manq(z)$, and the generator of this Markov process reads (whatever the value of~$z$)
  \begin{equation}
    \label{eq:conslangevinop}
    \mathcal{L}_{\Xi} = \poisson{\cdot,H}_{\Xi}+ \mathcal{L}_{\Xi}^{\rm thm} ,
  \end{equation}
  where the fluctuation-dissipation part is
  \begin{align*} 
    \mathcal{L}_{\Xi}^{\rm thm} &= \frac{1}{2} \op{div}_{p}\Big(
       \sigma_P \, \sigma_P^{T} \nabla_{p} \cdot \Big) - p^T M^{-1} 
    \gamma_P \nabla_{p},
  \end{align*}
with $(\sigma_P,\gamma_P)$ defined in~\eqref{eq:gamma_P}.
Using the fluctuation-dissipation relation~\eqref{eq:FDR_constraints}, the generator $\mathcal{L}_{\Xi}^{\rm thm}$ can be rewritten more compactly as
\begin{align} \label{eq:conslangevinop_thm}
\mathcal{L}_{\Xi}^{\rm thm} = \frac{1}{\beta} \, \op{e}^{\beta H}
\op{div}_{p}\Big( \rme^{-\beta H} \, \gamma_P \,
  \nabla_{p} \cdot\Big).
\end{align}
\end{proposition}

\begin{proof} 
We perform the computation in two steps: (i) We compute the generator of the Hamiltonian part of the constrained Langevin dynamics, which is~$\EL$ in the case $(\sigma, \gamma) = (0,0)$; (ii) we compute the generator of the ''thermostat'' part of~$\EL$, which is an Ornstein-Uhlenbeck process on momentum variable (corresponding to the second equation in~$\EL$ with $V=0$). 

Let us first consider (i), with $\Xi = (\xi,v_\xi)$ (the case $\Xi = (\xi,p_\xi)$ being similar). Note that
\begin{equation}
  \label{eq:poisson_Xi_H}
  \poisson{\Xi,H}(q,p) = \bmat v_\xi(q,p) \\  
\op{Hess}_q(\xi)(M^{-1} p, M^{-1} p) - \nabla \xi (q) ^T M^{-1}\nabla V (q)  \emat,
\end{equation}
where the Hessian operator $\op{Hess}$ is defined in~\eqref{eq:hessnot}. 
Now, \eqref{eq:Gram_vxi} implies that
\begin{equation}\label{eq:sgraminv}
\sgram^{-1} = \bmat
 G_M^{-1} \, \nabla v_\xi^T J \, \nabla v_\xi \, G_M^{-1} \ & -G_M^{-1} \\
 G_M^{-1} & 0
 \emat.
\end{equation}
Besides, $v_\xi(q_t,p_t)=0$ along a trajectory, since $(q_t,p_t) \in T^*\manq(z)$.
Therefore,
\begin{equation}
  \label{eq:dbcheckcons}
 \forall (q,p) \in T^*\manq(z), \qquad 
 \sgram^{-1} \poisson{\Xi,H}(q,p) = \bmat   f_{\rm rgd}^M(q,p) \\ 0 \emat ,
\end{equation}
where the notation $f_{\rm rgd}^M$ is introduced in~\eqref{eq:constforce}. Consider a test function $\ph:\R^{6N} \to \R$, and remark that 
\begin{equation}
  \label{eq:a}
  \poisson{\ph, \Xi} \bmat a \\ 0 \emat = - a^T \nabla \xi ^T \nabla_p \ph,
\end{equation}
so that, for any $a \in \R^\nc$, 
\begin{equation*}
    \poisson{\ph, \Xi}\sgram^{-1} \poisson{\Xi,H} = - (f_{\rm rgd}^M)^T \nabla \xi^T \nabla_p \ph.
\end{equation*}
Finally, for all $(q,p) \in T^*\manq(z)$,
\begin{equation}
  \label{eq:dbcheckcons2}
\begin{aligned}
    \poisson{\ph, H}_\Xi(q,p) & = -\nabla V(q)^T \nabla_p \ph(q,p) + f_{\rm rgd}^M(q,p) ^T \nabla \xi(q)^T \nabla_p \ph(q,p) \\
& \quad + p^T M^{-1} \nabla_q \ph(q,p) .
\end{aligned}
\end{equation}
The operator~\eqref{eq:dbcheckcons2} is the generator 
of the Hamiltonian part in~\eqref{eq:Langevinconst2}.

We turn to (ii). The diffusive part arises from the fluctuation term $\sigma_P(q_t)\, d W_t$
in~\eqref{eq:Langevinconst2}, and its expression 
\[
\frac{1}{2} \op{div}_{p}\Big(
P_M\, \sigma \sigma^{T}P_M^T \nabla_{p} \, \cdot \, \Big)
\]
is obtained directly from the standard It\^o calculus. Similarly, the dissipation operator 
is 
\[
-\Big( \gamma_P M^{-1} p \Big)^T \nabla_{p} = -p^T M^{-1} P_M \gamma P_M^{T} \nabla_{p}.
\]
The addition of these two contributions gives the expression of 
$\mathcal{L}_{\Xi}^{\rm thm}$. 
\end{proof}

With the expression~\eqref{eq:conslangevinop} of the generator at hand, it is easily checked that the process~$\EL$ satisfies the following equilibrium properties:
\begin{proposition}
  \label{p:revcons}
  When the fluctuation-dissipation relation~\eqref{eq:FDR_constraints} holds,
  the constrained Langevin dynamics~$\EL$ on $T^* \manq(z)$ 
  admits the Boltzmann-Gibbs distribution~\eqref{eq:canonicalconst} as a stationary measure,
  and is reversible up to momentum reversal with respect to~\eqref{eq:canonicalconst}: 
  If ${\rm Law}(q_0,p_0)  =  \mu_{T^* \manq(z)} $, then, for any $T > 0$,
  \[
    {\rm Law}(q_t,p_t; 0 \leq t \leq T) = {\rm Law}(q_{T-t}, -p_{T-t}; 0 \leq t \leq T).
    \]
    Moreover, if $P_M(q) \gamma P_M(q)^T$ is everywhere strictly positive in the sense of symmetric matrices on $T_q^* \manq(z)$, 
    then the process~$\EL$ is ergodic:  
    for any smooth test function~$\ph$,
    \[
    \lim_{T \to + \infty} \frac{1}{T}\int_0^{T} \ph(q_{t},p_{t}) \, dt  
    = \int_{T^* \manq(z)} \ph \, d\mu_{T^* \manq(z)}\qquad \mathrm{a.s.}
    \]
\end{proposition}

\begin{proof}
The stationarity and reversibility properties follow from the following detailed balance condition up to momentum reversal (see for instance Section~2.2 in~\cite{LelRouStoBook}): 
for any test functions~$\ph_1$,~$\ph_2$,
\begin{equation}
\label{eq:balance_proof3.2}
\int_{T^* \manq(z)} \ph_1 \, \mathcal{L}_{\Xi} (\ph_2) \, d\mu_{T^*\manq(z)} =
\int_{T^* \manq(z)} (\ph_2\circ S) \, \mathcal{L}_{\Xi}
(\ph_1\circ S) \, d\mu_{T^* \manq(z)},
\end{equation}
where $S \, : (q,p) \mapsto (q,-p)$ 
is the momentum flip. In view of the expression~\eqref{eq:conslangevinop} of the generator,
proving~\eqref{eq:balance_proof3.2} amounts to proving this property 
for the operators $\poisson{.,H}_{\Xi}$ and $\mathcal{L}_{\Xi}^{\rm thm}$.

For the Hamiltonian part~$\poisson{.,H}_{\Xi}$, the expression~\eqref{eq:dbcheckcons2} yields
\[
\poisson{\ph\circ S,H}_{\Xi}(q,p) = - \poisson{\ph,H}_{\Xi} (q,-p) = - \poisson{\ph,H}_{\Xi} (
S(q,p)),
\]
which states the time symmetry under momentum reversal of the Hamiltonian part of the 
equations of motion~$\EL$. 
On the other hand,
\[
{\rm e}^{-\beta H} \poisson{\cdot,H}_{\Xi}  = -\frac{1}{\beta}
\poisson{\cdot, {\rm e}^{-\beta H}}_{\Xi},
\]
so that
\begin{align*}
  {\rm e}^{-\beta H}\, (\ph_2\circ S) \,  \poisson{\ph_1\circ S,H}_{\Xi} &= -\pare{ {\rm e}^{-\beta H}\, \ph_2\,  \poisson{\ph_1,H}_{\Xi} }\circ S \\
&=\pare{{\rm e}^{-\beta H}\, \ph_1\,  \poisson{\ph_2,H}_{\Xi}  + \poisson{\ph_2\ph_1
, \frac{{\rm e}^{-\beta H}}{\beta}}_{\Xi} }\circ S ,
\end{align*}
and the divergence formula~\eqref{eq:divsympl} 
yields the balance condition~\eqref{eq:balance_proof3.2} for the Hamiltonian part, 
in view of the invariance of the distribution $\sigma_{T^*\manq(z)}$ 
under the momentum flip $S$.

For the thermostat part, it is easily checked, that
\[
\mathcal{L}_{\Xi}^{\rm thm}(\ph \circ S) = 
\mathcal{L}_{\Xi}^{\rm thm}(\ph )\circ S 
\]
for any smooth test function $\ph$, so that the detailed balance condition up to momentum reversal~\eqref{eq:balance_proof3.2} follows from the following more general detailed balance condition, in the case $v_z=0$ ($\mu_{\manq_{\xi,v_\xi}(z,v_z)}$ being defined in~\eqref{eq:mu_v}):
\begin{equation}
  \label{eq:balance_momenta_rev}
  \int_{\manq_{\xi,v_\xi}(z,v_z)} \!\!\! \ph_1 \, \mathcal{L}_{\Xi}^{\rm thm} (\ph_2) \, d\mu_{\manq_{\xi,v_\xi}(z,v_z)} =
\int_{\manq_{\xi,v_\xi}(z,v_z)} \!\!\! \ph_2 \, \mathcal{L}_{\Xi}^{\rm thm}
(\ph_1) \, d\mu_{\manq_{\xi,v_\xi}(z,v_z)}.
\end{equation}
It is interesting to prove~\eqref{eq:balance_momenta_rev} for a general $v_z\in  \R$ since it will be used in the proof of Theorem~\ref{th:crookslangcons} below. Consider the divergence formula~\eqref{eq:divpaff} in the affine space for the variable~$p$ (the position $q$
being fixed), with $$\phi = \gamma P_M^T \nabla_p(\ph_2) \,  {\rm e}^{-\beta H}  \, \ph_1 .$$
After integration in $q$, using the formula~\eqref{eq:conslangevinop_thm} for $\mathcal{L}_{\Xi}^{\rm thm}$ and~\eqref{eq:surfacemeas}, an expression symmetric in $(\ph_1,\ph_2)$ is obtained:
\[
\int_{\manq_{\xi,v_\xi}(z,v_z)} \ph_1 \, \mathcal{L}_{\Xi}^{\rm thm} (\ph_2) \, d\mu_{\manq_{\xi,v_\xi}(z,v_z)} = - \int_{\manq_{\xi,v_\xi}(z,v_z)}\nabla_p^T\ph_1  P_M \gamma  P_M^T \nabla_p\ph_2 \, d\mu_{\manq_{\xi,v_\xi}(z,v_z)} ,
\]
hence the detailed balance condition~\eqref{eq:balance_momenta_rev}.

Ergodicity is a consequence of the hypo-ellipticity of the operator
$\mathcal{L}_{\Xi}$ on $T^\ast \manq(z)$ (H\"ormander's criterion 
is satisfied, see~\cite{Hor67}), which is itself a consequence of the
fact that $P_M(q)\gamma  P_M(q)^T$ is strictly positive on each $T^\ast _q\manq(z)$. 
The proof can be carried out using local coordinates and the results from~\cite{Kli87}.
\end{proof}

\begin{remark}[Infinite stiffness limit]
  \label{rem:highosclang}
  We have considered in this section the Langevin dynamics~\eqref{eq:Langevinconst2} with constraints rigidly imposed by a projection onto the submanifold $T^*\Sigma(z)$. This dynamics samples the canonical distribution~\eqref{eq:canonicalconst}
  $\mu_{T^{*}\manq(z)}(dq \, dp)  =  Z_{z,0}^{-1} \, {\rm e}^{-\beta H(q,p)} \,
  \sigma_{T^{*} \manq(z)}(dq \, dp)$, with constraints on both positions and momenta. The marginal on positions of this distribution is, in view of~\eqref{eq:product_decomposition_pz}, 
  proportional to 
  \[
  \mathrm{e}^{-\beta V(q)}\sigma^{M}_{\manq(z)}(dq). 
  \]
  This is what we call in the following rigidly imposed constraints. The canonical distribution~\eqref{eq:canonicalconst}  with rigid constraints is naturally associated to the rigid free energy~\eqref{eq:Frgd} (and this is what justifies the qualification "rigid"), since, we recall,
  $$ F_{\rm rgd}^M(z) =  -\frac{1}{\beta} \ln \int_{T^\ast \manq(z)} {\rm e}^{-\beta H(q,p)} \sigma_{T^\ast \manq(z)}(dq\, dp) .$$
  
  Another way to impose some constraints on a system is to add a penalization term. In our context, this could be done by changing the potential energy $V$ to
  \[
  V_\varepsilon(q) = V(q) + \frac{1}{\varepsilon} \abs{\xi(q) - z}^2.
  \]
  It is easy to check that, in the limit $\varepsilon \to 0$ (infinite stiffness limit), the canonical measure associated to this potential is the canonical distribution~\eqref{eq:canonical} with positions conditioned by $\xi(q)=z$. This distribution is proportional to $\rme^{-\beta H(q,p) } \, \delta_{\xi(q)-z}(dq) \, dp$ and its marginal on positions is proportional to 
  \[
  \mathrm{e}^{-\beta V(q)} \delta_{\xi(q)-z}(dq). 
  \]
  This is what we call in the following softly imposed constraints. 
  The canonical distribution with soft constraints is naturally associated with the standard free energy, since, we recall,
  $$  F(z) = - \frac{1}{\beta} \ln  \int_{\manq(z) \times \mathbb{R}^{3N}}
  \rme^{-\beta H(q,p) } \, \delta_{\xi(q)-z}(dq) \, dp.$$
  
  Note that, in view of~\eqref{eq:coarea}, the marginal on positions for softly imposed constraints can be written in terms of rigidly imposed constraints through a modification of the potential:
  $${\rm e}^{-\beta V(q)} \delta_{\xi(q)-z}(dq) = {\rm e}^{-\beta (V + V_{{\rm fix}})(q)}\sigma^{M}_{\manq(z)}(dq)$$
  where
  \begin{equation}
    \label{eq:Vfix}
    V_{{\rm fix}}(q) = \frac{1}{2\beta} \ln \Big( \det G_{M}(q) \Big),
  \end{equation}
  is sometimes called the Fixman corrector (see~\cite{Fix78}). Thus, if $(q_t,p_t)$ satisfies~\eqref{eq:Langevinconst2} with the modified potential $V + V_{{\rm fix}}$ then $q_t$ samples (in the longtime limit) the probability measure proportional to ${\rm e}^{-\beta V(q)} \delta_{\xi(q)-z}(dq)$, and we thus refer to this dynamics as the softly constrained Langevin dynamics.
  
  These concepts will be used in Section~\ref{TI-sec:MC} to describe the computation of free energy differences for systems with molecular constraints.

  Finally, let us mention that the infinite stiffness limit $\varepsilon \to 0$ of the Langevin dynamics~\eqref{eq:Langevin_std} with the potential $V_\varepsilon$ is not (except for very specific forms of constraints) the softly constrained Langevin dynamics, as one would expect. We refer to [32] for example, where it is shown that adiabatic effective potentials (derived from the conservation of the ratio of energy over frequency of fast modes) are required to describe the limiting dynamics. However, a formal argument based on ``over-damping'' the fast modes indeed leads to the softly constrained Langevin dynamics. 
  \cqfd\end{remark}

\subsection{Numerical implementation}
\label{sec:consnum}

We consider in this section a numerical scheme based on a splitting of the 
Langevin dynamics~$\EL$ into a Hamiltonian part 
(Section~\ref{sec:numHamiltonconst})
and a fluctuation-dissipation part acting only on the momentum 
(Section~\ref{sec:numflucdissconst}). Such a splitting is standard for unconstrained
systems, but other splitting strategies for the Langevin equation can be considered as well 
(see~\cite{MilsteinTretyakov03,MilsteinTretyakov04}).

For simplicity, we restrict ourselves to constant matrices $\gamma$ 
and $\sigma$. Generalizations to position dependent matrices are straightforward. 

The Hamiltonian part of the Langevin dynamics~$\EL$ (namely~$\EL$ with $(\sigma,\gamma) = (0,0)$) is discretized using a velocity-Verlet scheme with constraints, which yields~\eqref{eq:Verletconst} below. The fluctuation-dissipation part on momentum variable in~$\EL$ is the following Ornstein-Uhlenbeck process (for a fixed given $q\in \manq(z)$):
\begin{equation}
  \label{eq:OU}
  \begin{cases}
     d p_t = - \gamma M^{-1} p_t \, dt + \sigma \, d W_t + \nabla \xi (q) \, d\lambda^{\rm OU}_t, &  \\
 \nabla \xi(q) M^{-1} p_t = 0,  & (C_p)
  \end{cases}
\end{equation}
which can be rewritten as (see~\eqref{eq:Langevinconst2})
\[
d p_t = - \gamma_P(q) M^{-1} p_t \, dt + \sigma_P(q) \, d W_t.
\]
This equation can be explicitly integrated on $[0,t]$ to obtain:
\begin{equation}
  \label{eq:OUexpl}
  p_t = \rme^{- t \, \gamma_P(q)M^{-1}} p_0+ \int_0^t \rme^{-(t-s) \gamma_P(q)M^{-1}} \, \sigma_P(q) \, d W_s .
\end{equation}
However, the matrix exponential $\rme^{- t \, \gamma_P(q)M^{-1}}$ may be difficult to compute
in practice (except for certain choices of $\gamma$ and $M$, see the discussion at the end
of Section~\ref{sec:numflucdissconst}). 
Instead of performing an exact integration, 
\eqref{eq:OU} can be discretized using a midpoint Euler scheme, 
which yields~\eqref{eq:flucdiss1} and~\eqref{eq:flucdiss2} below. 

The numerical scheme we investigate, termed 
midpoint Euler-Verlet-midpoint Euler splitting, is therefore the following:
\begin{align} 
 & \begin{cases}
    \dps p^{n+1/4} = p^{n} -\frac{\dt}{4} \gamma \, M^{-1} (p^n+p^{n+1/4}) 
  + \sqrt{\frac\dt2} \, \sigma \, {\mathcal G}^n \\
  \phantom{p^{n+1/4} =} + \nabla\xi(q^n) \, \lambda^{n+1/4},  &\\[6pt]
  \nabla \xi(q^n)^{T} M^{-1} p^{n+1/4} = 0,  &(C_p)\\
  \end{cases} \label{eq:flucdiss1} \\ 
  &\begin{cases}
     p^{n+1/2} = \dps p^{n+1/4} - \frac{\dt}{2} \nabla V (q^{n}) 
  + \nabla \xi(q^n) \, \lambda^{n+1/2}, &\\[6pt] 
  q^{n+1} =  q^{n} + \dt \, M^{-1} \, p^{n+1/2}, \\[6pt] 
  \xi(q^{n+1}) = z,  &(C_q)\\[6pt]
  p^{n+3/4} = \dps p^{n+1/2} - \frac{\dt}{2} \nabla V (q^{n+1}) 
  + \nabla \xi(q^{n+1}) \, \lambda^{n+3/4}, &\\[6pt] 
  \nabla\xi (q^{n+1})^{T} M^{-1} p^{n+3/4} = 0,  &(C_p)\\ 
  \end{cases} \label{eq:Verletconst} \\ 
  &\begin{cases}
       \dps p^{n+1} =  p^{n+3/4} -\frac{\dt}{4} \gamma \, M^{-1} (p^{n+3/4}+p^{n+1}) 
  + \sqrt{\frac\dt2} \, \sigma \, {\mathcal G}^{n+1/2}  \\
  \phantom{p^{n+1} =} + \nabla\xi(q^{n+1}) \, \lambda^{n+1}, &  \\[6pt]
  \nabla \xi(q^{n+1})^{T} M^{-1} p^{n+1} = 0, & \hspace{-1cm} (C_p)\\
  \end{cases} \label{eq:flucdiss2} 
\end{align}
where $({\mathcal G}^n)_{n \geq 0}$ and $({\mathcal G}^{n+1/2})_{n \geq 0}$ are sequences of independently and identically distributed (i.i.d.) Gaussian random variables of mean $0$ and covariance 
matrix~$\mathrm{Id}_{3N}$.

Note that when $\gamma=0$ and $\sigma=0$, the scheme~\eqref{eq:flucdiss1}-\eqref{eq:Verletconst}-\eqref{eq:flucdiss2} becomes deterministic, and reduces to~\eqref{eq:Verletconst}, which is a scheme for the deterministic Hamiltonian equations of motion with position
constraints $\xi(q) = z$. The latter scheme is referred to as the ''Hamiltonian scheme~\eqref{eq:Verletconst}'' below.

\subsubsection{Comments on the Hamiltonian scheme~\eqref{eq:Verletconst}}
\label{sec:numHamiltonconst}

The Hamiltonian part~\eqref{eq:Verletconst} of the scheme, often
called 'RATTLE' in the literature, is an explicit integrator, and is a modification of the classical 'SHAKE' algorithm (see Chapter~VII.$1$ in \cite{HairerLubichWanner06}, or Chapter~$7$ in \cite{leimkuhler-reich-04} for more precisions and historical references). In~\eqref{eq:Verletconst}, $\lambda^{n+1/2} \in \R^\nc$ are the Lagrange
multipliers associated with the position constraints $(C_{q})$, and
$\lambda^{n+3/4}\in \R^\nc$ are the Lagrange multipliers associated with the
velocity constraints $(C_{p})$. The nonlinear constraints $(C_{q})$
are typically enforced using Newton's algorithm. In~\eqref{eq:Verletconst}, the (linear) momentum projection $(C_p)$ is always well defined since we assumed that the Gram matrix $G_M(q)$ is invertible. On the other hand, the nonlinear projection used to enforce the position constraints $\xi(q^{n+1}) = z$ is in general well defined only on a subset of phase space.

\begin{definition}[Domain $D_{\dt}$]
  The domain $D_{\dt} \subset T^\ast  \manq(z)$ is defined as the set of configurations $(q^{n},p^{n+1/4})\in T^\ast  \manq(z)$ such that there is a unique solution $(q^{n+1},p^{n+3/4})$ verifying~\eqref{eq:Verletconst}.
\end{definition}
Solving the position constraints $(C_q)$ consists in projecting onto $\manq(z)$ a point in a $\dt$-neighborhood of $q^n$. Thus, by the implicit function theorem, the domain $D_{\dt}$ verifies:
\[
\lim_{\dt \to 0} D_{\dt} =  T^\ast \manq(z) .
\]
It may happen that there is no solution if the time-step is too large, 
and, even for small time-steps, that several 
projections exist, see for instance Example~$2$ in Chapter~$7$
of \cite{leimkuhler-reich-04}.  In practice, $D_{\dt}$ can be chosen to be the set of $(q^{n},p^{n+1/4})$ such that the Newton algorithm enforcing the constraints $(C_q)$ has converged within a given precision threshold and a limited number of iterations.
 
As for the Verlet scheme in the unconstrained case, the
associated numerical flow shares two important qualitative properties with the
exact flow: It is time reversible and
symplectic (see~\cite{LeiSke94}). This implies quasi-conservation of energy, in the sense that energy is conserved within a given precision threshold over exponentially long times, see \cite{HairerLubichWanner06,leimkuhler-reich-04}.

\subsubsection{Comments on the fluctuation-dissipation part~\eqref{eq:flucdiss1} and~\eqref{eq:flucdiss2}}\label{sec:numflucdissconst}
The new momentum $p^{n+1/4}\in T^\ast\manq(z)$ in~\eqref{eq:flucdiss1} (or $p^{n+1}$ in~\eqref{eq:flucdiss2}) may be obtained by first integrating the unconstrained dynamics with a midpoint scheme, and then computing the Lagrange multiplier $\lambda^{n+1/4}$ (or $\lambda^{n+1}$) by solving the following linear system implied by the constraints $(C_p)$:
\[
\begin{aligned}
\nabla\xi(q^n)^T M^{-1} \left(\op{Id}+\frac\dt4 \gamma M^{-1} \right)^{-1} & 
\left[ \left(\op{Id}-\frac\dt4 \gamma M^{-1} \right) \, p^n \right. \\
& \left. \quad + \sqrt{\frac{\dt}{2}} \,  \sigma \, {\mathcal G}^n + \nabla \xi(q^n) \, \lambda^{n+1/4} 
\right] = 0 .
\end{aligned}
\]
A sufficient criteria for stability is
\[
\frac{\dt}{4} \gamma \leq M .
\]
Besides, it can be checked (see Sections~2.3.2 and~3.3.5 in~\cite{LelRouStoBook}) that the Markov chain induced by the fluctuation-dissipation part of the scheme~\eqref{eq:flucdiss1} (or~\eqref{eq:flucdiss2}) verifies a detailed balance equation (both in the plain sense and 
up to momentum reversal) with respect to the stationary measure $\kappa_{T^\ast_q \manq(z)}^{M^{-1}}(dp)$. The latter is defined as the kinetic probability distribution
\begin{equation}
  \label{eq:kinetic_part}
  \kappa_{T^\ast_q \manq(z)}^{M^{-1}}(dp) =  \left(\frac{\beta}{2\pi}\right)^{(3N-\nc)/2} \exp\left(-\beta \frac{p^T M^{-1} p}{ 2}\right)
\, \sigma^{M^{-1}}_{T^{\ast}_q \manq(z)}(dp),
\end{equation}
and is the marginal in the $p$-variable of the canonical distribution $\mu_{T^*\manq(z)}(dq \, dp)$ conditioned by a given $q\in \manq(z)$. Moreover, if $\gamma_P  := P_M \, \gamma \, P_M^T$ is strictly positive in the sense of symmetric linear transformations of $T^*_q\manq(z)$, then the Markov chain on momentum variable induced by~\eqref{eq:flucdiss1} (or~\eqref{eq:flucdiss2}) alone is ergodic with respect to~$\kappa_{T^\ast_q \manq(z)}^{M^{-1}}(dp)$. 

Finally, an important simplification occurs in the integration of~\eqref{eq:OU} in the special case when $\gamma$ and $M$ are equal up to a multiplicative constant (so that $\gamma M^{-1}$ is proportional to identity). Indeed in this case the equality $\pare{\gamma_P(q)M^{-1}}^n = P_M(q) \pare{\gamma \, M^{-1}}^n$ holds for any $n \geq 0$, and~\eqref{eq:OUexpl} simplifies to
\begin{equation}
  \label{eq:OUexpl_simpl}
  p_t = P_M(q) \left( \rme^{- t \, \gamma M^{-1}} p_0+  
  \int_0^t \rme^{-(t-s) \gamma M^{-1}} \, \sigma  \, d W_s \right).
\end{equation}
The numerical integration of~\eqref{eq:OU} can thus be carried out in two steps: 
(i) exactly integrating~\eqref{eq:OU} without constraint, 
and then (ii) projecting the result onto~$T^\ast_q \manq(z)$.

\subsubsection{Metropolis-Hastings correction}
\label{sec:nummetroconst}

Usually, the invariant probability distribution sampled by the solution of a numerical scheme is biased by the time discretization. Relying on (i) the time symmetry (up to momentum reversal)
and (ii) the preservation 
of the phase space measure $\sigma_{T^*\manq(z)}(dq \, dp)$ by the solution of the RATTLE scheme~\eqref{eq:Verletconst}, it
is possible to eliminate the time discretization error in the splitting scheme
\eqref{eq:flucdiss1}-\eqref{eq:Verletconst}-\eqref{eq:flucdiss2} by resorting to a Generalized Hybrid Monte Carlo algorithm.

\begin{algorithm}[GHMC with constraints]
  \label{a:GHMCconst}
  Consider an initial configuration $(q^0,p^0) \in T^*\manq(z)$,
  and a sequence $({\mathcal G}^n,{\mathcal G}^{n+1/2})_{n\geq 0}$ of independently and identically distributed 
  standard Gaussian vectors. Iterate on $n \geq 0$:  
  \begin{enumerate}
  \item Evolve the momentum according to the midpoint Euler scheme~\eqref{eq:flucdiss1}, and compute the energy $E^n = H(q^n,p^{n+1/4})$ of the new configuration;
  \item Integrate the Hamiltonian part according to the RATTLE
    scheme~\eqref{eq:Verletconst}, denote $(\widetilde{q}^{n+1},\widetilde{p}^{n+3/4})$ the resulting state, and set $E^{n+1} = H(\widetilde{q}^{n+1},\widetilde{p}^{n+3/4})$.
  \item Accept the proposal
    $(q^{n+1},p^{n+3/4}):=(\widetilde{q}^{n+1},\widetilde{p}^{n+3/4})$ with probability
    \[
    \min \Big ( \rme^{ -\beta (E^{n+1}-E^n) },1 \Big ).
    \]
    Otherwise, reject and flip the momentum:
    $(q^{n+1},p^{n+3/4}) = (q^n,-p^{n+1/4})$.
  \item Evolve the momentum according to the midpoint Euler scheme~\eqref{eq:flucdiss2}.
  \end{enumerate}
\end{algorithm}
By construction, the GHMC algorithm with constraints leaves
invariant the equilibrium distribution $\mu_{T^*\manq(z)}(dq\,dp)$ (see Section~3.3.5 in \cite{LelRouStoBook}).

To understand the momentum reversal required upon rejection, 
it is useful to write more explicitly the Markov
chain as the composition of a Metropolis-Hastings part, where the proposal is obtained 
by a RATTLE step followed by a momentum reversal (the latter operation is needed to ensure
the symmetry of the proposition), which is then accepted or rejected;
and another momentum reversal (which leaves invariant the targeted probability measure $\mu_{T^*\manq(z)}(dq\,dp)$). When the proposal is accepted, the two momentum reversals
cancel out each other. On the other hand, when the proposal is rejected, momenta are
actually reversed. See~\cite[Section~2.1.4]{LelRouStoBook} for more background 
on generalized Metropolis-Hastings algorithms.

In the above, we implicitly assume that the RATTLE scheme~\eqref{eq:Verletconst} is everywhere well defined.
 In practice however, it is necessary to modify Algorithm~\ref{a:GHMCconst} by restricting the sampled configurations to $D_{\dt}$. This can be achieved by introducing additional tests in steps (1), (2) and (4), and rejecting the states that have gone outside the set $D_{\dt} \subset T^\ast \manq(z)$ where the position constraint $(C_q)$ is well defined. By doing so, the global algorithm has an invariant equilibrium distribution given by $\mu_{T^*\manq(z)}(dq\,dp)$
conditioned on the set of states $D_{\dt}$. This invariant distribution can be written explicitly as follows:
\begin{equation}\label{eq:mu_with_cutoff}
\frac{1}{Z_{z,0,\dt}} {\rm e}^{-\beta H(q,p)} \one_{(q,p)\in D_{\dt} }\sigma_{T^*\manq(z)}(dq \, dp).
\end{equation}
Alternatively, the rejection tests in steps (1), (2) and~(4) of Algorithm~\ref{a:GHMCconst} can be performed with a cut-off parameter $R_{\dt}>0$ on the momentum variable,
chosen so that the position constraint $(C_q)$ in~\eqref{eq:Verletconst} is everywhere well defined
when $\frac{1}{2} p^T M^{-1} p  \leq R_{\dt}$. This can be achieved when there exists $R_{\dt}>0$ small enough so that $\manq(z) \times \{ \frac{1}{2} p^T M^{-1} p \leq R_{\dt} \} \subset D_{\dt} \subset T^\ast \manq(z)$. Since this is useful for later purposes (see the discussion at the end of Section~\ref{sec:num_ti}), we provide a rough estimate of $R_{\dt}$ in terms of $\dt$, assuming
for simplicity that~$\manq(z)$ is compact. First, by the implicit function theorem, there exists
$\alpha > 0$ such that, for all $q \in \manq(z)$ and $\delta q$ with norm 
$\| \delta q \| < \alpha$, there is a unique $\lambda \in \R^\nc$ satisfying
\[
\xi\left(q + M^{-1} (\delta q + \nabla \xi(q) \lambda)\right) = z.
\] 
Therefore, there exists $a > 0$ small enough such that, when $\|p^{n+1/4}\| \leq a/\dt$,
the RATTLE scheme in~\eqref{eq:Verletconst} is well defined, namely
there exists a unique $q^{n+1}$ satisfying the constraint $(C_q)$.
This shows that 
\begin{equation}
  \label{eq:estimate_Rdt}
  R_{\dt} \geq A\dt^{-2}
\end{equation}
for some $A > 0$.

The invariant probability distribution of the Markov chain generated by
GHMC with the additional rejection steps ensuring $\frac{1}{2} p^T M^{-1} p \leq R_{\dt}$,
is given by~\eqref{eq:mu_with_cutoff}, and actually reads
\[
\frac{1}{Z_{z,0,\dt}} {\rm e}^{-\beta H(q,p)} \one_{ \frac{1}{2} p^T M^{-1} p \leq R_{\dt}}
\sigma_{T^*\manq(z)}(dq \, dp).
\]
Its marginal distribution in the position variable is then exactly given by:
\begin{equation*}
\frac{1}{Z_{z}} {\rm e}^{-\beta V(q)}  \sigma^M_{\manq(z)}(dq).
\end{equation*}
This is also the marginal distribution in the position variable of~$\mu_{T^*\Sigma(z)}$.
Note however, that if $R_{\dt}$ is too small, only small momenta will be sampled in step~(1) of Algorithm~\ref{a:GHMCconst}, and the correlation time of the sampling will be large. In practice, the threshold $R_{\dt}$ should be tuned in preliminary computations so that:
(i)~$R_{\dt}$ is small enough so that the maximal number $N_{\rm max}$ of iterations for the Newton algorithm used to enforce $(C_q)$ in~\eqref{eq:Verletconst} is never reached;
(ii)~$R_{\dt}$ is large enough so that the correlation time of the sampling is as small as possible.

Let us end this section with a warning: It is now known that the correction of the bias 
in discretizations of the Langevin dynamics by a Metropolization of the scheme may reduce
the efficiency of the sampling, see for instance~\cite{AR08}.

\subsection{Exact sampling on a submanifold with overdamped dynamics}
\label{sec:ovd_limit_constraints}

Constrained overdamped Langevin processes (or Brownian dynamics) are solutions of the 
stochastic differential equation (see also \cite{LelRouStoBook,ciccotti-lelievre-vanden-einjden-08})
\begin{equation}
  \label{eq:overdampconst}
  \left \{ \begin{aligned}
    d q_t & = -\nabla V(q_t) \, dt + \sqrt{\frac{2}{\beta}} \, d W_t +  \nabla \xi (q_t) \, 
    d \lambda_t , \\
    \xi(q_t) & = z,
  \end{aligned} \right.
\end{equation}
where $\lambda_t$ is an adapted stochastic process. Equivalently, \eqref{eq:overdampconst} can be rewritten 
in the Stratonovitch form as
\[
d q_t= - P(q_t)  \nabla V(q_t) \, dt + \sqrt{\frac{2}{\beta}} P(q_t) \, \circ d W_t,
\]
where $\circ$ denotes the Stratonovitch integration, 
and $P$ is the projector defined by~\eqref{eq:proj} with the choice $M=\op{Id}$:
\begin{equation}
\label{eq:G_M=Id}
P(q) = \op{Id} -  \nabla \xi(q) \, G^{-1}(q) \nabla \xi(q)^T,
\qquad
G(q) =  \nabla \xi(q)^T \nabla \xi(q) .
\end{equation}
It can be shown that~\eqref{eq:overdampconst} verifies the detailed balance condition for (and is ergodic with respect to) the invariant distribution
\begin{equation}
  \label{eq:over_inv}
  Z_z^{-1} \, {\rm e}^{-\beta V(q)}\sigma^{\op{Id}}_{\manq(z)}(dq),
\end{equation}
which is the marginal in the $q$-variable of the canonical distribution with 
constraints~\eqref{eq:canonicalconst} for the choice $M=\mathrm{Id}$.
It is easy to generalize all our results to scalar products associated with a general symmetric definite positive matrix $M$ upon considering~\eqref{eq:overdampM}, see Remark~\ref{rmk:discrete_ovd_M} below.

The constrained overdamped 
Langevin process~\eqref{eq:overdampconst} may be obtained from a scaling limit 
of the constrained Langevin dynamics~$\EL$ (in the limit when either the mass goes to zero, or the damping~$\gamma$ goes to infinity), see Propositions~$2.14$ and~$2.15$ in~\cite{LelRouStoBook}. 

Likewise, at the discrete level, an Euler-Maruyama discretization of the overdamped process~\eqref{eq:overdampconst} can be obtained as a particular case of the numerical discretization~\eqref{eq:flucdiss1}-\eqref{eq:Verletconst}-\eqref{eq:flucdiss2} for the Langevin equation~$\EL$, yielding a Markov chain $(q^n)_{n \ge 0}$ on positions. 
The condition that the mass goes to~0 is replaced by the condition that the mass is proportional
to the time-step. This is the content of the following proposition.

\begin{proposition}
\label{p:langtooverd}
Suppose that the following relation is satisfied:
\begin{equation}
  \label{eq:scaling_num_ovd_const}
  \frac{\dt}{4} \gamma = M = \frac{\dt}{2} \mathrm{Id}. 
\end{equation}
With a slight abuse of notation, the mass matrix and the friction matrix 
are rewritten as $M \, \mathrm{Id}$ and $\gamma \, \mathrm{Id}$ with $M,\gamma \in \mathbb{R}$.
Then the splitting scheme~\eqref{eq:flucdiss1}-\eqref{eq:Verletconst}-\eqref{eq:flucdiss2} yields the following Euler scheme for the overdamped Langevin constrained dynamics~\eqref{eq:overdampconst}:
\begin{equation}
  \label{eq:Eulerconst}
  \begin{cases}
      \dps q^{n+1} = q^{n} - \dt   \nabla V (q^{n}) 
    + \sqrt{\frac{2\dt }{\beta}} \, {\mathcal G}^n + \nabla \xi (q^n) \, \lambda^{n+1}_{\rm od},  \\[6pt]
    \xi(q^{n+1}) = z,
  \end{cases}
\end{equation}
where $({\mathcal G}^{n})_{n\geq 0}$ are independent and identically distributed 
centered and normalized Gaussian
variables, and $(\lambda^n_{\rm od})_{n \geq 1}$ are the Lagrange multipliers
associated with the constraints $(\xi(q^{n}) = z)_{n \geq 1}$. 
Moreover, the Lagrange multipliers in \eqref{eq:Verletconst} verify:
\begin{eqnarray}
2 \lambda^{n+1/2} &=& G^{-1}(q^n) \Big( \nabla \xi (q^n)^T \pare{q^{n+1}-q^{n}} + \dt  \nabla \xi (q^n)^T \nabla V (q^n) \Big) \label{eq:lag_od_1} \\
 &=&  \lambda^{n+1}_{\rm od} + \sqrt{\frac{2\dt }{\beta}}  G^{-1}(q^n)  \nabla \xi (q^n) ^T {\mathcal G}^n,   \label{eq:lag_od_2} 
\end{eqnarray}
as well as
\begin{equation}
2 \lambda^{n+3/4} = G^{-1}(q^{n+1}) \Big( \nabla \xi (q^{n+1})^T \pare{q^{n}-q^{n+1}} + \dt  \nabla \xi (q^{n+1})^T \nabla V (q^{n+1}) \Big). \label{eq:lag_od_3}
 \end{equation}
where $G$ is defined in~\eqref{eq:G_M=Id}.  
\end{proposition}

\begin{proof}
Irrespective of $p^n$, the choice~\eqref{eq:scaling_num_ovd_const} in the scheme
\eqref{eq:flucdiss1}-\eqref{eq:Verletconst}-\eqref{eq:flucdiss2} leads to
\[
p^{n+1/4} = \sqrt{\frac\dt8} \, \sigma \, {\mathcal G}^n + \frac12 \nabla\xi(q^n) \, \lambda^{n+1/4},
\]
where $\lambda^{n+1/4}$ is associated with the constraints $\nabla \xi (q^n) ^T p^{n+1/4} =0$. This gives 
\[
p^{n+1/2} = -\frac{\dt}{2} \nabla V(q^n) + 
\sqrt{\frac{\dt}{8}} \, \sigma \, {\mathcal G}^n + \nabla\xi(q^n) \left( \frac12 \, 
\lambda^{n+1/4} + \lambda^{n+1/2} \right),
\]
where $\lambda^{n+1/2}$ is such that $\xi(q^{n+1})=z$.
The fluctuation-dissipation relation~\eqref{eq:FDR_constraints} can be reformulated 
in this context as
\[
\sigma \sigma^T= \frac{2}{\beta} \gamma =\frac{4}{\beta} \ \op{Id},
\]
and the scheme~\eqref{eq:Eulerconst} is recovered by taking the associated Lagrange multiplier equal to $\lambda^{n+1}_{\rm od} = \lambda^{n+1/4} + 2\lambda^{n+1/2}$. 
Finally, remarking that $G_{M} = \frac{2 }{\dt} G$ and computing explicitly $\lambda^{n+1/2}$ and $\lambda^{n+3/4}$ in \eqref{eq:Verletconst} yields \eqref{eq:lag_od_1}-\eqref{eq:lag_od_2}-\eqref{eq:lag_od_3}.
\end{proof}

This point of view allows to construct a Metropolis correction
to the Euler scheme~\eqref{eq:Eulerconst}, using the
Generalized Hybrid Monte Carlo scheme (Algorithm~\ref{a:GHMCconst}) with the 
time-step chosen according to~\eqref{eq:scaling_num_ovd_const}.
In this way, assuming that the position constraint $(C_q)$ in~\eqref{eq:Verletconst} is everywhere well defined, we obtain a Markov chain $(q^{n})_{n\geq 0}$ discretizing the overdamped dynamics~\eqref{eq:Eulerconst} which \emph{exactly} samples the 
invariant distribution~\eqref{eq:over_inv}. 
Deriving such a Metropolis-Hastings correction to the Euler scheme~\eqref{eq:Eulerconst} without
resorting to phase-space dynamics does not seem to be natural.

\begin{remark}[Discrete overdamped limit]
\label{rmk:discrete_ovd_M}
Proposition~\ref{p:langtooverd} can be seen as a discrete version of the zero-mass limit of the Langevin dynamics. It is also possible to obtain a discrete version of the overdamped limit ($\gamma \to \infty$) of the Langevin dynamics by assuming that the parameters satisfy the relation
\begin{equation*}
\frac{\dt}{4} \gamma = M \propto \mathrm{Id},
\end{equation*}
which is less restrictive than~\eqref{eq:scaling_num_ovd_const}. Equation~\eqref{eq:Eulerconst} is then obtained with $\dt$ replaced by
\begin{equation}
\label{eq:dteff}
\Delta s = \frac{\dt^{2}}{2M} = \frac{2\dt}{\gamma}.
\end{equation}
In this case, the effective discretization time-step $\Delta s$ for the overdamped Langevin dynamics is thus different from the time-step $\Delta t$ originally used for the discretization of the Langevin dynamics. This is reminiscent of the fact that the overdamped limit (at the continuous level) of the Langevin dynamics requires a change of timescale to obtain the overdamped Langevin dynamics.

Note also that in the more general case
$$\frac{\dt}{4} \gamma = M$$
where $M$ is not supposed to be proportional to the identity, the following numerical scheme
is obtained:
\begin{equation*}
\begin{cases}
\dps q^{n+1} = q^{n} - \widetilde{\Delta s} \, M^{-1} \nabla V (q^{n})
+ \sqrt{\frac{2 \widetilde{\Delta s} }{\beta}} \, M^{-1/2} {\mathcal G}^n + M^{-1}\nabla \xi (q^n) \, \lambda^{n+1}_{\rm od}, \\[6pt]
\xi(q^{n+1}) = z,
\end{cases}
\end{equation*}
where $\widetilde{\Delta s}=\dt^2/2$. This is a discretization of the overdamped dynamics
\begin{equation}
\label{eq:overdampM}
\left \{ \begin{aligned}
d q_s & = -M^{-1} \nabla V(q_s) \, ds + \sqrt{\frac{2}{\beta}} \, M^{-1/2} d W_s + M^{-1} \nabla \xi (q_s) \, d \lambda_s , \\
\xi(q_s) & = z,
\end{aligned} \right.
\end{equation}
which is a generalization of~\eqref{eq:overdampconst} to a scalar product on $\Sigma(z)$ induced by a general positive definite mass matrix~$M$. 
\cqfd\end{remark}


\section{Thermodynamic integration with constrained Langevin dynamics}\label{sec:Langti}

In this section, we focus on the computation of the gradient of the rigid free 
energy~\eqref{eq:Frgd}
\[
F_{\rm rgd}^M(z) =  -\frac{1}{\beta} \ln \int_{T^\ast \manq(z)} {\rm e}^{-\beta H(q,p)} 
\sigma_{T^\ast \manq(z)}(dq\, dp),
\]
using a numerical discretization of the constrained Langevin process~$\EL$. 

As explained in the introduction, we may indeed concentrate on the computation of the rigid free energy~\eqref{eq:Frgd}, since the standard free energy~\eqref{eq:Fbis} can be computed from the latter one using~\eqref{eq:FFtilde}. The relation~\eqref{eq:FFtilde} can be proved with the co-area formula~\eqref{eq:coarea}. Indeed, the free energy defined in~\eqref{eq:Fbis} can be rewritten as
(where C denotes a constant which may vary from line to line):
\begin{align}
  F(z) &= -\frac{1}{\beta} \ln \int_{\manq(z)\times \R^{3N}} {\rm e}^{-\beta H(q,p)} \delta_{\xi(q)-z}(dq) \, dp \\
  &=-\frac{1}{\beta} \ln \int_{\manq(z)} {\rm e}^{-\beta V(q)} (\det  G_M(q) )^{-1/2} \sigma_{\manq(z)}^{M}(dq) + {\rm C} \nonumber \\
  & = -\frac{1}{\beta} \ln \int_{T^\ast \manq(z)} {\rm e}^{-\beta H(q,p)} (\det  G_M(q) )^{-1/2} \sigma_{T^\ast \manq(z)}(dq \, dp) + {\rm C} \nonumber 
\end{align}
hence 
\begin{equation}
  \label{eq:Fqp_final_difference}
  F(z) = F_{\rm rgd}^M(z) -\frac{1}{\beta} \ln \int_{T^{*}\manq(z)}   
  (\det G_M) ^{-1/2} d\mu_{T^{*}\manq(z)} + {\rm C},
\end{equation}
where surface measures are defined in Section~\ref{sec:measures}. 
Note that the rigid free energy $F_{\rm rgd}^M$
indeed depends explicitly on the mass matrix since
\begin{equation}
\label{eq:explicit_dependence_M}
F_{\rm rgd}^M(z) = -\frac{1}{\beta} \ln \int_{\manq(z)} {\rm e}^{-\beta V(q)} \sigma_{\manq(z)}^M(dq) 
+ \mathrm{C}.
\end{equation}

This section is organized as follows. First, we show how systems with 
molecular constraints and systems with constrained values of the reaction coordinate
can be treated in a unified framework (Section~\ref{TI-sec:MC}). We then relate the Lagrange
multipliers arising in the constrained Langevin dynamics, and the gradient of the rigid free
energy (the so-called mean force) in Section~\ref{sec:lag}.
We consider the numerical computation of the mean force in Section~\ref{sec:num_ti},
where we prove consistency results for the corresponding approximation formulas.
Finally, some numerical results on a model system illustrate 
the approach in Section~\ref{sec:simple_example_WCA}.

\subsection{Molecular constraints}
\label{TI-sec:MC}

We discuss here how to generalize all the computations to systems with molecular constraints, generalizing thereby some results of~\cite{ciccotti-kapral-vanden-eijnden-05}. 
Without loss of generality, we consider rigidly imposed molecular constraints,
see Remark~\ref{rmk:choice_mc} below for a discussion of this assumption.
This section can be considered as independent of the remainder of the paper and may therefore be omitted in a first reading.

In practice, many systems are subject to molecular constraints, such as fixed lengths for covalent bonds, or fixed angles between covalent bonds. The reader is referred to \cite{Rap} for practical aspects related to the simulation of molecular constraints. In the context of free energy computations, two types of constraints are therefore considered: first, the molecular constraints, 
$$
\xi_{\rm mc}(q)=(\xi_{{\rm mc},1}(q),\ldots, \xi_{{\rm mc},\overline{m}} (q) )=0,
$$
for $\overline{m} < 3N$, and second, the reaction coordinates denoted in this section by $\xi_{\rm rc} \, : \, \R^{3N} \to \R^{\nc}$, with $\overline{m} +\nc < 3N$. The submanifold of molecular constraints is denoted by
$$
\Sigma_{{\rm mc}} = \setbig{q \in \R^{3N} \ \big | \ \xi_{\rm mc}(q) =0},
$$
and the submanifold associated with the reaction coordinates by
$$\Sigma_{\rm rc}(z_{\rm rc})=\setbig{q \in \R^{3N} \, \big | \, \xi_{\rm rc}(q) =z_{\rm rc}}.$$
It is assumed that the full Gram matrix:
$$
G_M^{{\rm mc},{\rm rc}} := \nabla (\xi_{\rm mc},\xi_{\rm rc})^T M^{-1} \nabla  (\xi_{\rm mc},\xi_{\rm rc}) \in \R^{(\overline{m}+\nc)\times (\overline{m}+\nc)}
$$
is everywhere invertible on $\Sigma_{\rm mc} \cap \Sigma_{\rm rc}(z_{\rm rc})$. Likewise, we denote
$$
G_M^{\rm rc} := \nabla \xi_{\rm rc}^T M^{-1} \nabla  \xi_{\rm rc} \in \R^{\nc \times \nc},
$$
and
$$
G_M^{\rm mc} := \nabla \xi_{\rm mc}^T M^{-1} \nabla  \xi_{\rm mc} \in \R^{\overline{m} \times \overline{m}} .
$$
Assuming rigid mechanical constraints on the molecular constraints $\xi_{\rm mc}$ (see Remark~\ref{rmk:choice_mc} below), we are led to considering the canonical distribution
\begin{align}
  \mu_{T^{*}\Sigma_{\rm mc}} (dq \, dp) & = \frac{1}{Z_{\rm mc}}  {\rm e}^{-\beta H(q,p)} \delta_{\xi_{\rm mc} (q),p_{\xi_{\rm mc}}(p,q)}(dq \, dp) \label{eq:gibbsmc} \\
  & =\frac{1}{Z_{\rm mc}}  {\rm e}^{-\beta H(q,p)} \sigma_{T^{*}\Sigma_{\rm mc}}(dq \, dp), \nonumber
\end{align}
to describe systems with molecular constraints at a fixed temperature. The measure~$\sigma_{T^{*}\Sigma_{\rm mc}}$ denotes the phase space measure on $T^*\Sigma_{\rm mc}$, equal by~\eqref{eq:deltasigmap} to the conditional measure $ \delta_{\xi_{\rm mc} (q),p_{\xi_{\rm mc}}(q,p)}(dq \, dp )$ associated with the constraints $(\xi_{\rm mc} (q) = 0 ,p_{\xi_{\rm mc}}(q,p)=0)$, where $p_{\xi_{\rm mc}}$ is the effective momentum~\eqref{eq:effp} associated with $\xi_{\rm mc}$. 

\begin{remark}[On the choice of the distribution~\eqref{eq:gibbsmc}]
  \label{rmk:choice_mc}
The distribution $\mu_{T^{*}\Sigma_{\rm mc}}$ in~\eqref{eq:gibbsmc} is obtained by constraining rigidly $\xi_{\rm mc}(q)$ to~0, and not 'softly' (in which case $\delta_{\xi_{\rm mc} (q),p_{\xi_{\rm mc}}(p,q)}(dq \, dp) $ would be replaced by $\delta_{\xi_{\rm mc} (q)}(dq) \, dp$, see also Remark~\ref{rem:highosclang}). As explained in Section~\ref{sec:sampling}, the distribution~\eqref{eq:gibbsmc} is the equilibrium distribution of a Langevin process (thermostated Hamiltonian dynamics) with rigid position constraints $\xi_{\rm mc} (q)=0$. Choosing whether constraints should be soft or rigid is a modeling choice, and there is no clear consensus on this issue in the current literature.

Note however that it is possible to rewrite the remainder of this section by considering the softly constrained potential rather than the rigidly constrained potential, up to an appropriate modification of~\eqref{eq:FFtildemc} below, with the help of some Fixman corrective potential.
\cqfd\end{remark}
By associativity of the conditioning of measures, the distribution $\mu_{T^{*}\Sigma_{\rm mc}}$ conditioned by a value of the reaction coordinates $\xi_{\rm rc}(q)=z_{\rm rc}$ is given, up to a normalizing factor, by:
$$
{\rm e}^{-\beta H(q,p)} \delta_{\xi_{\rm mc} (q),p_{\xi_{\rm mc}}(q,p),\xi_{\rm rc}(q)-z_{\rm rc}}(dq\,dp).
$$
Therefore, considering the marginal probability distribution of the reaction coordinates $\xi_{\rm rc}(q)$ leads to the following definition of the free energy associated with~$\xi_{\rm rc}$:
$$
F^{\rm mc}(z_{\rm rc}) = -\frac{1}{\beta} \ln \int_{T^\ast \manq_{\rm mc} \cap (\manq_{\rm rc}(z_{\rm rc}) \times \mathbb{R}^{3N})} {\rm e}^{-\beta H(q,p)} \delta_{\xi_{\rm mc} (q),p_{\xi_{\rm mc}}(q,p),\xi_{\rm rc}(q)-z_{\rm rc}}(dq\,dp) .
$$
The conditional distribution can be decomposed as follows, using the co-area formulas~\eqref{eq:coarea}-\eqref{eq:coarea_effective_momentum} and the definition of effective momentum~\eqref{eq:effp}:
\begin{align*}
\begin{split}
& \delta_{\xi_{\rm mc} (q),p_{\xi_{\rm mc}}(q,p),\xi_{\rm rc}(q)-z_{\rm rc}}(dq\,dp)  \\
&\qquad = \delta_{p_{\xi_{\rm mc}}(q,p)}(dp)\delta_{\xi_{\rm mc} (q),\xi_{\rm rc}(q)-z_{\rm rc}}(dq) 
\end{split} \\
&\qquad = \pare{ \det G_M^{\rm mc}(q) }^{1/2}\sigma^{M^{-1}}_{T^{*}_q\Sigma_{\rm mc}}(dp)\,   \pare{ \det G_M^{{\rm mc},{ \rm rc}}(q) }^{-1/2} \sigma^{M}_{\manq_{\rm rc}(z_{\rm rc})\cap \Sigma_{\rm mc}}(dq).
\end{align*}
Integrating out the momentum in the linear space $T^{\ast }_q\Sigma_{\rm mc}$ with scalar product $\langle p_1,p_2 \rangle_{M^{-1}}=p_1^T M^{-1} p_2$, the free energy can be rewritten as:
$$
F^{\rm mc}(z_{\rm rc}) = -\frac{1}{\beta} \ln \int_{{\manq_{\rm rc}(z_{\rm rc}) \cap \Sigma_{\rm mc}}} {\rm e}^{-\beta V(q)} \pare{ \frac{  \det G_M^{\rm mc} (q) }{ \det G_M^{{\rm mc}, { \rm rc} }(q)}  }^{1/2} \sigma^{M}_{\manq_{\rm rc}(z_{\rm rc}) \cap \Sigma_{{\rm mc}}}(dq) + {\rm C} .
$$
As a consequence, the free energy $F^{\rm mc}$ can be computed from the generalized rigid free energy:
\begin{align}\label{eq:Frgdmc}
  F_{\rm rgd}^{{\rm mc},M}(z_{\rm rc},0) &= -\frac{1}{\beta} \ln \int_{T^\ast\! \pare{\manq_{\rm rc}(z_{\rm rc})\cap \Sigma_{{\rm mc}} }} {\rm e}^{-\beta H(q,p)} \sigma_{T^{*}(\manq_{\rm rc}(z_{\rm rc})\cap \Sigma_{{\rm mc}})}(dp\,dq), 
\end{align}
using the following formula, similar to~\eqref{eq:FFtilde}:
\begin{equation}
  \label{eq:FFtildemc}
  \begin{aligned}
    & F^{\rm mc}(z_{\rm rc}) - F_{\rm rgd}^{{\rm mc},M}(z_{\rm rc},0) \\
    & \quad = -\frac{1}{\beta} \ln \int_{T^{*}\pare{\manq_{\rm rc}(z_{\rm rc})\cap \Sigma_{{\rm mc}} }} \! \!  \frac{ \pare{ \det G_M^{\rm mc} }^{1/2}}{ \pare{\det G_M^{{\rm mc},{\rm rc}} }^{1/2}}  d\mu_{T^{*}(\manq_{\rm rc}(z_{\rm rc})\cap \Sigma_{{\rm mc}} )} + {\rm C}.
  \end{aligned}
\end{equation}
In the above, $\mu_{T^{*}(\manq_{\rm rc}(z_{\rm rc})\cap \Sigma_{{\rm mc}} )}$ is defined similarly to~\eqref{eq:canonicalconst}. The case of molecular constraints can therefore be treated within the general framework considered in this paper, the sampling of the canonical measure $\mu_{T^{*}(\manq_{\rm rc}(z_{\rm rc})\cap \Sigma_{{\rm mc}} )}$ and the computation of the rigid free energy~\eqref{eq:Frgdmc} being the problems at hand.

\begin{remark}[The overdamped Langevin case]
For systems with molecular constraints, the measure sampled by 
the overdamped Langevin dynamics~\eqref{eq:overdampconst} with $\xi=(\xi_{\rm mc},\xi_{\rm rc})$
is the marginal distribution in the position
variables of $\mu_{T^{*}(\manq_{\rm rc}(z_{\rm rc})\cap \Sigma_{{\rm
mc}} )}$ in the special case $M=\mathrm{Id}$, namely
\[
Z_{z}^{-1}  {\rm e}^{-\beta V(q)} \, 
\sigma^{\mathrm{Id}}_{\manq_{\rm rc}(z_{\rm rc})\cap \Sigma_{{\rm mc}}}(dq).
\]
The rigid free energy which is thus naturally computed with such a dynamics is
\[
F_{\rm rgd}^{{\rm mc},\mathrm{Id}}(z_{\rm rc},0) = -\frac{1}{\beta} \ln
\int_{\manq_{\rm rc}(z_{\rm rc})\cap \Sigma_{{\rm mc}} } {\rm
  e}^{-\beta V(q)} \sigma^{\mathrm{Id}}_{\manq_{\rm rc}(z_{\rm rc})\cap
\Sigma_{{\rm mc}}}(dq),
\]
and the actual free energy is thus recovered by a formula similar to~\eqref{eq:FFtildemc}, namely
\[
\begin{aligned}
  & F^{\rm mc}(z_{\rm rc}) - F_{\rm rgd}^{{\rm mc},\mathrm{Id}}(z_{\rm rc},0) \\
  & \qquad = -\frac{1}{\beta} \ln \int_{\manq_{\rm rc}(z_{\rm rc})\cap \Sigma_{{\rm mc}} } 
  \frac{ \pare{ \det G_M^{\rm mc}}^{1/2}}{ \pare{\det G_{\mathrm{Id}}^{{\rm mc},{\rm rc}} }^{1/2}}
  \, {\rm e}^{-\beta V} \, d\sigma^{\mathrm{Id}}_{\manq_{\rm rc}(z_{\rm rc})\cap \Sigma_{{\rm mc}})}
  + {\rm C}.
\end{aligned}
\]
\end{remark}

\subsection{The mean force and the Lagrange multipliers}
\label{sec:lag}

In this section, the average of the constraining force~\eqref{eq:constforce} is related to the gradient of the rigid free energy~\eqref{eq:Frgd} (or mean force). We also give a similar result
for the following generalized rigid free energy:
\begin{align}\label{eq:F_xi}
F^{\Xi}_{\rm rgd}(\zeta) =  -\frac{1}{\beta} \ln \int_{\manq_{\Xi}(\zeta)} {\rm e}^{-\beta H(q,p)} \sigma_{\manq_{\Xi}(\zeta)}(dq\, dp).
\end{align}

\begin{proposition}\label{p:meanforce}
The constraining force $f_{\rm rgd}^M: T^*\manq(z) \to \R^{\nc}$ defined in~\eqref{eq:constforce} as
\[
f_{\rm rgd}^M(q,p) = G_M^{-1}(q) \nabla \xi (q)^T M^{-1}\nabla V (q) - G_M^{-1}(q) \op{Hess}_q(\xi)(M^{-1} p, M^{-1} p),
\]
yields on average the rigid free energy derivative:
\begin{equation}
  \label{eq:meanforce_lang}
  \nabla_{z} F_{\rm rgd}^M(z) =   \int_{T^*\manq(z)} f_{\rm rgd}^M(q,p) \, \mu_{T^*\manq(z)}(dq\,dp).
\end{equation}
Moreover, for general constraints~\eqref{eq:fullconst} and the associated generalized free energy~\eqref{eq:F_xi}, the formula can be extended as follows: The generalized constraining force is
\begin{equation}\label{eq:genconstforce}
 \bmat f^\Xi \\ g^\Xi \emat := \sgram^{-1} \poisson{\Xi,H},
\end{equation}
where 
\[
\sgram(q,p) = \lbrace\Xi,\Xi \rbrace (q,p) = \nabla\Xi^T(q,p) \, J \, \nabla \Xi (q,p)
\]
is defined in~\eqref{eq:sgram}, and the rigid mean force is
\begin{equation}
  \label{eq:genconstforcevar}
  \nabla_{\zeta} F^\Xi_{\rm rgd}(\zeta)  =  \frac{1}{Z_\zeta}\int_{\manq_{\Xi}(\zeta)} \bmat f^\Xi \\ g^\Xi \emat  {\rm e}^{-\beta H} \, d\sigma_{\manq_{\Xi}(\zeta)},
\end{equation}
where $\dps Z_\zeta= \int_{\manq_{\Xi}(\zeta)} {\rm e}^{-\beta H} d\, \sigma_{\manq_{\Xi}(\zeta)}$, and $F^\Xi_{\rm rgd}(\zeta)$ is defined in~\eqref{eq:F_xi}. When $(q,p)$ verifies $p_\xi(q,p) = v_\xi(q,p)= 0$, then $g^\Xi(q,p)=0$ and $f^\Xi(q,p)= f_{\rm rgd}^M(q,p)$.
 \end{proposition}

\begin{proof}
Formulas~\eqref{eq:genconstforce} and~\eqref{eq:genconstforcevar} are obtained directly by replacing~$\ph$ by~${\rm e}^{-\beta H}$ in Lemma~\ref{l:consint} (see the Appendix). The fact that $(f^\Xi(q,p),g^\Xi(q,p))= (f_{\rm rgd}^M(q,p),0)$ in the tangential case (namely when $p_\xi(q,p) = v_\xi(q,p)= 0$) is a consequence of~\eqref{eq:dbcheckcons}.
 \end{proof}

The following lemma gives a momentum-averaged version of the constraining force (a similar formula exists in the overdamped case, see Equations~(4.8)-(4.9) in~\cite{ciccotti-lelievre-vanden-einjden-08}, for example).

\begin{lemma}
  \label{l:meanforce_av} 
  The rigid mean force~\eqref{eq:meanforce_lang} can be rewritten as:
   \begin{equation}
     \label{eq:avfbar}
     \nabla_{z} F_{\rm rgd}(z) =  \int_{T^*\manq(z)}  \overline{f}_{\rm rgd}^M(q) \, \mu_{T^*\manq(z)}(dq\,dp),
   \end{equation}
   where
   \begin{equation}\label{eq:constforceav}
     \overline{f}_{\rm rgd}^M(q) = G_M^{-1}(q) \nabla \xi (q)^T M^{-1}\nabla V (q) - \beta^{-1} G_M^{-1}(q) \op{Hess}_q(\xi) : \pare{M^{-1}P_M(q)}.
\end{equation}
\end{lemma}

\begin{proof}
  Consider the Gaussian distribution ${\kappa}^{M^{-1}}_{T^*_q\manq(z)}(dp)$ defined in~\eqref{eq:kinetic_part}:
  \[
  \kappa_{T^\ast_q \manq(z)}^{M^{-1}}(dp) =  
  \left(\frac{\beta}{2\pi}\right)^{(3N-\nc)/2} \exp\left(-\beta \frac{p^T M^{-1} p}{ 2}\right)
  \, \sigma^{M^{-1}}_{T^{\ast}_q \manq(z)}(dp),
  \]
  which is the marginal distribution in the momentum variable of the canonical distribution $\mu_{T^*\manq(z)}(dq \, dp)$, conditioned by a given $q\in \manq(z)$. Proving Lemma~\ref{l:meanforce_av} amounts to showing that the average of the constraining force $f_{\rm rgd}^M$ with respect to ${\kappa}^{M^{-1}}_{T^\ast  _q\manq(z)}(dp)$ yields $\overline{f}_{\rm rgd}^M$:
\[
\overline{f}_{\rm rgd}^M(q) = \int_{T^\ast  _q\manq(z)}  
f_{\rm rgd}^M(q,p) \, {\kappa}^{M^{-1}}_{T^\ast  _q\manq(z)}(dp).
\]
First, we compute the covariance matrix
\[
\mathcal{C} := \op{cov}\Big(\kappa^{M^{-1}}_{T^\ast  _q\manq(z)}\Big)
\]
of the Gaussian distribution $\kappa^{M^{-1}}_{T^\ast  _q\manq(z)}(dp)$. Since $\kappa^{M^{-1}}_{T^\ast  _q\manq(z)}(dp)$ is a centered Gaussian distribution, $\mathcal{C}$ satisfies, for all $p_1,p_2 \in \R^{3N}$,
\begin{align*}
  p_1^T M^{-1} \mathcal{C} M^{-1} p_2  
&:= \int_{T^\ast  _q\manq(z)} \pare{ p^TM^{-1} p_1    }  \pare{ p^TM^{-1} p_2    }   \kappa^{M^{-1}}_{T^\ast  _q\manq(z)}(dp) \\
&= \int_{T^\ast  _q\manq(z)} \pare{ p^TM^{-1} P_M(q) p_1    }  \pare{ p^TM^{-1} P_M(q) p_2    }   \kappa^{M^{-1}}_{T^\ast  _q\manq(z)}(dp).
   \end{align*}
Denoting $\langle p_1,p_2\rangle_{M^{-1}} = p_1^T M^{-1} p_2$, this yields
\begin{align*}
 p_1^T M^{-1} \mathcal{C} M^{-1} p_2  
&= \int_{T^\ast  _q\manq(z)} \langle p, P_M(q) p_1 \rangle_{M^{-1}}\langle p, P_M(q) p_2 \rangle_{M^{-1}} \frac{\dps {\rm e}^{-\frac{\beta  }{2} \langle p,p\rangle_{M^{-1}}}}{\dps \left(2\pi/\beta\right)^{(3N-\nc)/2}} 
\, \sigma^{M^{-1}}_{T^{\ast }_q \manq(z)}(dp)   \\
&=\beta^{-1} \langle P_M(q) p_1  ,  P_M(q) p_2  \rangle_{M^{-1}},
   \end{align*}
so that 
\begin{equation}
  \label{eq:cov_mat}
  \mathcal{C} = \beta^{-1} P_M(q) M.
\end{equation}
This gives
$$
 \int_{T^\ast  _q\manq(z)} \op{Hess}_q(\xi)(M^{-1} p, M^{-1} p) \, \kappa^{M^{-1}}_{T^\ast  _q\manq(z)}(dp) =   \beta^{-1} \op{Hess}_q(\xi) : \pare{M^{-1}P_M(q)} .
$$
Averaging~\eqref{eq:constforce} over momenta thus leads to the desired result.
 \end{proof}
  
Free energy derivatives can also be obtained from the Lagrange multipliers of the Langevin constrained process~$\EL$. This is very useful in practice since it avoids the computation of second order derivatives of the reaction coordinates which appear in the expressions of $f^M_{\rm rgd}$ and $\overline{f}^M_{\rm rgd}$ (see the discussion at the beginning of Section~\ref{sec:avg_Lagrange}):

\begin{theorem} Consider the rigidly constrained Langevin process solution of~$\EL$, with associated Lagrange multipliers $\lambda_t$. Assume that $\nabla \xi$, $G_M^{-1}$ and $\sigma$ are bounded functions on $\manq(z)$, and $\gamma_P$ is strictly positive on $T_q^*\Sigma(z)$ (in the sense of symmetric matrices). Then,
the almost sure convergence~\eqref{eq:avlagr} claimed in the introduction holds:
\begin{equation}
  \label{eq:avg_lagr_for_interpretation}
  \lim_{T \to + \infty}\frac{1}{T}\int_0^T d \lambda_t  = \nabla_{z} F_{\rm rgd}^M(z) \qquad {\rm a.s}
\end{equation}
A similar result holds for the `Hamiltonian part' of the Lagrange multipliers, defined by:
\begin{equation}
  \label{eq:lagrtilde}
  d \lambda_t^{\rm ham}  =  d \lambda_t + G_M^{-1} \nabla\xi(q_t) ^T  M^{-1} \pare{-\gamma(q_t)   M^{-1} p_t \, dt  +\sigma(q_t) d W_t } 
= f^M_{\rm rgd}(q_t,p_t) \, dt.
\end{equation}
Indeed, the following almost sure convergence holds:
\begin{equation}
  \label{eq:avlagr_tilde}
  \lim_{T \to + \infty}\frac{1}{T}\int_0^T d \lambda_t^{\rm ham} =  
  \nabla_{z} F_{\rm rgd}^M(z) \qquad {\rm a.s}.
\end{equation}
\end{theorem}

The estimator based on~\eqref{eq:avlagr_tilde} has a smaller variance than the estimator
based on~\eqref{eq:avlagr} (or~\eqref{eq:avg_lagr_for_interpretation} above)
since only the bounded variation part is retained, and the
martingale part due to the Brownian increments and the dissipation term are subtracted
out. Similar results on variance reduction where obtained in the overdamped case
in~\cite{ciccotti-lelievre-vanden-einjden-08}.

\begin{proof} 
  Recall the expression~\eqref{eq:first_expression_lagrange} of the Lagrange multipliers, which can be decomposed as the sum of the constraining force, a dissipation term and a martingale (fluctuation) term:
  \begin{align}
    d \lambda_t   
    &= f_{\rm rgd}^M(q_t,p_t) \, dt + G_M^{-1} \nabla\xi (q_t)^T M^{-1}\pare{ \gamma(q_t) M^{-1} p_t  \, dt  -\sigma(q_t) d W_t }.
    \label{eq:lagr}
\end{align}
The result follows from three facts. First, the process is ergodic with respect to the equilibrium distribution $\mu_{T^*\manq(z)}(dq \, dp)$ and averaging $f_{\rm rgd}^M$ yields the rigid free energy derivative in view of Proposition~\ref{p:meanforce}. This already shows~\eqref{eq:avlagr_tilde}.

Second, the Gaussian distribution of $\mu_{T^*\manq(z)}(dq \, dp)$ with respect to momentum variables is centered, which yields:
\[
 \int_{T^*\manq(z)}   G_M^{-1}(q) \nabla\xi(q)^TM^{-1}\gamma(q) M^{-1}\,  p \, \mu_{T^*\manq(z)}(dq \, dp) = 0.
\]
Third, the variance of the martingale term can be uniformly bounded as
\[
\E\abs{\frac{1}{\sqrt{T}}\int_0^T G_M^{-1}(q_t)\nabla\xi^T(q_t)M^{-1}\sigma(q_t) d W_t}^2 
\leq \norm{\op{Tr}(G_M^{-1}\nabla\xi^TM^{-1} \sigma\sigma^T M^{-1} \nabla\xi G_M^{-1} ) }_{\infty}.
\]
This implies the almost sure convergence 
\[
\dps \lim_{T \to +\infty} \frac{1}{T}\int_0^T G_M^{-1}(q_t)\nabla\xi (q_t) ^TM^{-1}\sigma(q_t) d W_t 
= 0,
\]
see for example Theorem 1.3.15 in~\cite{duflo-97}.
\end{proof}

The fact that averaging the Lagrange multiplier in~\eqref{eq:avg_lagr_for_interpretation} indeed yields the mean force may not be intuitive. This is actually very much related to the cost interpretation of the Lagrange multipliers in optimization, see~\cite[Remark~3.29]{LelRouStoBook}.

\subsection{Numerical discretization of the mean force}
\label{sec:num_ti}

Estimates of the mean force  based on either~\eqref{eq:meanforce_lang}, \eqref{eq:avfbar} or~\eqref{eq:avlagr_tilde} can be obtained.

\subsubsection{Averaging local rigid mean forces}

Free energy derivatives can be computed by averaging $\overline{f}_{\rm rgd}^M(q)$ or $f_{\rm rgd}^M(q,p)$ with respect to the distribution $\mu_{T^*\manq(z)}(dq \, dp)$, for instance using the estimators: 
\[
\lim_{K \to +\infty} \frac{1}{K}\sum_{k=0}^{K-1} \overline{f}_{\rm rgd}^M(q^k)
\]
or
\[
\lim_{K \to +\infty} \frac{1}{K}\sum_{k=0}^{K-1} f_{\rm rgd}^M(q^k,p^k).
\]
The functions $\overline{f}_{\rm rgd}^M(q)$ and $f_{\rm rgd}^M(q,p)$ may thus be called ``rigid local mean forces''. 
Note that using the momentum-averaged local mean force $\overline{f}_{\rm rgd}^M$ instead of the original $f_{\rm rgd}^M$ reduces the variance
since the fluctuations of the momentum variable have been averaged out analytically.
Table~\ref{tab:comparison_variance_Langevin_TI} below confirms this analysis, although the variance reduction appears to be small in our specific numerical experiment.

Assuming the convergence of the constrained splitting scheme~\eqref{eq:flucdiss1}-\eqref{eq:Verletconst}-\eqref{eq:flucdiss2} in the probability distribution sense\footnote{This convergence is also called weak convergence in probability theory. The proof of convergence in the present case may be carried out using classical results, see {\it e.g.} \cite{EthKur86}.} to the limiting Langevin process~$\EL$, the convergence of these estimators to $\nabla_z F^M_{\rm rgd}(z)$ is ensured, when taking first the limit $\dt \to 0$ with $K=N_{\dt}$ such that $N_{\dt} \dt \to T$, and then $T \to \infty$.
 
\subsubsection{Averaging the Lagrange multipliers}
\label{sec:avg_Lagrange}

Free energy derivatives can also be computed using the Lagrange multipliers of a Langevin constrained process according to~\eqref{eq:avlagr} or~\eqref{eq:avlagr_tilde}. This technique avoids the possibly cumbersome computation of second order derivatives $\op{Hess}_q(\xi) $ of the reaction coordinate, which appear in the expressions of $f^M_{\rm rgd}$ or $\overline{f}_{\rm rgd}^M$. Besides, the Lagrange multipliers are needed anyway for the numerical integration of the dynamics.

The computation can be performed as before with a longtime simulation of the splitting scheme~\eqref{eq:flucdiss1}-\eqref{eq:Verletconst}-\eqref{eq:flucdiss2} discretizing the Langevin process with constraints. The following approximation formula can for instance be used:
\begin{equation}\label{eq:freeestim}
 \nabla_z F_{\rm rgd}^M(z) \simeq \frac{1}{K\dt }\sum_{k=0}^{K-1} (\lambda^{k+1/2} +\lambda^{k+3/4}) 
\end{equation}
where $(\lambda^{k+1/2},\lambda^{k+3/4})$ are the Lagrange multipliers in the Hamiltonian part~\eqref{eq:Verletconst}. The consistency of this estimator is given by the following proposition.

\begin{proposition}[Consistency]\label{p:consist} The approximation formula~\eqref{eq:freeestim} is consistent. More precisely, the Lagrange multipliers $(\lambda^{n+1/2},\lambda^{n+3/4})$ in~\eqref{eq:flucdiss1}-\eqref{eq:Verletconst}-\eqref{eq:flucdiss2} are both equivalent when $\dt \to 0$ to the constraining force defined in~\eqref{eq:constforce}:
  $$
  \begin{cases}
    \dps \lambda^{n+1/2}   =    f_{\rm rgd}^M(q^n,p^{n+1/2}) \frac{\dt}{2} +{\rm O }(\dt^2), \\
    \dps \lambda^{n+3/4}   =    f_{\rm rgd}^M(q^{n+1},p^{n+1/2}) \frac{\dt}{2} + {\rm O } (\dt^2) .
  \end{cases}
  $$
  Moreover, the following second order consistency holds for the sum of the Lagrange multipliers:
  \begin{equation}
    \label{eq:lag_order_var_other}
    \lambda^{n+1/2}  +  \lambda^{n+3/4}= \frac{\dt}{2} \Big( f_{\rm rgd}^M(q^n,p^{n+1/2}) +  
    f_{\rm rgd}^M(q^{n+1},p^{n+1/2}) \Big)  +  { \rm O} (\dt^3), 
  \end{equation}
together with the variant:
\begin{equation}
\label{eq:lag_order_var}
  \dps \lambda^{n+1/2}  +  \lambda^{n+3/4}  =  \frac{\dt}{2} \Big(
  f_{\rm rgd}^M(q^n,p^{n+1/4}) + f_{\rm rgd}^M(q^{n+1},p^{n+3/4}) \Big) +  { \rm O} (\dt^3).
\end{equation}
\end{proposition}

The variant~\eqref{eq:lag_order_var}, which involves positions and momenta at the beginning 
and at the end of the Hamiltonian steps only, is used in~\eqref{eq:estim_lag_metro} 
below to estimate the time discretization error in the thermodynamic integration method based on
the estimator~\eqref{eq:avlagr_tilde}.

\begin{proof}
For sufficiently small time-steps $\dt$, the implicit function theorem ensures that the two projection steps associated with the nonlinear constraints in~\eqref{eq:flucdiss1}-\eqref{eq:Verletconst}-\eqref{eq:flucdiss2} have a unique smooth solution. A Taylor expansion with respect to~$\dt$ 
of the position constraints gives
\begin{align*}
  z & =  \xi(q^{n+1}) = \xi(q^n + \dt  M^{-1} p^{n+1/2})  \\
  &=  \xi(q^n) + \dt \nabla \xi (q^n) ^T  M^{-1} p^{n+1/2} 
  + \frac{\dt^2}{2} \op{Hess}_{q^n}(\xi)(M^{-1} p^{n+1/2}, M^{-1} p^{n+1/2}) \\
  &\quad + \frac{\dt^3}{6} \op{D}^3_{q^n}(\xi)(M^{-1} p^{n+1/2}, M^{-1} p^{n+1/2},M^{-1} p^{n+1/2}) + { \rm O}(\dt^4),
\end{align*}
where $D^3_q(\xi)(x,y,z) \in \R^{\nc}$ denotes the order $3$ differential of $\xi$ computed at $q$ and evaluated with the vectors $x,y,z \in \R^{3N}$. We denote
\[
\alpha^{n+1/2}(q) :=  G_M^{-1}(q)\op{D}_{q}^3(\xi)(M^{-1} 
p^{n+1/2}, M^{-1} p^{n+1/2},M^{-1} p^{n+1/2}).
\]
Then, the fact that $z= \xi(q^{n+1}) = \xi(q^n)$ and the identity
\[
\nabla \xi  (q^n) ^T  M^{-1} p^{n+1/2} = - \frac{\dt}{2} \nabla \xi (q^n) ^T M^{-1}  \nabla V(q^n) + G_M(q^n)  \lambda^{n+1/2} 
\]
yield the following expansion of $ \lambda^{n+1/2}$ in terms of $(q^n,p^{n+1/2})$:
\[
\lambda^{n+1/2} = f_{\rm rgd}^M(q^n,p^{n+1/2}) \frac{\dt}{2} -\frac{\dt^2}{6} \alpha^{n+1/2}(q^n)  + {\rm O }(\dt^3).
\]
By time symmetry, the same computation holds for $\lambda^{n+3/4}$, starting from $(q^{n+1},p^{n+3/4})$ and by formally replacing $\dt$ by $-\dt$. This can be double checked by Taylor expanding with respect to $\dt$ the position constraints, as done above for $\lambda^{n+1/2}$.
It thus holds:
\[
\lambda^{n+3/4} = f_{\rm rgd}^M(q^{n+1},p^{n+1/2}) \frac{\dt}{2} +\frac{\dt^2}{6} \alpha^{n+1/2}(q^{n+1})  + {\rm O }(\dt^3).
\]
The sum of the multipliers therefore reads
\begin{align*} 
  & \lambda^{n+1/2} + \lambda^{n+3/4}  -  f_{\rm rgd}^M(q^n,p^{n+1/2}) \frac{\dt}{2}  -  f_{\rm rgd}^M(q^{n+1},p^{n+1/2}) \frac{\dt}{2} \\
  &\qquad   = \frac{\dt^2}{6} \pare{\alpha^{n+1/2}(q^{n+1}) - \alpha^{n+1/2}(q^{n}) } + { \rm O} (\dt^3) = { \rm O} (\dt^3),
\end{align*}
which gives~\eqref{eq:lag_order_var_other}.
Now, using the previous calculations, we remark that:
\[
\begin{cases}
  \dps p^{n+1/2} = p^{n+1/4} -\frac{\dt}{2} \nabla V (q^n) + \frac{\dt}{2} \nabla \xi (q^{n})   f_{\rm rgd}^M(q^n,p^{n+1/2}) + {\rm O}(\dt ^2 ), \\[6pt]
  \dps  p^{n+1/2} = p^{n+3/4} + \frac{\dt}{2} \nabla V (q^{n+1}) - \frac{\dt}{2} \nabla \xi (q^{n+1})   f_{\rm rgd}^M(q^{n+1},p^{n+1/2}) + {\rm O}(\dt ^2 ).
\end{cases}
\]
Thus, it holds
\begin{align*}
& f_{\rm rgd}^M(q^n,p^{n+1/2}) + f_{\rm rgd}^M(q^{n+1},p^{n+1/2}) = 
f_{\rm rgd}^M(q^n,p^{n+1/4}) + f_{\rm rgd}^M(q^{n+1},p^{n+3/4})  \\
& \qquad + \nabla_p f_{\rm rgd}^M(q^n,p^{n+1/4}) \pare{ -\frac{\dt}{2} \nabla V (q^n) + \frac{\dt}{2} \nabla \xi (q^{n})   f_{\rm rgd}^M(q^n,p^{n+1/2})    } \\
& \qquad - \nabla_p f_{\rm rgd}^M(q^{n+1},p^{n+3/4}) \pare{ -\frac{\dt}{2} \nabla V (q^{n+1}) + \frac{\dt}{2} \nabla \xi (q^{n+1})   f_{\rm rgd}^M(q^{n+1},p^{n+1/2})    }  + {\rm O}(\dt ^2 )\\
& = f_{\rm rgd}^M(q^n,p^{n+1/4}) + f_{\rm rgd}^M(q^{n+1},p^{n+3/4}) + {\rm O}(\dt ^2 ).
\end{align*}
This gives the claimed second order consistency of the sum of the Lagrange multipliers~\eqref{eq:lag_order_var}.
\end{proof}

Let us discuss the convergence of the approximation~\eqref{eq:freeestim}.
Assuming again that the constrained splitting scheme~\eqref{eq:flucdiss1}-\eqref{eq:Verletconst}-\eqref{eq:flucdiss2} converges in the probability distribution sense 
to the limiting Langevin process~$\EL$,  
the following convergence in probability distribution occurs
when $\dt \to 0$ and $N_{\dt} \dt \to T$:
\begin{equation}\label{eq:consistlagr}
 \lim_{\dt\to 0  } {\rm Law} \pare{\frac{1}{N_{\dt }\dt }\sum_{n=0}^{N_{\dt}-1} (\lambda^{n+1/2} +\lambda^{n+3/4})  }= {\rm Law} \pare{ \frac{1}{T}\int_0^T d\lambda_t^{\rm ham} }.
\end{equation}
This shows the convergence of the estimate~\eqref{eq:freeestim} of the mean force when taking first the limit $\dt \to 0$ and then $T \to \infty$.

\subsubsection{Estimates relying on the Metropolized scheme}

When the scheme~\eqref{eq:flucdiss1}-\eqref{eq:Verletconst}-\eqref{eq:flucdiss2} is complemented
with a Metropolis step (see Algorithm~\ref{a:GHMCconst}), it is possible to prove a result on the longtime limit
of trajectorial averages (\emph{i.e.} letting first the number of iterations go to infinity, and then 
taking the limit $\dt \to 0$), upon assuming
the irreducibility of the numerical scheme. 

Indeed, let us consider the Markov chain $(q^k,p^k)$ generated by the GHMC scheme in
Algorithm~\ref{a:GHMCconst}, and assume (i) the irreducibility of the
 Markov chain, and (ii) that appropriate rejections outside the
set $\overline{D}_{\dt} = \manq(z) \times \{ \frac12 p^T M^{-1} p \leq R_{\dt} \}$ 
are made in the steps (1)-(2)-(4) of the
algorithm. In particular, the projection steps associated with the nonlinear
constraints in Step~(2) of Algorithm~\ref{a:GHMCconst} are
well defined.

Then, by ergodicity, an average of the analytic expression of the local rigid mean force $\overline{f}_{\rm rgd}^M$ given in~\eqref{eq:constforceav} yields an estimate of  the free energy without time discretization error:
\[
\lim_{K \to + \infty}\frac{1}{K} \sum_{k=0}^{K-1} \overline{f}_{\rm rgd}^M(q^k) =\nabla_z F_{\rm rgd}^M(z) \quad \mathrm{a.s.}
\]
If $f_{\rm rgd}^M$ is used instead of~$\overline{f}_{\rm rgd}^M$, then the mean force is computed with some exponentially small error: almost surely,
\begin{align*}
\lim_{K \to + \infty}\frac{1}{K} \sum_{k=0}^{K-1} f_{\rm rgd}^M(q^k,p^k) &=
\frac{ \dps \int_{ \overline{D}_{\dt}} f_{\rm rgd }^M(q,p) 
     \mu_{T^\ast  \manq(z)}(dq \, dp)}{ \dps \int_{ \overline{D}_{\dt} } 
     \mu_{T^\ast  \manq(z)}(dq \, dp) } \\
&= \frac{ \dps \int_{ T^*\manq(z)} f_{\rm rgd }^M(q,p) 
     \mu_{T^\ast  \manq(z)}(dq \, dp)}{ \dps \int_{ T^*\manq(z) } 
     \mu_{T^\ast  \manq(z)}(dq \, dp) } + \mathrm{O}\left(\mathrm{e}^{-\alpha \dt^{-2}}\right) \\
&=\nabla_z F_{\rm rgd}^M(z) + \mathrm{O}\left(\mathrm{e}^{-\alpha \dt^{-2}}\right)
\end{align*}
for some $\alpha > 0$.
The error arising from replacing $\overline{D}_{\dt}$ with $T^*\manq(z)$ is indeed 
exponentially small in view of~\eqref{eq:estimate_Rdt} (namely 
$R_{\Delta t} \geq A \, \Delta t^{-2}$) and
using the fact that the marginal distribution in the $p$-variable is Gaussian.

Likewise, for estimates based on Lagrange multipliers, the following longtime averaging
holds: almost surely,
\begin{equation}
  \label{eq:estim_lag_metro}
   \lim_{K \to + \infty}\frac{1}{K \dt}\sum_{k=0}^{K-1} 
  (\lambda^{k+1/2} +\lambda^{k+3/4}) 
   = \frac{ \dps \int_{ \overline{D}_{\dt}} f_{\rm rgd }^M(q,p) 
     \mu_{T^\ast  \manq(z)}(dq \, dp)}{ \dps \int_{ \overline{D}_{\dt} } 
     \mu_{T^\ast  \manq(z)}(dq \, dp) } + {\rm O}(\dt^2),
\end{equation}
where we have used the estimate~\eqref{eq:lag_order_var} on the Lagrange multipliers. The limit $\dt \to 0$ is obtained by a dominated convergence argument:
\[
\lim_{\dt\to 0} \, \lim_{K \to + \infty}\frac{1}{K \dt}\sum_{n=0}^{K-1} (\lambda^{k+1/2} +\lambda^{k+3/4})  =\nabla_z F_{\rm rgd}^M(z) \quad \mathrm{a.s.}
\]
Note that, due to the Metropolis correction in Algorithm~\ref{a:GHMCconst}, the time discretization error in the sampling of the invariant measure is removed. The only remaining time discretization errors come from (i)~the approximation of the local mean force by the Lagrange multipliers (this is a second order error), and (ii)~the integration domain being $\overline{D}_{\dt}$ instead of $T^*\manq(z)$ (as discussed above, this is an exponentially small error in~$\dt$). In conclusion, the left-hand side of~\eqref{eq:estim_lag_metro} is an approximation of~$\nabla_z F_{\rm rgd}^M(z)$ up to a ${\rm O}(\dt^2)$ error term.

\subsubsection{Overdamped limit}

Finally, let us emphasize that free energy derivatives can 
be computed with the estimator~\eqref{eq:freeestim}
within the overdamped Langevin framework, using the scheme~\eqref{eq:Eulerconst} and the expressions~\eqref{eq:lag_od_1}-\eqref{eq:lag_od_2}-\eqref{eq:lag_od_3} of Proposition~\ref{p:langtooverd}. Let us recall that the latter are equivalent to the scheme~\eqref{eq:flucdiss1}-\eqref{eq:Verletconst}-\eqref{eq:flucdiss2} with fluctuation-dissipation matrices satisfying
$\frac{\dt}{4} \gamma = M = \frac{\dt}{2} \op{Id}$.
This leads to the original free energy estimator 
(recall that, for the overdamped dynamics, $\R^{3N}$ is equipped with the 
scalar product associated with the identity matrix):
\begin{equation}\label{eq:freeestim_od}
 \nabla_z F_{\rm rgd}^\op{Id}(z) \simeq \frac{1}{K\dt}\sum_{k=0}^{K-1} (\lambda^{k+1/2} + \lambda^{k+3/4}),
\end{equation}
which can be seen as a variant of the variance reduced estimator proposed directly for the overdamped scheme~\eqref{eq:Eulerconst} in~\cite{ciccotti-lelievre-vanden-einjden-08}:
\begin{equation*}
 \nabla_z F_{\rm rgd}^{\rm Id}(z) \simeq \frac{1}{K\dt }\sum_{k=0}^{K-1} \pare{ \lambda^{k+1}_{\rm od} + \sqrt{\frac{2 \dt}{\beta}} G^{-1}(q^k)  \nabla \xi (q^k) ^T {\mathcal G}^k} = 
\frac{1}{K\dt }\sum_{k=0}^{K-1} 2\lambda^{k+1/2}.
\end{equation*}
The rigorous justification of the consistency of~\eqref{eq:freeestim_od} in the limit $\dt \to 0$ follows from the results of~\cite{ciccotti-lelievre-vanden-einjden-08}.
See also Section~\ref{sec:time_continuous_limit_ovd} below for similar results.

\subsection{Numerical illustration}
\label{sec:simple_example_WCA}

We consider a system composed of $N$ particles in a 2-dimensional periodic box of side
length~$L$, interacting through the purely repulsive 
WCA pair potential, which is a truncated Lennard-Jones potential:
\[
V_{\rm WCA}(r) = \left \{ \begin{array}{cl}
\dps 4 \varepsilon \left [ \left ( \frac{\sigma}{r} \right )^{12} 
  - \left ( \frac{\sigma}{r}\right )^6 \right ] + \varepsilon & \quad {\rm if \ } r \leq r_0, \\
0 & \quad {\rm if \ } r > r_0,
\end{array} \right.
\]
where $r$ denotes the distance between two particles, 
$\varepsilon$ and $\sigma$ are two positive parameters and $r_0=2^{1/6}\sigma$.
Among these particles, two (numbered 1 and 2 in the following) are
designated to form a dimer while the others are
solvent particles. Instead of the above WCA
potential, the interaction potential between the two particles of the dimer 
is a double-well potential
\begin{equation}
  \label{intro:eq:VS}
  V_{\rm S}(r) = h \left [ 1 - \frac{(r-r_0-w)^2}{w^2} \right ]^2,
\end{equation}
where $h$ and $w$ are two positive parameters.  
The total energy of the system is therefore, for $q \in (L\mathbb{T})^{dN}$ with $d=2$,
\[
V(q) = V_{\rm S}(|q_1-q_2|) + \sum_{3 \leq i<j \leq N} V_{\rm WCA}(|q_i-q_j|)
+ \sum_{i=1,2} \sum_{3 \leq j \leq N} V_{\rm WCA}(|q_i-q_j|).
\]
See~\cite{DBC99,SBB88} for instance for other computational studies using this model.

The potential $V_{\rm S}$ exhibits two energy minima, one corresponding
to the compact state where the length of the dimer is $r=r_0$, and one corresponding to
the stretched state where this length is
$r=r_0+2w$. The energy barrier separating both states is $h$. 
The reaction coordinate used to describe the transition from the compact to 
the stretched state is the normalized bond length of the dimer molecule:
\begin{equation}
  \label{coord_reac_diatomic}
  \xi(q) = \frac{|q_1-q_2|-r_0}{2w},
\end{equation}
where $q_1$ and $q_2$ are the positions of the two particles forming the
dimer. The compact state (resp. the stretched state) corresponds to the
value $z=0$ (resp. $z=1$) of the reaction coordinate.

The inverse temperature is set to $\beta = 1$, with $N = 100$ particles ($N-2$ solvent particles
and the dimer) with solvent density 
$\rho = (1-2/N) a^{-2} = 0.436$,
since there are $N-2$ solvent 
particles in a square box of side length $L=a\sqrt{N}$ with $a = 1.5$. The parameters
describing the WCA interactions are set to $\sigma = 1$ and $\varepsilon = 1$, and
the additional parameters for the dimer are $w = 2$ and $h = 2$. 

For this system, $M=\op{Id}$ and $|\nabla \xi|$ is constant, 
so that the rigid free energy $F^M_{\rm rgd }(z)$ is equal to the free energy $F(z)$.

The mean force is estimated at the values $z_i = z_{\rm min} + i \Delta z$,
with $z_{\rm min} = - 0.2$, $z_{\rm max} = 1.2$ and $\Delta z = 0.014$,
by ergodic averages obtained with the projected dynamics
with Metropolis correction (Algorithm~\ref{a:GHMCconst}, where in the simple case
considered here, the fluctuation-dissipation part can be integrated exactly). 
For each value of~$z$, we integrate the dynamics on a time 
$T = 2 \times 10^4$ with a step size $\dt = 0.02$, using a scalar friction
coefficient~$\gamma = 1$.

The resulting mean force profile is presented in Figure~\ref{fig:TI_lang_1}, 
together with the associated free energy profile.
\begin{figure}
\centering
\includegraphics[width=9cm]{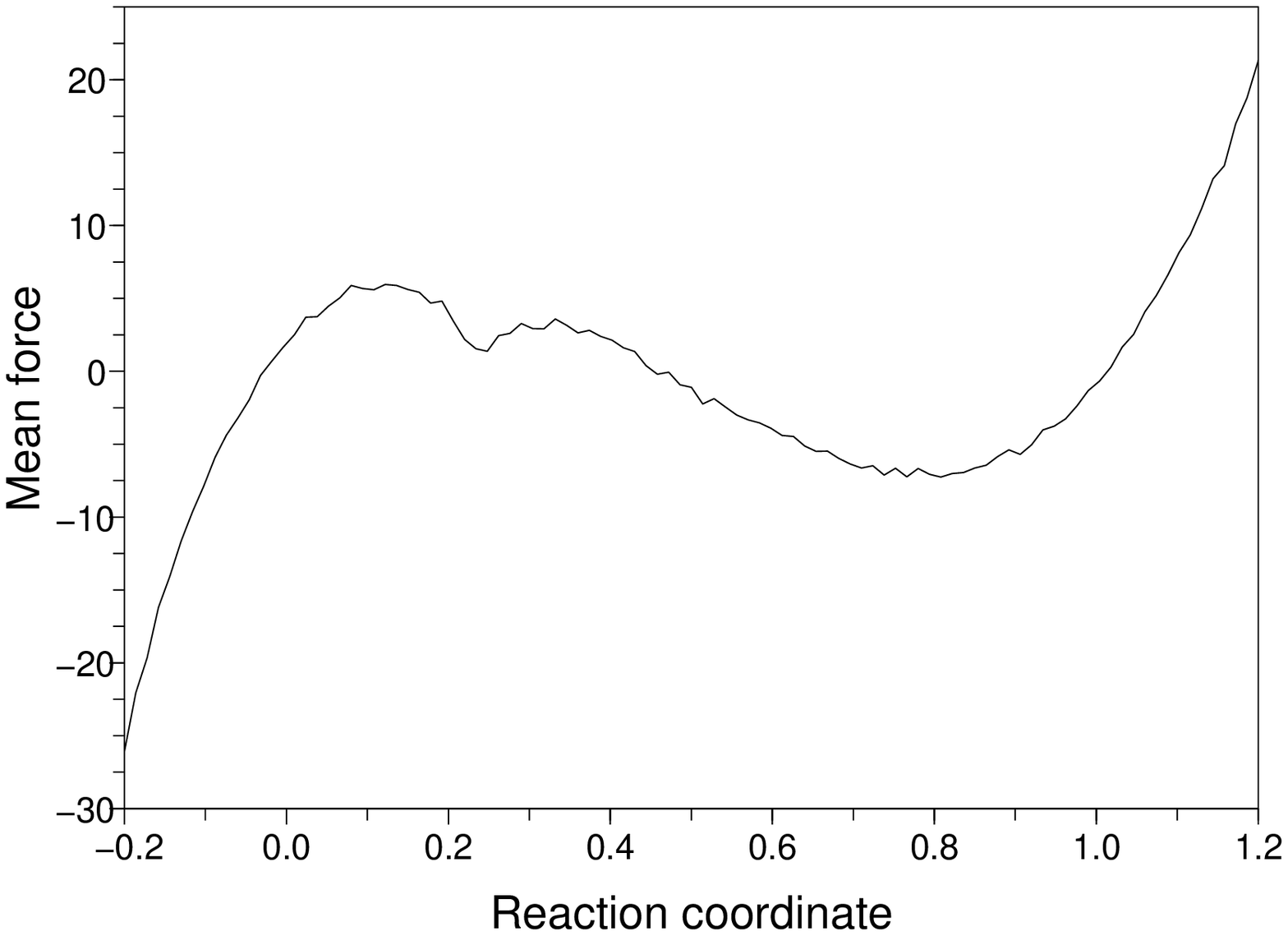}
\includegraphics[width=9cm]{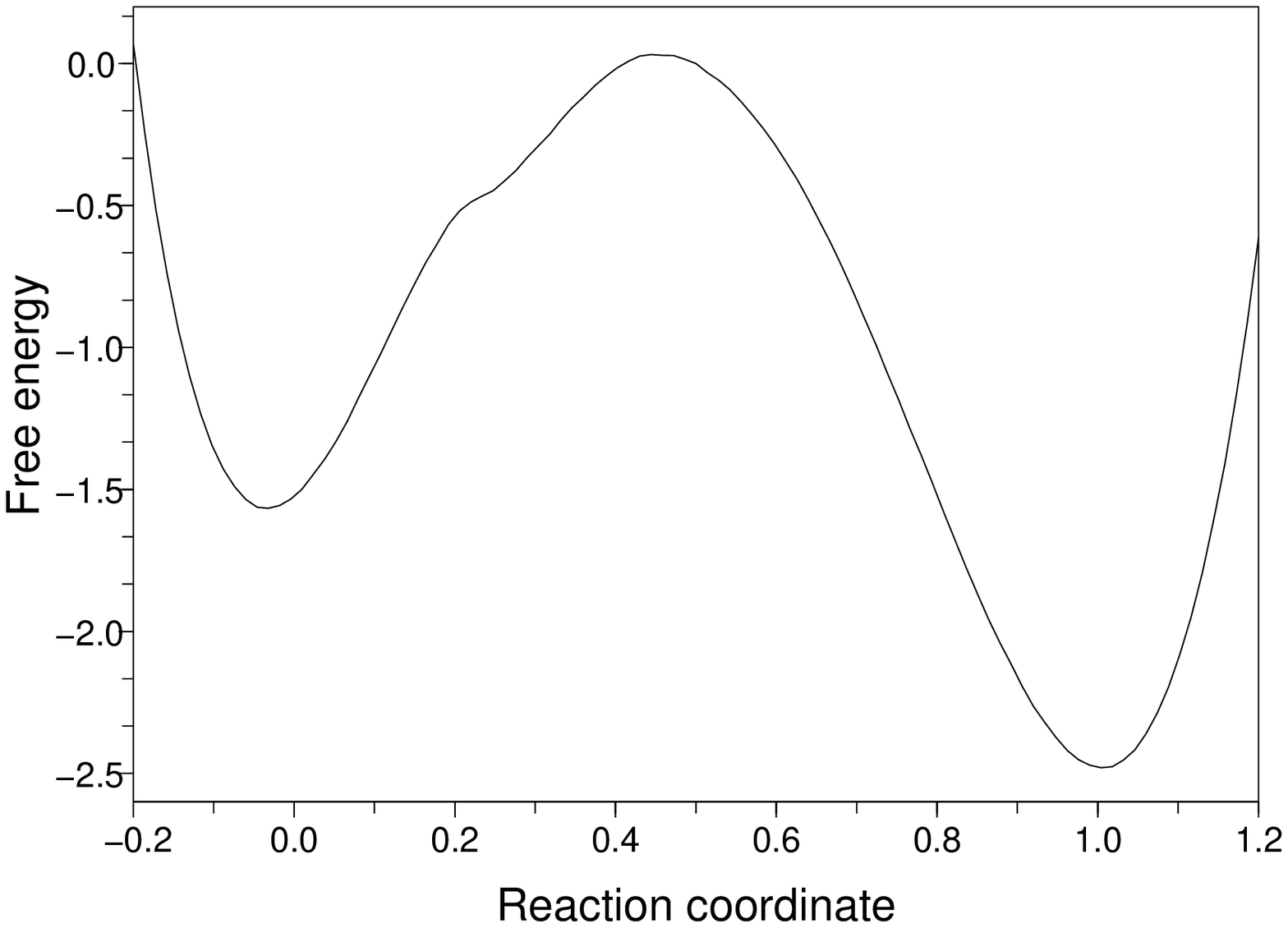}
\caption{\label{fig:TI_lang_1} 
  Top: Estimated mean force. Bottom: Corresponding free energy profile.
}
\end{figure}
Figure~\ref{fig:TI_lang_2} 
compares the analytical constraining force $f^M_{\rm rgd}(q^n,p^n)$ and the Lagrange multipliers, see Proposition~\ref{p:consist}. In Figure~\ref{fig:TI_lang_2}, the $x$-axis represents the blocks of $10^5$ simulation steps, concatenated for the $101$ different values of $z_i$. It can be seen that the difference between $f^M_{\rm rgd}(q^n,p^n)$ and the Lagrange multipliers is small in any cases, though somewhat larger for the lowest values of~$\xi$.
\begin{figure}
\centering
\includegraphics[width=9cm]{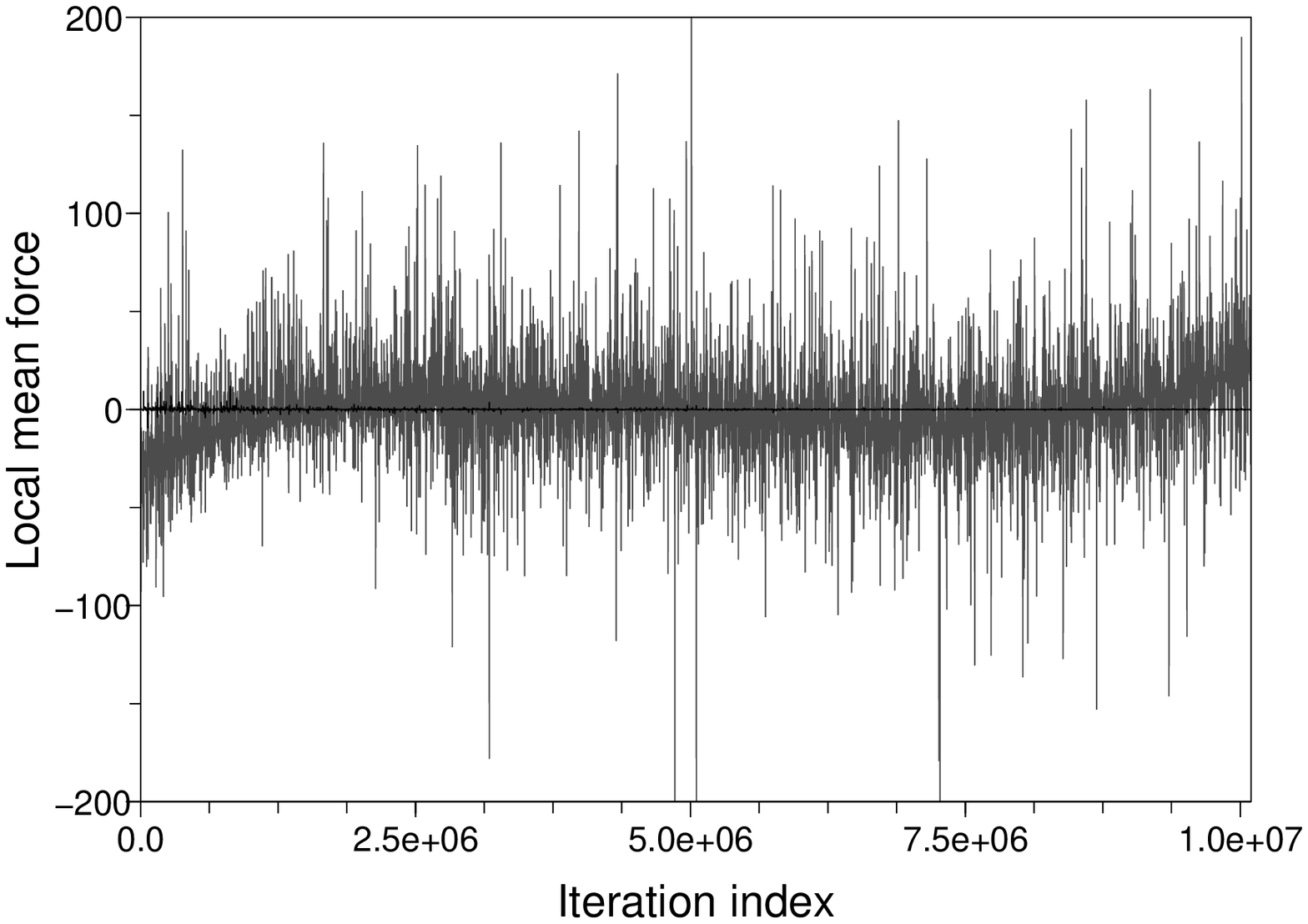}
\includegraphics[width=9cm]{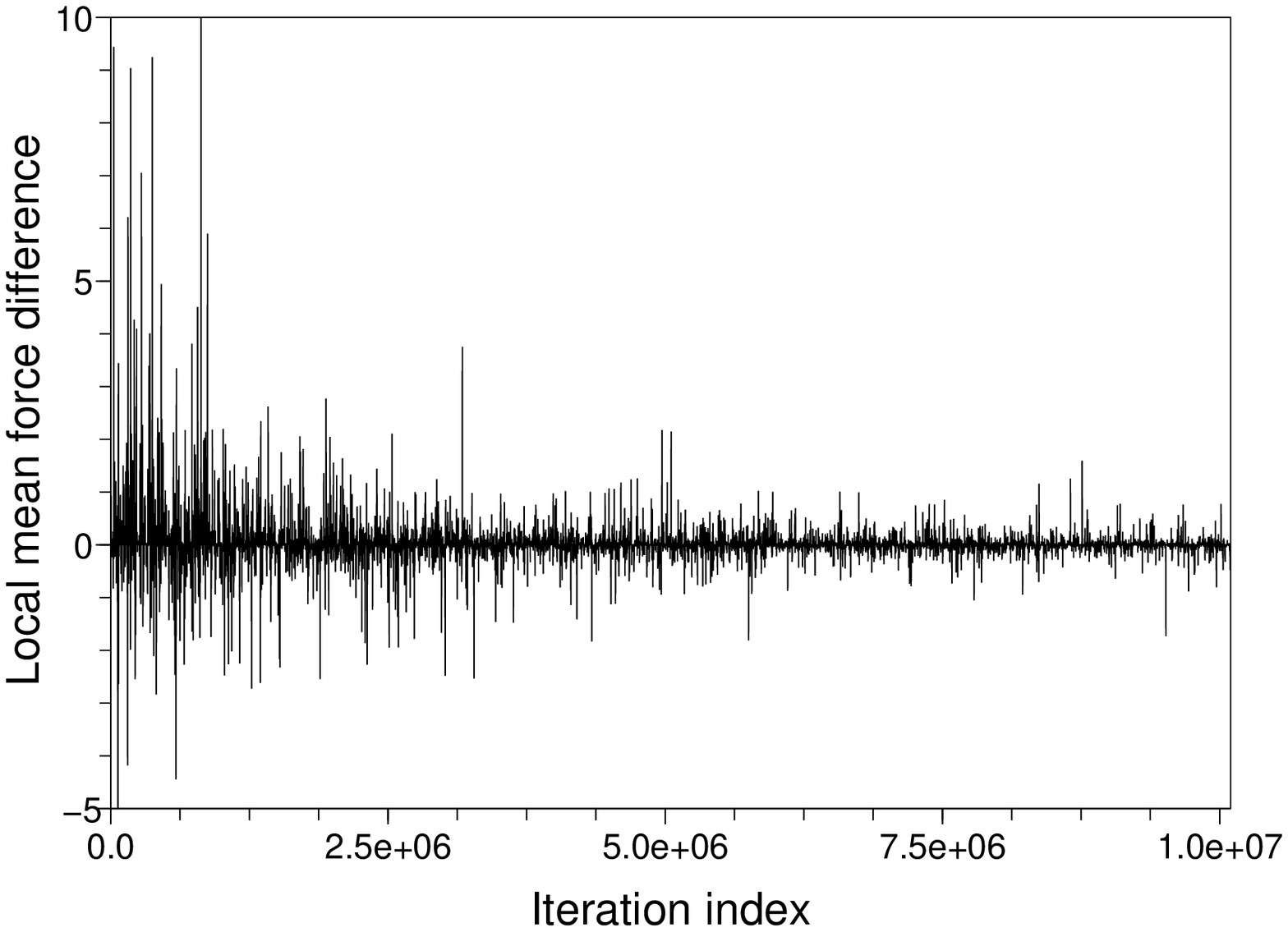}
\includegraphics[width=9cm]{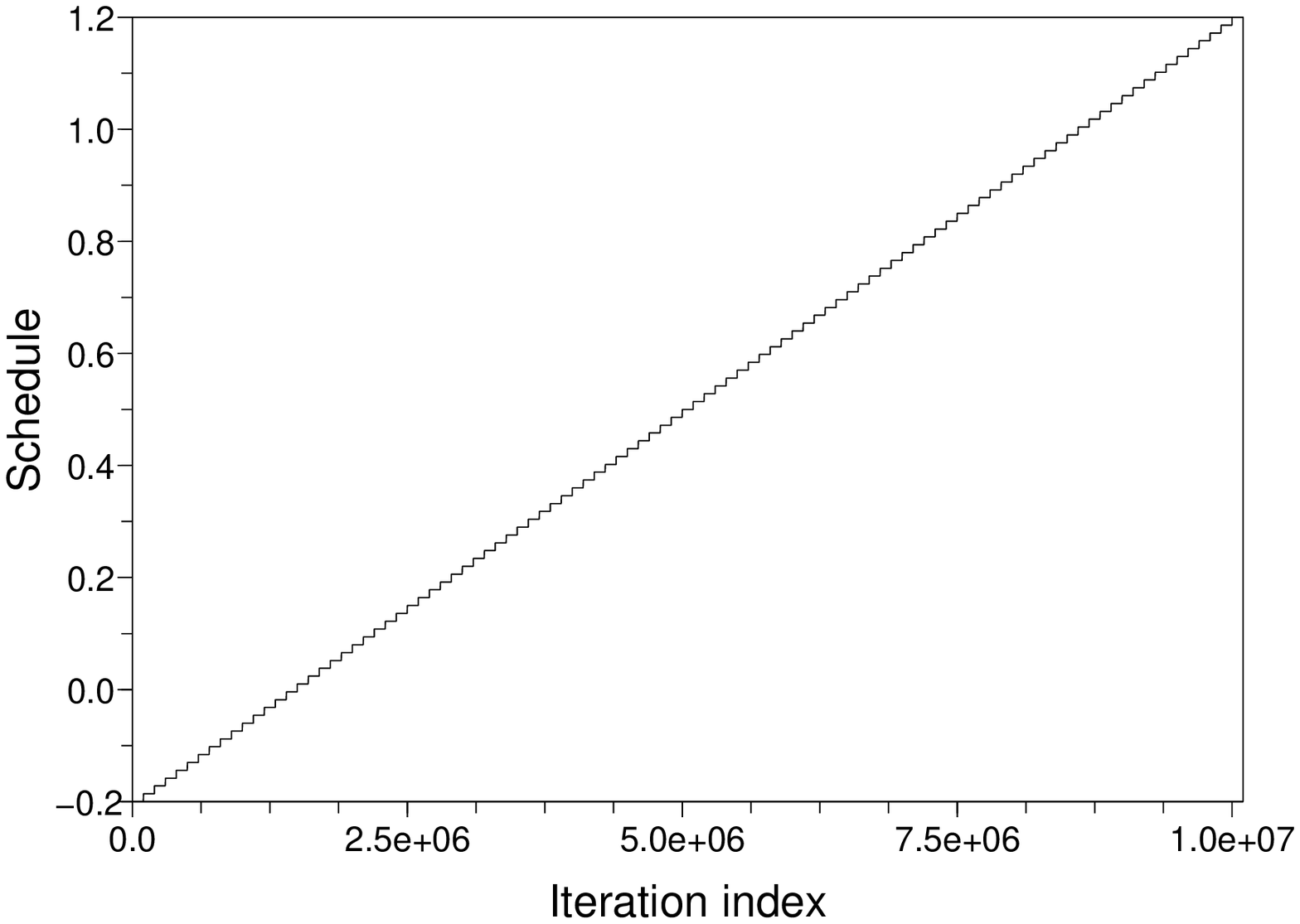}
\caption{\label{fig:TI_lang_2} 
  Top: 
  The constraining force $f^M_{\rm rgd}(q^n,p^n)$ (pale line), and the difference between the
  constraining force and its estimate from the Lagrange multipliers (dark line).
  Middle: 
  Zoom on the difference between the
  constraining force and the Lagrange multipliers. Note the difference of 
  scales for the $y$-axis.
  Bottom: The schedule
  $\xi$ is piecewise constant.
  In all figures, the $x$-axis represents the blocks of $10^5$ simulation steps, 
  concatenated for the $101$ different values of $z_i$.
}
\end{figure}

It can be checked numerically that the differences
$|\lambda^{n+1/2} - f_{\rm rgd}^M(q^{n},p^{n+1/2})\frac{\dt}{2}|$ and $|\lambda^{n+3/4} - f_{\rm rgd}^M(q^{n+1},p^{n+1/2})\frac{\dt}{2}|$ are indeed of order $\dt^2$, and that the difference 
\[
\abs{ \lambda^{n+1/2} +\lambda^{n+3/4}   -  f_{\rm rgd}^M(q^n,p^{n+1/2}) \frac{\dt}{2}  -  f_{\rm rgd}^M(q^{n+1},p^{n+1/2}) \frac{\dt}{2} }
\]
is indeed of order $\dt^3$ (by computing the average of these elementary differences for various step sizes). The Lagrange multipliers are in any case very good approximations to the constraining force~$f^M_{\rm rgd}$.

Let us finally discuss the efficiency of the different estimators of the mean force, in
terms of their variances. They can be written as the empirical average of the following random sequences:
\[
\pare{
f^M_{\rm rgd}(q^n,p^n), 
\overline{f}^M_{\rm rgd}(q^n),
\frac{\lambda^{n+1/2} +\lambda^{n+3/4}}{\dt}
},
\]
where $q^n,p^n,\lambda^{n+1/2},\lambda^{n+3/4}$ are given by the numerical scheme~\eqref{eq:flucdiss1}-\eqref{eq:Verletconst}-\eqref{eq:flucdiss2}. The correlations in time (between the iterates) are very similar for the three methods,
and we therefore simply compute the variance over all the samples.
Table~\ref{tab:comparison_variance_Langevin_TI} compares the so-obtained
standard errors over $10^5$ time-steps with $\dt = 0.02$ (simulation time $T=2 000$
for each value of the reaction coordinate). The results show that the 
different estimators are more or less equivalent.
This is related to the fact that the essential source of variance comes from the 
sampling of the positions, and not the sampling of the velocities. Note however that, 
for the smallest value of
the reaction coordinate, the estimator 
based on the averaged local mean force $\overline{f}^M_{\rm rgd}(q^n)$ appears to 
be better in terms of variance. 

\begin{table}[ht]
\begin{center}
{\begin{tabular}{c@{\quad}c@{\quad}c@{\quad}c}
\hline
$z$ & $f^M_{\rm rgd}(q^n,p^n) $ & $\dps \frac{\lambda^{n+1/2} +\lambda^{n+3/4}}{\dt}$ 
& $\overline{f}^M_{\rm rgd}(q^n)$ \\[10pt]
\hline
-0.2 & 22.1 & 21.9 & 14.7 \\  
 0.0 & 16.0 & 15.5 & 15.4 \\  
 0.2 & 23.1 & 22.5 & 22.9 \\  
 0.4 & 21.1 & 20.4 & 21.0 \\  
 0.6 & 21.4 & 20.7 & 21.3 \\  
 0.8 & 21.6 & 20.9 & 21.5 \\  
 1.0 & 21.4 & 20.6 & 21.4 \\  
 1.2 & 21.0 & 20.3 & 20.9 \\  
\hline  
\end{tabular}}
\caption{Standard error (square-root of the variance) 
  of three mean force estimators, with correlations in time
  neglected, for different values $z$ of the reaction coordinate.}
\label{tab:comparison_variance_Langevin_TI}
\end{center}
\end{table}

%
%

\section{Hamiltonian and Langevin nonequilibrium dynamics}
\label{sec:lang_jarz}

This section presents nonequilibrium Hamiltonian and
Langevin dynamics with time-evolving constraints. We thus consider $(q_t,p_t)$ solution to the dynamics~$\NL$, which we recall for convenience:
\[
\NL \qquad
  \left\{
  \begin{aligned}
    d q_t & = M^{-1}p_t \, dt,\\[6pt]
    d p_t & = -\nabla V (q_t)  \,dt -\gamma_P(q_t) M^{-1} p_t \, dt+ \sigma_P(q_t)  \,d W_t  + \nabla \xi(q_t) \, d \lambda_t, \\[6pt]
    \xi(q_t) & = z(t), \hspace{5cm} (C_q(t))
  \end{aligned} \right.
\]
We prove in particular the fluctuation identity~\eqref{eq:FK_multi_bis_corr}, 
see~\eqref{eq:FK_multi_bis_corr_Section4} below. Recall (see~\eqref{eq:gamma_P}) that, for simplicity, we assume in this section that the fluctuation-dissipation matrices are assumed to be of the form 
$(\sigma_P,\gamma_P) = ( P_M \, \sigma , P_M \, \gamma \, P_M^T)$ with $\gamma, \sigma \in \mathbb{R}^{3N \times 3N}$.
At variance with the previous sections, we do \emph{not} assume that $\gamma_P$ is strictly positive. Actually, $\gamma_P = 0$ corresponds to an interesting case: the deterministic 
Hamiltonian dynamics.

To our knowledge, the standard work fluctuations derived so far 
(except for our previous work~\cite{lelievre-rousset-stoltz-07-a}) apply only
to the case of time-dependent Hamiltonians. 
It is possible to consider transitions in the values of some reaction
coordinate in this framework upon resorting to 
steered molecular dynamics techniques. In this case, a penalty term 
$\varepsilon^{-1} (\xi(q)-z(t))^2$ (with $\varepsilon$ small) 
is used in the energy of the system 
to ``softly'' constrain the system to remain close to the submanifold 
$\Sigma(z(t)) = \{ q \in \R^{3N} \, | \, \xi(q) = z(t) \}$ at time~$t$. 
However, it is observed in practice that 
the statistical fluctuations increase with smaller~$\varepsilon$ (see~\cite{PKTS03}). 
We propose instead to replace the stiff constraining potential $\varepsilon^{-1} (\xi(q)-z)^2$ 
by a projection onto the submanifold $\Sigma(z)$. This is reminiscent of
the replacement of stiff constrained Langevin dynamics by rigidly constrained ones,
see Remark~\ref{rem:highosclang}.

This section is organized as follows. 
We first define the generalized free energy which is naturally computed 
with~$\NL$, and relate it to the standard 
free energy~\eqref{eq:Fbis} in Section~\ref{sec:Langti_FE}. Then, we give some precisions on the 
nonequilibrium dynamics~$\NL$ in Section~\ref{sec:generators_jarz}.
Next, we prove an appropriate version of the Jarzynski-Crooks fluctuation equality in 
Section~\ref{sec:jarz_lang_cons}. A numerical discretization of the nonequilibrium dynamics
is proposed in Section~\ref{sec:num_jarz}, together with various approximations
of the work. In particular, we propose a numerical strategy to obtain a 
Jarzynski-Crooks identity without time discretization error (see Section~\ref{sec:discrete_jarz}).
We then consider the overdamped limit when the mass matrix~$M$ 
goes to 0 (see Section~\ref{sec:ovd_limit_disc}).
Finally, in Section~\ref{sec:num_res_jarz}, we present some numerical results 
for the model system already considered in Section~\ref{sec:simple_example_WCA}.

\subsection{Generalized free energy}
\label{sec:Langti_FE}

For $(q_t,p_t)$ solution to the Langevin dynamics~$\NL$, the reaction coordinate evolution $\xi(q_t)=z(t)$ implies that $v_\xi(q_t,p_t) = \dot{z}(t)$, so that, at each time $t \geq 0$, the system $(q_t,p_t)$ belongs to the state space $\manq_{\xi,v_\xi}(z(t),\dot{z}(t))$. 
As a consequence, the free energy difference computed in this section by the Jarzynski relation  without correction (see~\eqref{eq:FK_multi_bis} below) is in fact the generalized rigid free energy $F_{\rm rgd}^{\Xi}$ defined in~\eqref{eq:F_xi}, in the special case $\Xi=(\xi,v_\xi)^T$:
\[
F^{\Xi}_{\rm rgd}(\zeta) =  -\frac{1}{\beta} \ln \int_{\manq_{\Xi}(\zeta)} {\rm e}^{-\beta H(q,p)} \sigma_{\manq_{\Xi}(\zeta)}(dq\, dp).
\] 
The latter free energy is associated to the normalization constant $Z_{z(t),\dot{z}(t)}$ of the distribution $\mu_{\manq_{\xi,v_\xi}(z(t),\dot{z}(t))}$ defined by~\eqref{eq:mu_v}. The generalized rigid free energy~\eqref{eq:F_xi} can be explicitly related to the usual free energy as follows. First, remark that, for a fixed~$q$,
\begin{align*}
  \begin{split}
    & \int_{\manq_{v_\xi(q,\cdot)}(v_z)} \exp \pare{-\frac{\beta}{2} p^T M^{-1}p } \sigma^{M^{-1}}_{\manq_{v_\xi(q,\cdot)}(v_z)}(dp) \\
&\qquad  = \exp\pare{-\frac{\beta}{2} v_z^T G_M^{-1}(q) v_z} \int_{T^*_q\manq(z)} \exp\left(-\frac{\beta}{2} p^T M^{-1}p \right) \sigma^{M^{-1}}_{T^*_q\manq(z)}(dp)
  \end{split}\\
&\qquad  = \pare{2 \pi \beta^{-1} }^{ \frac{3N-\nc}{2}} \exp\pare{-\frac{\beta}{2} v_z^T G_M^{-1}(q) v_z}.
\end{align*}
In the above, the change of variable $p\to p - \nabla \xi(q) G_M^{-1}(q) v_z$ has been used, in the space
\[
\manq_{v_\xi(q,\cdot)}(v_z) = \Big \{ p\in \R^{3N} \, \Big \vert \, \nabla \xi(q)^T M^{-1}\Big(p - \nabla \xi(q) G_M^{-1}(q) v_z \Big)=0 \Big \} .
\]
Note that $\frac{1}{2} v_z^T G_M^{-1}(q) v_z$ can be interpreted as the kinetic energy of the reaction coordinate~$\xi$.
Using the decomposition of measures~\eqref{eq:surfacemeas} and the above calculations, an alternative expression of the generalized free energy is:
\begin{equation}
  \label{eq:F_veff2}
   F^{\xi,v_\xi}_{\rm rgd}(z,v_z) = -\frac{1}{\beta} \ln \int_{\manq(z)} \exp\left(-\beta V(q) - \frac{\beta}{2} v_z^T G_M^{-1}(q) v_z \right) \sigma^{M}_{\manq(z)}(dq) + {\rm C },
\end{equation}
where, as usual, ${\rm C }$ denotes a generic constant (independent of $z$) whose value may vary from line to line.
As a consequence, the standard free energy~\eqref{eq:Fbis} is easily recovered from the generalized free energy, using relations similar to~\eqref{eq:FFtilde}. Indeed, using~\eqref{eq:F_veff2}, and with computations similar to the ones leading to~\eqref{eq:Fqp_final_difference}, the difference of the two free energies writes:
\begin{equation}
  \label{eq:FFtilde_v}
  \begin{aligned}
   & F(z) - F^{\xi,v_\xi}_{\rm rgd}(z,v_z) \\
   & \quad = \dps -\frac{1}{\beta} \ln \int_{\manq_{\xi,v_\xi}(z,v_z)} \hspace{-1cm} (\det G_M(q)) ^{-1/2} \exp\left(\frac{\beta}{2} v_z^T G_M^{-1}(q)v_z\right) \, \mu_{\manq_{\xi,v_\xi}(z,v_z)}(dq\,dp) + {\rm C}.
   \end{aligned}
\end{equation}
In practical nonequilibrium computations, the profile $t \mapsto F(z(t))$ can then be computed by adding a corrector to the work value in the Jarzynski estimator computing $F^{\xi,v_\xi}_{\rm rgd}(z(t),\dot{z}(t))$. This yields the identity~\eqref{eq:FK_multi_bis_corr_Section4} mentioned in the introduction and proved below (see the discussion after Theorem~\ref{th:crookslangcons}).

\subsection{Dynamics and generators}
\label{sec:generators_jarz}

The explicit expression of the Lagrange multipliers in~$\NL$ is obtained by a computation similar 
to~\eqref{eq:first_expression_lagrange} for the case without switching, by differentiating twice the constraints over time:
$$
\frac{d^2}{dt^2} \xi(q_t) = \ddot{z}(t).
$$
In view of the special structure of $(\sigma_P,\gamma_P)$, this leads to
\begin{align}
  d \lambda_t & =
  f_{\rm rgd}^M(q_t,p_t) \, dt   + G_M^{-1}(q_t) \ddot{z}(t) \, dt \nonumber\\
  & \quad + G_M^{-1} \nabla\xi(q_t)^TM^{-1}\pare{ \gamma_P(q_t) M^{-1}p_t \, dt  -\sigma_P(q_t) \, d W_t}  \nonumber \\
  & = f_{\rm rgd}^M(q_t,p_t) \, dt   + G_M^{-1}(q_t) \ddot{z}(t) \, dt.
  \label{eq:laglangjarz}
\end{align}
The expression \eqref{eq:laglangjarz} does not depend on the fluctuation-dissipation tensors $(\sigma_P,\gamma_P)$. This leads to simplified computations and motivates the special form of the latter matrices. The momentum evolution~$\NL$ thus simplifies as
\begin{equation}
  \label{eq:Langconsjarz2}
  \begin{aligned}
  d p_t & =  -  \nabla V(q_t) \, dt + \nabla \xi (q_t)  f_{\rm rgd}^M(q_t,p_t) \, dt  + \nabla \xi (q_t) G_M^{-1}(q_t) \ddot{z}(t) \, dt  \\
  & \quad -\gamma_P(q_t) M^{-1}p_t \, dt +\sigma_P(q_t) \, d W_t .
  \end{aligned}
\end{equation}

Let us denote by ${\mathcal L}^{\rm f}_{t}$ the generator of the forward dynamics $t \mapsto (q_t,p_t)$ defined in~$\NL$. The latter has a backward switching version,
\[
t' \mapsto (q^{\rm b}_{t'},p^{\rm b}_{t'}),
\]
obtained by using a time reversed switching $t' \mapsto  z(T-t')$, and by reversing
the momentum first in the initial condition, and then reversing them
back after the time evolution (see~\cite{CG08} for more general 
backward dynamics). More precisely, the backward dynamics 
can be defined through its generator
\begin{equation}
  \label{eq:backgen}
  {\mathcal L}_{t'}^{{\rm b}}= \mathcal{R} \, {\mathcal L}^{\rm f}_{T-t'} \, \mathcal{R}, 
\end{equation}
where ${\mathcal L}^{\rm f}_{T-t'}$ is the generator of the forward process at time $t=T-t'$, and $\mathcal{R} \, : \, \phi \mapsto \phi \circ S$ is the momentum flip operator 
with $S(q,p) = (q,-p)$. Thus $t' \mapsto (q^{\rm b}_{t'},-p^{\rm b}_{t'})$ is solution of the forward evolution equation~$\NL$ with a switching schedule $t' \mapsto  z(T-t')$. Therefore, the time evolution of the backward dynamics is given by  
\begin{equation}\label{eq:Langevinconstjarzback}
  \begin{cases}
    d q^{\rm b}_{t'} = -M^{-1} p^{\rm b}_{t'}  dt' , \\[6pt]
    d p^{\rm b}_{t'}= \nabla V(q^{\rm b}_{t'}) \, dt' -\gamma_P(q^{\rm b}_{t'}) M^{-1} p^{\rm b}_{t'} \, dt'+ \sigma_P(q^{\rm b}_{t'}) \,  d W_{t'}^{\rm b}  + \nabla \xi(q^{\rm b}_{t'}) \,  d \lambda_{t'}^{\rm b} , \\[6pt]
\xi(q^{\rm b}_{t'}) = z(T-t').
  \end{cases}
\end{equation}
In the following proposition, the expressions of ${\mathcal L}^{\rm f}_{t}$ 
and $ {\mathcal L}_{t'}^{{\rm b}}$ are explicitly written. 

\begin{proposition}
  \label{prop:generators_noneq}
  Consider $\zeta(t)=(z(t),\dot{z}(t))$.
  Then, the generator of the forward process~$\NL$ 
  at time $t \in [0,T]$ reads:
  \begin{equation}
    \label{eq:fwgensimple}
    \mathcal{L}_t^{\rm f} = \poisson{\, \cdot \, , H }_{\Xi}+ \mathcal{L}_{\Xi}^{\rm thm} + \poisson{\cdot, \Xi}\sgram^{-1}\dot{\zeta}(t) ,
  \end{equation}
  and the generator of the backward process~\eqref{eq:Langevinconstjarzback} at time $t' \in [0,T]$ reads:
  \begin{equation}
    \label{eq:bwgensimple}
    \mathcal{L}^{\rm b}_{t' }=  -\poisson{\, \cdot \, , H }_{\Xi} + \mathcal{L}_{\Xi}^{\rm thm} - \poisson{\cdot, \Xi}   \sgram^{-1} \dot{\zeta}(T-t'),
  \end{equation}
  where
  \[
  \mathcal{L}_{\Xi}^{\rm thm} = \frac{1}{\beta} \, \op{e}^{\beta H}
  \op{div}_{p}\Big( \rme^{-\beta H} \, \gamma_P \,
  \nabla_{p} \cdot\Big)
  \]
  is the fluctuation-dissipation operator defined in~\eqref{eq:conslangevinop_thm}.
\end{proposition}

\begin{proof}
First, let us consider the terms in~$\NL$ arising from the Hamiltonian evolution and from the switching ({\it i.e.} without fluctuation-dissipation, which amounts to setting $\gamma_P=0$ and $\sigma_P =0$ in~$\NL$).
Since during this dynamics $v_\xi(q_t,p_t)=\dot{z}(t)$, \eqref{eq:poisson_Xi_H} yields:
\begin{align*}
\poisson{\Xi,H}(q_t,p_t)
&=\bmat \dot{z}(t) \\   \op{Hess}_{q_t}(\xi)(M^{-1}p_t,M^{-1}p_t)  - \pare{\nabla \xi ^T M^{-1} \nabla V }(q_t)   \emat,
\end{align*}
so that, using~\eqref{eq:sgraminv},
\begin{equation}\label{eq:inter}
\sgram^{-1}(q_t,p_t)\pare{\poisson{\Xi,H}(q_t,p_t) - \dot{\zeta}(t) }  = 
\bmat G_M^{-1}(q_t) \ddot{z}(t) +  f_{\rm rgd}^M(q_t,p_t)  \\ 0 \emat.
\end{equation}
With~\eqref{eq:a}, we then obtain:
\begin{equation}\label{eq:inter2}
  \poisson{\ph, \Xi} \sgram^{-1}\pare{\dot{\zeta}(t)-\poisson{\Xi,H} } (q_t,p_t) =
  \pare{ G_M^{-1}(q_t) \ddot{z}(t) +  f_{\rm rgd}^M(q_t,p_t)}^T
  \nabla \xi(q_t)^T\nabla_p \ph(q_t,p_t) .
\end{equation}
Now, the Hamiltonian part of the switched dynamics~$\NL$ (see also~\eqref{eq:Langconsjarz2}) can be recognized in~\eqref{eq:inter2}, so that the generator $ \mathcal{L}_t^{\rm f}$ when $(\gamma_P,\sigma_P)=(0,0)$ reads: for any smooth test function~$\ph$,
\begin{align}\label{eq:genswitched}
    \mathcal{L}^{\rm f}_t(\ph)=&\pare{ \nabla \xi f_{\rm rgd}^M +\nabla\xi  G_M^{-1} \ddot{z}(t) }^T  \nabla_p \ph \nonumber  -\pare{\nabla V }^T\nabla_p \ph + p^TM^{-1} \nabla_q \ph
\nonumber\\
  =& \poisson{\ph, \Xi}\sgram^{-1}(\dot{\zeta}(t) -\poisson{\Xi,H}) +\poisson{\ph , H}  \nonumber  \\
  =& \poisson{\ph , H}_{\Xi} + \poisson{\ph, \Xi}\sgram^{-1}\dot{\zeta}(t).
\end{align}
The full expression of the generator $\mathcal{L}_t^{\rm f}$ is then obtained by adding the
terms arising from the fluctuation-dissipation. 
These terms are directly obtained from the terms involving~$\gamma_P$ and~$\sigma_P$ 
in~\eqref{eq:Langconsjarz2}, as in the proof of Proposition~\ref{s:p:Lconstgen}.

The generator of the backward switching process given by~\eqref{eq:Langevinconstjarzback} can be obtained from similar computations. First, 
the thermostat parts in~\eqref{eq:Langevinconstjarzback} and 
in~$\NL$ are the same.
Consider now the Hamiltonian part (obtained by taking $(\gamma_P,\sigma_P)=(0,0)$) in the dynamics~\eqref{eq:Langevinconstjarzback}.  By definition of the backward dynamics, and the expression~\eqref{eq:genswitched} of the forward dynamics, the Hamiltonian part reads
\begin{align*}
  \mathcal{L}^{\rm b}_{t'}(\ph)(q,p) & = 
  \mathcal{R}\mathcal{L}^{\rm f}_{T-t'} \pare{\mathcal{R}(\ph)}(q,p) \nonumber  \\
  & = \pare{ \nabla \xi(q) f_{\rm rgd}^M(q,p) +\nabla\xi(q)  G_M^{-1}(q) \ddot{z}(T-t') }^T  
  \pare{-\nabla_p \ph} \\
  & \quad -\nabla V(q) ^T\pare{ - \nabla_p \ph} - p^TM^{-1} \nabla_q \ph,
\end{align*}
so that 
\begin{align*}
\mathcal{L}^{\rm b}_{t'}\ph &= - \mathcal{L}^{\rm f}_{T-t'}\ph
 = -\poisson{\ph , H}_{\Xi} - \poisson{\ph, \Xi}\sgram^{-1}\dot{\zeta}(T-t').
\end{align*}
This gives~\eqref{eq:bwgensimple}.
\end{proof}

\subsection{Jarzynski-Crooks identity}
\label{sec:jarz_lang_cons}

Before stating the main result of this section (Theorem~\ref{th:crookslangcons} below),
we need to introduce a notion of work. This quantity is most conveniently defined
for deterministic dynamics, but the corresponding definition is also valid for 
stochastic dynamics.

We define the work $\pare{\W_t}_{t \geq 0}$ associated with the constraining force~$\nabla \xi(q_t) \, d \lambda_t$ in~$\NL$ as the physical displacement multiplied by the force:
\begin{align}
  d \W_t& : =   \pare{\frac{d q_t}{dt}}^T  
  \circ \Big( \nabla \xi (q_t) \, d \lambda_t \Big) =  \pare{\frac{d q_t}{dt}}^T  \nabla \xi (q_t) \,\circ  d \lambda_t 
  = \dot{z}^T(t) \circ  d \lambda_t \nonumber \\
  &  =  \dot{z}^T(t) \,  d \lambda_t. \label{eq:work} 
\end{align}
By convention, $\W_0 = 0$. In the above computations, we used 
successively the fact that $t\mapsto \xi(q_t)$, 
and then $t \mapsto z(t)$ are differentiable processes, 
so that Stratonovitch and It\^o integrations are equivalent.
Let us introduce the deterministic version of the nonequilibrium 
process~$\NL$ (\textit{i.e.} $(\gamma_P,\sigma_P)=(0,0)$):
\begin{equation}\label{eq:Hamjarz}
\left\{
  \begin{aligned}
    d \tilde{q}_t & = M^{-1}\tilde{p}_t \, dt,\\[6pt]
    d \tilde{p}_t & = -\nabla V (\tilde{q}_t)  \,dt + \nabla \xi(\tilde{q}_t) \, 
    d \tilde{\lambda}_t, \\[6pt]
    \xi(\tilde{q}_t) & = z(t), \hspace{5cm} (C_q(t))
  \end{aligned} \right.
\end{equation}
and denote by $\Phi_{t,t+h}: \manq_{\xi,v_\xi}(z(t), \dot{z}(t)) \to \manq_{\xi,v_\xi}(z(t+h), \dot{z}(t+h)) $ the associated flow between time $t \in [0,T]$ and $t+h \in [0,T]$.
The work can now be written out more explicitly using the following lemma:

\begin{lemma} 
  The infinitesimal variation of the work~\eqref{eq:work} reads:
  $$  d \W_t =w(t,q_t,p_t) \, dt, $$
  where for all $t \in [0,T]$ and all $(q,p) \in \manq_{\xi,v_\xi}(z(t),\dot{z}(t))$,
  \begin{align}
    w(t,q,p) &=  \dot{\zeta}(t) ^T\sgram^{-1} \poisson{\Xi,H}(q,p)   
    \label{eq:workexpr1}\\
    &= \dot{z}(t)^T\pare{  G_M^{-1}(q) \ddot{z}(t) + f_{\rm rgd}^M(q,p) }   \label{eq:workexpr2} \\
    & = \left.\pare{ \frac{d}{dh} H \circ \Phi_{t,t+h} }\right|_{h=0}(q,p) .   \label{eq:workexpr3}
  \end{align}
  The total exchanged work is then a time integral associated with the 
  path $t \mapsto (q_t,p_t)$, and is denoted by:
  \[
  \W_{0,T}\pare{ \set{q_t,p_t}_{0 \leq t \leq T}} = \W_T - \W_0 = \int_{0}^{T} w(t,q_t,p_t)\, dt.
  \]
\end{lemma}

Note that the expression~\eqref{eq:workexpr3} can be interpreted as the energy variation of the system during the switching when the stochastic thermostat is turned off.

\begin{proof} 
The expression of the Lagrange multipliers in~\eqref{eq:laglangjarz} 
yields~\eqref{eq:workexpr2}:
\[
\dot{z}(t)^T d \lambda_t = \dot{z}(t)^T\pare{ G_M^{-1}(q_t) \ddot{z}(t) + 
  f_{\rm rgd}^M(q_t,p_t) } \, dt.
\]
Moreover, \eqref{eq:inter} gives:
\begin{align*}
  \dot{z}(t)^T \pare{  G_M^{-1}(q_t) \ddot{z}(t) + f_{\rm rgd}^M(q_t,p_t) } \, dt  &  =  
  \dot{\zeta}(t)^T \sgram^{-1}(q_t,p_t)\pare{\poisson{\Xi,H}(q_t,p_t) -\dot{\zeta}(t)} \\
  & =  \dot{\zeta}(t)^T \sgram^{-1} \poisson{\Xi,H}(q_t,p_t), 
\end{align*}
where in the last line we have used $\dot{\zeta}(t)^T \sgram^{-1} \dot{\zeta}(t)= 0$. This gives~\eqref{eq:workexpr1}.
To prove~\eqref{eq:workexpr3}, we compute the variations of the energy 
$H(\tilde{q}_t,\tilde{p}_t)$ for $(\tilde{q}_t,\tilde{p}_t)$ solution of~\eqref{eq:Hamjarz} 
with initial condition~$(q,p)$:
\begin{align*}
  d H(\tilde{q}_t,\tilde{p}_t) & = \tilde{p}_t^TM^{-1} \, d \tilde{p}_t +  \tilde{p}_t^TM^{-1} \nabla V(\tilde{q}_t) \, dt \\
& = \dot{z}(t) ^ T \, d \tilde{\lambda}_t   \\
& = \dot{z}(t)^T\pare{ G_M^{-1}(\tilde{q}_t) \ddot{z}(t) + f_{\rm rgd}^M(\tilde{q}_t,\tilde{p}_t) } \, dt.
\end{align*}
The last equality is obtained using the computation of the Lagrange multipliers in~\eqref{eq:laglangjarz}. This yields~\eqref{eq:workexpr3}.
\end{proof}

We are now in position to state the main result of this section.

\begin{theorem}[Jarzynski-Crooks fluctuation identity]
  \label{th:crookslangcons}
  Consider the normalization $Z_{z(t),\dot{z}(t)}$ for the 
  canonical distribution $\mu_{\manq_{\xi,v_\xi}(z(t),\dot{z}(t))}$ defined in~\eqref{eq:mu_v}.
  Denote by $\{q_t,p_t\}_{0 \leq t \leq T}$ the solution of the forward 
  Langevin dynamics~$\NL$ with initial conditions
  distributed according to
  \begin{equation}\label{eq:ICf}
  (q_0,p_0) \sim \mu_{\manq_{\xi,v_\xi}(z(0),\dot{z}(0))}(dq\, dp ),
  \end{equation}
  and by $\{q^{\rm b}_{t'},p^{\rm b}_{t'}\}_{0 \leq t' \leq T }$ the solution of the backward 
  Langevin process~\eqref{eq:Langevinconstjarzback} with initial conditions
  distributed according to
  \begin{equation}\label{eq:ICb}
  (q^{\rm b}_0,p^{\rm b}_0) \sim \mu_{\manq_{\xi,v_\xi}(z(T),\dot{z}(T))}(dq\, dp ).
  \end{equation}
  Then, the following Jarzynski-Crooks identity holds on~$[0,T]$: for any bounded 
  path functional~$\ph_{[0,T]}$, 
  \begin{equation}
    \label{eq:Jarz_Crooks}
    \frac{Z_{z(T),\dot{z}(T)}}{Z_{z(0),\dot{z}(0)}} 
    = \frac{ 
    \E \pare{ \ph_{[0,T]}\left(\{q_t,p_t\}_{0 \leq t \leq T}\right) \,  
      {\rm e}^{-\beta\W_{0,T}\left(\set{q_t,p_t}_{t\in [0,T]}\right) } } }
    {\E \pare{ \ph^{\rm r}_{[0,T]}\left(\{q^{\rm b}_{t'},p^{\rm b}_{t'}\}_{0 \leq t' \leq T }\right) }},
  \end{equation}
  where $(\, \cdot \,)^{\rm r}$ denotes the composition with 
  the operation of time reversal of paths:
  \begin{equation}
    \label{eq:time_reversal_on_paths}
    \ph_{[0,T]}^{\rm r}\Big(\{ q^{\rm b}_{t'},p^{\rm b}_{t'}\}_{0 \leq t' \leq T}\Big) = 
    \ph_{[0,T]}\Big(\{ q^{\rm b}_{T-t},p^{\rm b}_{T-t} \}_{0 \leq t \leq T}\Big).
  \end{equation} 
\end{theorem}

Note that the theorem still holds in the Hamiltonian case, 
\emph{i.e.} when $(\gamma_P,\sigma_P) =(0,0)$. 
The choice $\ph_{[0,T]} = 1$ in~\eqref{eq:Jarz_Crooks} 
leads to the following work fluctuation identity:
\begin{equation}
\label{eq:FK_multi_bis}
F^{\xi,v_\xi}_{\rm rgd}(z(T),\dot{z}(T))-F^{\xi,v_\xi}_{\rm rgd}(z(0),\dot{z}(0)) = 
-\frac1\beta \ln\Big[ \E\left(\rme^{-\beta \W_{0,T}(\set{q_t,p_t}_{t\in [0,T]}) }\right) \Big].
\end{equation}
Besides, upon choosing a path functional $\exp(\theta \beta\W_{0,T})$,
it is possible to obtain a family of free energy estimators, parameterized by $\theta$ 
and where both forward and backward paths are weighted by the exponential
of some work. Moreover, the standard Crooks equality on ratios of probability density functions
of work values is also a consequence of~\eqref{eq:Jarz_Crooks}, see 
Section~4.2.2 in~\cite{LelRouStoBook}.

Note also that the choice $\ph_{[0,T]}(q,p) = \phi(q_T,p_T)$ leads to the following 
representation of the canonical distribution~$\mu_{\manq_{\xi,v_\xi}(z(T),\dot{z}(T))}$:
\begin{equation}
\label{eq:FK_multi_dist}
\frac{\E \pare{ \phi(q_T,p_T) {\rm e}^{-\beta\W_{0,T}\pare{\set{q_t,p_t}_{t\in [0,T]}}}  }}{\E \pare{ {\rm e}^{-\beta\W_{0,T}\pare{\set{q_t,p_t}_{t\in [0,T]}}}  }}  = \int_{\manq_{\xi,v_\xi}(z(T),\dot{z}(T)) } \phi(q,p) \,  \mu_{\manq_{\xi,v_\xi}(z(T),\dot{z}(T))}(dq\, dp ).
\end{equation}
The usual free energy profile $z \mapsto F(z)$ can therefore be computed using the relations~\eqref{eq:corr}-\eqref{eq:FK_multi_bis_corr} presented in the introduction. Indeed, Equation~\eqref{eq:FK_multi_bis_corr} (see~\eqref{eq:FK_multi_bis_corr_Section4} below) can be proved by combining~\eqref{eq:FFtilde_v} and~\eqref{eq:FK_multi_bis}--\eqref{eq:FK_multi_dist} as follows:
\begin{align}
F(z(T))-F(z(0)) &= \pare {F(z(T)) - F^{\xi,v_\xi}_{\rm rgd}(z(T),\dot{z}(T)) } - \pare{F(z(0))- F^{\xi,v_\xi}_{\rm rgd}(z(0),\dot{z}(0))} \nonumber \\
& \quad -\frac1\beta \ln\Big[ \E\left(\rme^{-\beta \W_{0,T}(\set{q_t,p_t}_{t\in [0,T]}) }\right) \Big] \nonumber \\
&=  -\frac1\beta \ln \E \pare{  (\det G_M(q_T)) ^{-1/2} {\rm e}^{\frac{\beta}{2} \dot{z}(T)^T G_M^{-1}(q_T) \dot{z}(T)}\, {\rm e}^{-\beta\W_{0,T}\pare{\set{q_t,p_t}_{t\in [0,T]}}}  } \nonumber \\
&\quad + \frac1\beta \ln \E \pare{  (\det G_M(q_0)) ^{-1/2} {\rm e}^{\frac{\beta}{2} \dot{z}(0)^T G_M^{-1}(q_0)\dot{z}(0)} } \nonumber \\
& =  -\frac1\beta \ln \left ( \frac{ \E\left(\rme^{-\beta \left[ \W_{0,T}\pare{\set{q_t,p_t}_{t\in [0,T]}} + C(T,q_T)   \right]  }\right) } {  \E\left(\rme^{-\beta  C(0,q_0)   }  \right) } \right ),
\label{eq:identity_to_approximate}
\end{align}
where the corrector $C(t,q)$ is defined in~\eqref{eq:corr}:
\[
C(t,q) = \frac{1}{2\beta} \ln \Big( \det G_M(q) \Big) - 
\frac{1}{2} \dot{z}(t)^T G_M^{-1}(q) \dot{z}(t).
\]
This leads to the following relation:
\begin{equation}
\label{eq:FK_multi_bis_corr_Section4}
F(z(T)) - F(z(0)) = -\frac1\beta \ln  \left( \frac{\E\left(\rme^{-\beta \left[\W_{0,T}\pare{ \set{q_t,p_t}_{0 \leq t \leq T} }  +C(T,q_T)   \right]  }\right)}{\E\left(\rme^{-\beta C(0,q_0)   }\right)} \right),
\end{equation}
Estimators of the free energy based on~\eqref{eq:FK_multi_bis_corr} 
can then be constructed, see Chapter~4 in~\cite{LelRouStoBook} for a review.

\medskip

Before turning to the proof of Theorem~\ref{th:crookslangcons}, we first give the general lemma which enables to deduce the Jarzynski-Crooks 
fluctuation identity from a {\em nonequilibrium detailed balance condition} (similar to the one presented in~\cite{CG08} for switchings arising from a time-dependence in the Hamiltonian). 
\begin{lemma}\label{lem:jarz_gen} 
Let $(q_t,p_t)_{0 \le t \le T}$ (resp. $(q^{\rm b}_t,p^{\rm b}_t)_{0 \le t \le T}$) be a Markov process with infinitesimal generator ${ \mathcal L}^{\rm f} _t$ (resp. ${ \mathcal L}^{\rm b} _t$) and initial conditions distributed according to~\eqref{eq:ICf} (resp.~\eqref{eq:ICb}). Let us assume that the following  nonequilibrium detailed balance condition is satisfied: 
  for any two smooth test functions $\ph_1$, $\ph_2$,
  \begin{equation}
    \label{eq:jarzbalancecons}
    \begin{split}
      & \int_{\manq_{\xi,v_\xi}(z(t),\dot{z}(t))} \pare{   \ph_1{ \mathcal L}^{\rm f} _t(\ph_2) - \ph_2 {\mathcal L}^{\rm b}_{T-t}(\ph_1)   }  {\rm e }^{- \beta H} \,d \sigma_{\manq_{\xi,v_\xi}(z(t),\dot{z}(t))}  \\
      &\quad =  \int_{\manq_{\xi,v_\xi}(z(t),\dot{z}(t))} \beta w^{}(t,\cdot) \ph_1\ph_2 {\rm e }^{- \beta H} \,d \sigma_{\manq_{\xi,v_\xi}(z(t),\dot{z}(t))}  \\
      & \quad \quad + \frac{d}{dt} \pare{ \int_{\manq_{\xi,v_\xi}(z(t),\dot{z}(t))} \ph_1 \ph_2 {\rm e }^{- \beta H} \,d \sigma_{\manq_{\xi,v_\xi}(z(t),\dot{z}(t))} }.
  \end{split}
  \end{equation}
  Then the Jarzynski-Crooks fluctuation identity~\eqref{eq:Jarz_Crooks} holds.
\end{lemma}

\begin{proof}
We use in this proof the short-hand notation $Z_t$, $\pi_t$ and $\mS_t$ for the 
partition function $Z_{z(t),\dot{z}(t)}$, the (unnormalized) distribution 
$Z_{z(t),\dot{z}(t)} \mu_{\manq_{\xi,v_\xi}(z(t),\dot{z}(t))} = \mathrm{e}^{-\beta H} 
\sigma_{\manq_{\xi,v_\xi}(z(t),\dot{z}(t))}$ 
and the submanifold $\manq_{\xi,v_\xi}(z(t),\dot{z}(t))$ respectively. 

Let us introduce the following weighted transition operators: for any bounded test function~$\ph$,
  \begin{align}
    P_{t,T}^{\rm f}(\ph)(q,p)  &= 
    \E\Big(\ph(q_T,p_T) \, {\rm e}^{-\beta \W_{t,T}\pare{\set{q_s,p_s}_{s\in [t,T]}  } } \, \Big | \, (q_t,p_t) = (q,p) \Big),
    \label{eq:Pfwd}\\ 
    P^{\rm b}_{t',T}(\ph)(q,p) & =  
    \E\Big(\ph(q^{\rm b}_{T},p^{\rm b}_{T})\, \Big | \, (q^{\rm b}_{t'},p^{\rm b}_{t'})=(q,p) \Big), \label{eq:Pbwd}
\end{align}
where $(q_t,p_t)_{0 \le t \le T}$ (resp. $(q^{\rm b}_t,p^{\rm b}_t)_{0 \le t \le T}$) is a Markov process with infinitesimal generator ${ \mathcal L}^{\rm f} _t$ (resp. ${ \mathcal L}^{\rm b} _t$), and
$\W_{t,T} = \W_{0,T} - \W_{0,t}$.

We assume that these operators are well defined and smooth with respect to time for sufficiently smooth test functions defined in an open neighborhood of $\manq_{\xi,v_\xi}(z(t),\dot{z}(t))$  
  and $\manq_{\xi,v_\xi}(z(t'),\dot{z}(t'))$ respectively (for any $t,t' \in [0,T]$). 

The transition operators satisfy the following 
backward Kolmogorov evolution equations:
\[
\begin{cases}
 \partial_t P^{\rm f}_{t,T} = -{\mathcal L}_t^{\rm f} P^{\rm f}_{t,T} + \beta \,
w(t,\cdot) \, P_{t,T}^{\rm f}, \\[6pt]
P^{\rm f}_{T,T} = {\rm Id}, 
\end{cases}
\quad
\begin{cases}
\partial_{t'} P^{\rm b}_{t',T} = -{\mathcal L}^{\rm b}_{t'} 
P^{\rm b}_{t',T}, \\[6pt]
P^{\rm b}_{T,T} = {\rm Id}. 
\end{cases}
\]
Consider now two test functions $\ph_0$ and $\ph_T$. The balance 
condition~\eqref{eq:jarzbalancecons} implies
\[
\frac{d}{dt} \pare{ \int_{\mS_t} P^{\rm f}_{t,T}(\ph_T) \, P^{\rm b}_{T-t,T}(\ph_0) \, d\pi_t }=0.
\]
Integrating this equality on $[0,T]$ yields 
\begin{equation}
\label{eq:two_pt_Crooks}
\int_{\mS_0} P_{0,T}^{\rm f}(\ph_T) \, \ph_0 \, d\pi_0 = 
\int_{\mS_T} \ph_T \, P^{\rm b}_{0,T}(\ph_0) \, d\pi_T,
\end{equation}
which is the Crooks identity~\eqref{eq:Jarz_Crooks} for path functionals of the form 
\[
\ph_{[0,T]}(q,p) = \ph_0(q_0,p_0) \, \ph_T(q_T,p_T). 
\]
Indeed,
\[
\int_{\mS_0} P_{0,T}^{\rm f}(\ph_T) \, \ph_0 \, d\pi_0 = Z_0 \, 
\mathbb{E} \Big [  \ph_T(q_T,p_T) \, \ph_0(q_0,p_0) \, 
{\rm e}^{-\beta \W_{0,T}(\set{q_s,p_s}_{s\in[0,T]})  } \Big ], 
\]
while
\[
\int_{\mS_T} \ph_T \, P^{\rm b}_{0,T}(\ph_0) \, d\pi_T = Z_T \, 
\mathbb{E} \Big [ \ph_T(q^{\rm b}_0,p^{\rm b}_0) \, \ph_0(q_T^{\rm b},p_T^{\rm b}) \Big ].
\]
Then, using the Markov property of the forward and backward processes, 
Crooks identity~\eqref{eq:Jarz_Crooks} can be extended to 
finite-dimensional path functionals of the form:
\begin{equation}
\label{eq:finite_d_paths}
\ph_{[0,T]}(q,p) = \ph_0(q_{0},p_0) \dots \ph_k(q_{t_k},p_{t_k}) \dots \ph_{K}(q_{T},p_T)
\end{equation}
with $0 = t_0 < t_1 \cdots < t_K = T$ by 
repeatedly using a variant of~\eqref{eq:two_pt_Crooks} on time subintervals $[t_k,t_{k+1}]$
(see the proof of Theorem~4.10 in \cite{LelRouStoBook} for further precisions).
This allows to conclude since finite dimensional time marginal laws characterize
the distribution on continuous paths, see for instance~\cite{EthKur86}.
\end{proof}

We are now in position to write the 

\begin{proof}[Proof of Theorem~\ref{th:crookslangcons}]
By Lemma~\ref{lem:jarz_gen}, it is sufficient to prove the nonequilibrium detailed balance~\eqref{eq:jarzbalancecons} for the Markov processes $(q_t,p_t)_{0 \le t \le T}$ and $(q^{\rm b}_t,p^{\rm b}_t)_{0 \le t \le T}$, solutions to~$\NL$ and~\eqref{eq:Langevinconstjarzback} respectively, with generators~$\mathcal{L}_t^{\rm f}$ and~$\mathcal{L}^{\rm b}_{t' }$ defined by~\eqref{eq:fwgensimple} and~\eqref{eq:bwgensimple}.

 First, using Lemma~\ref{l:consint}, we compute  the variation of the unnormalized canonical equilibrium distribution with constraints with respect to the switching:
\begin{align}
  &\frac{d}{dt}  \pare{ \int_{\manq_{\xi,v_\xi}(z(t),\dot{z}(t))} \!\! \ph_1 \ph_2 \, {\rm e }^{- \beta H} \, d \sigma_{\manq_{\xi,v_\xi}(z(t),\dot{z}(t)) }} \nonumber \\ 
  & \qquad = \int_{\manq_{\xi,v_\xi}(z(t),\dot{z}(t))} \!\! \dot{\zeta}(t)^T \sgram^{-1} \poisson{\Xi ,  \ph_1 \ph_2\,  {\rm e }^{- \beta H}  } \,d \sigma_{\manq_{\xi,v_\xi}(z(t),\dot{z}(t))}.
  \label{eq:varjarzmeas}
\end{align}
On the other hand, \eqref{eq:workexpr1} and Proposition~\ref{prop:generators_noneq} give
\begin{equation}
  \label{eq:jarzbalance_pre}
  \begin{split}
  &\dps \ph_1 { \mathcal L}^{\rm f}_t(\ph_2) - \ph_2 {\mathcal L}^{\rm b}_{T-t}(\ph_1) - \beta w^{}(t,\cdot) \ph_1\ph_2 \\
&= \poisson{ \ph_1\ph_2, H }_{\Xi} +  {\rm e }^{\beta H}  \poisson{\ph_1\ph_2 {\rm e }^{- \beta H },\Xi} \sgram^{-1} \dot{\zeta}(t) \\
&\dps   \quad + \ph_1 \frac{1}{\beta}\op{e}^{\beta H}  \op{div}_{p}\pare{ \op{e}^{-\beta H} \gamma_P \nabla_{p} \ph_2}
- \ph_2 \frac{1}{\beta}\op{e}^{\beta H}  \op{div}_{p}\pare{ \op{e}^{-\beta H}  \gamma_P \nabla_{p} \ph_1} .
\end{split}
\end{equation}
Now, \eqref{eq:jarzbalancecons} can be verified in two steps. 
First, the last two terms in~\eqref{eq:jarzbalance_pre} (the ''thermostat'' terms) cancel 
out after integration with respect to 
$\op{e}^{-\beta H} \, d \sigma_{\manq_{\xi,v_\xi}(z(t),\dot{z}(t))}$ thanks to 
the detailed balance condition~\eqref{eq:balance_momenta_rev}. 
Then, an integration of~\eqref{eq:jarzbalance_pre} with respect to 
$\rme^{-\beta H} \, d \sigma_{\manq_{\xi,v_\xi}(z(t),\dot{z}(t))}$ gives, 
in view of~\eqref{eq:varjarzmeas} and~\eqref{eq:divsympl},
\begin{align*}
&\int_{\manq_{\xi,v_\xi}(z(t),\dot{z}(t))}  \Big( \ph_1 { \mathcal L}^{\rm f}_t(\ph_2) - \ph_2 {\mathcal L}^{\rm b}_{T-t}(\ph_1) - \beta w(t,\cdot) \ph_1\ph_2 \Big){\rm e }^{- \beta H} d \sigma_{\manq_{\xi,v_\xi}(z(t),\dot{z}(t)) }\\
& \qquad = \int_{\manq_{\xi,v_\xi}(z(t),\dot{z}(t))}\poisson{\ph_1\ph_2 {\rm e }^{- \beta H },\Xi} \sgram^{-1} \dot{\zeta}(t) \, d \sigma_{\manq_{\xi,v_\xi}(z(t),\dot{z}(t)) } \\
& \qquad = \frac{d}{dt}  \pare{ \int_{\manq_{\xi,v_\xi}(z(t),\dot{z}(t))} \ph_1 \ph_2 \, {\rm e }^{- \beta H} \, d \sigma_{\manq_{\xi,v_\xi}(z(t),\dot{z}(t)) }},
\end{align*}
which is indeed~\eqref{eq:jarzbalancecons}.  Note that the time-regularity on the evolution semi-groups~\eqref{eq:Pfwd}-\eqref{eq:Pbwd} required to make these computations rigorous is proved in the overdamped case in the proof of Theorem~A.5 in \cite{lelievre-rousset-stoltz-07-a}. A similar proof can be carried out for constrained Langevin equations.
\end{proof}

\subsection{Numerical schemes}
\label{sec:num_jarz}

In this section, a numerical scheme for the nonequilibrium dynamics~$\NL$ and the associated free energy estimator are presented. As for Langevin processes with constraints (Section~\ref{sec:consnum}), a splitting between the Hamiltonian part and the thermostat part of the dynamics~$\NL$ leads to a simple and natural scheme (see~\eqref{eq:flucdissconstjarz1}-\eqref{eq:Verletconstswitched}-\eqref{eq:flucdissconstjarz2} below). Note that a consistent numerical scheme in the case of Hamiltonian dynamics can be obtained by considering only~\eqref{eq:Verletconstswitched} (this corresponds to $\gamma = \sigma = 0$). 
Besides, we propose a discrete Jarzynski-Crooks identity without 
time discretization error, see Section~\ref{sec:discrete_jarz}.

The reaction coordinate path is first discretized as $\{ z(0), \dots, z(t_{N_T}) \}$ where $N_T$ is the number of time-steps. For simplicity, equal time increments are used, so that $\dt=\frac{T}{N_T}$ and $t_n =n \dt$. 
The deterministic Hamiltonian part in the equations of
motion~$\NL$ with switched position constraints
$\xi(q) = z(t)$ can be integrated by a velocity-Verlet algorithm with
constraints similar to~\eqref{eq:Verletconst}. The
fluctuation-dissipation term in~$\NL$ can be
integrated similarly to the constrained case without
switching~\eqref{eq:flucdiss1}-\eqref{eq:flucdiss2}, using an
Ornstein-Uhlenbeck process on the momentum variable approximated by a
midpoint Euler scheme. In conclusion, the splitting scheme for the Langevin
dynamics with time-evolving constraints reads as follows: Take initial conditions $(q^0,p^0)$ distributed according to $\mu_{\manq_{\xi,v_\xi}\left(z(t_0),\frac{z(t_1)-z(t_0)}{\dt}\right)}$ and iterate on $0 \leq n \leq N_T-1$:
\begin{align}
  &  \begin{cases}
       \dps    p^{n+1/4} =   p^{n} -\frac{\dt}{4} \gamma_P(q^n) M^{-1}(p^{n+1/4}+p^{n}) 
       + \sqrt{ \frac{\dt}{2}} \sigma_P(q^n) {\mathcal G}^n, 
     \end{cases}\label{eq:flucdissconstjarz1}\\[6pt]
  &\begin{cases}
    \dps  p^{n+1/2} = p^{n+1/4} - \displaystyle{\frac{\dt}{2} \nabla V (q^{n})} + \nabla \xi(q^n) \lambda^{n+1/2}, &\\[6pt]
    \dps  q^{n+1} = q^{n} + \dt \ M^{-1} p^{n+1/2},  &\\[6pt]
    \dps   \xi(q^{n+1})  = z(t_{n+1}), &(C_q) \\[6pt]
    \dps  p^{n+3/4} = p^{n+1/2} - \displaystyle{\frac{\dt}{2} \nabla V (q^{n+1})} +\nabla \xi(q^{n+1}) \lambda^{n+3/4},&\\[6pt]
    \dps   \nabla \xi (q^{n+1})^{T} M^{-1} p^{n+3/4} = \frac{ \dps  z(t_{n+2})- z(t_{n+1})}{\dps \dt},
    &(C_p)
   \end{cases}\label{eq:Verletconstswitched}\\
  &   \begin{cases}  \dps    p^{n+1} = p^{n+3/4} -\frac{\dt}{4} \gamma_P(q^{n+1}) M^{-1}(p^{n+3/4}+p^{n+1}) \\
        \dps \phantom{p^{n+1} =}  + \sqrt{\frac{\dt}{2}} \sigma_P(q^{n+1}) {\mathcal G}^{n+1/2}, 
      \end{cases}\label{eq:flucdissconstjarz2}
\end{align}
where $({\mathcal G}^n)$ and $({\mathcal G}^{n+1/2})$ are sequences of i.i.d. Gaussian random variables of mean~$0$ and covariance matrix~$\mathrm{Id}_{3N}$.
Note that the momenta obtained from~\eqref{eq:flucdissconstjarz1}-\eqref{eq:Verletconstswitched}-\eqref{eq:flucdissconstjarz2} satisfy
\begin{equation}\label{eq:const_satisf}
\nabla\xi(q^n)^T M^{-1}p^{n+1/4} = \nabla\xi(q^n)^T M^{-1}p^n = \nabla\xi(q^n)^T M^{-1}p^{n-1/4} 
= \frac{ z(t_{n+1})- z(t_{n})}{\dt}, 
\end{equation}
so that constraints on momenta are automatically enforced, and no Lagrange multiplier is needed 
in~\eqref{eq:flucdissconstjarz1} and~\eqref{eq:flucdissconstjarz2}.

We comment in the subsequent sections on the different parts of the scheme.

\subsubsection{Comments on the Hamiltonian scheme~\eqref{eq:Verletconstswitched}}
\label{sec:numHamiltonconstswitched}
The Lagrange multipliers~$\lambda^{n+1/2}$ are associated with the position constraints $(C_{q})$, and the Lagrange multipliers $\lambda^{n+3/4}$ are associated with the  velocity constraints $(C_{p})$. In $(C_p)$, the velocity of the switching at time $t_{n+1}$ is discretized as:
\[
\dot{z}(t_{n+1}) \simeq \frac{ z(t_{n+2})- z(t_{n+1})}{\dt}.
\]
The latter choice is motivated by the following observation: The position after one step of an unconstrained motion, given by 
\[
\tilde{q}^{n+1}= q^n+\dt \ M^{-1} p^{n+1/4}- \displaystyle{\frac{\dt^2}{2} M^{-1} \nabla V (q^{n})},
\]
already satisfies $(C_{q})$ up to error terms of order two with respect to $\dt$. Indeed, using~\eqref{eq:const_satisf}:
\begin{align*}
\xi(\tilde{q}^{n+1})&=\xi(q^n)+ \dt   \nabla \xi (q^{n})^{T} M^{-1} p^{n+1/4} + \mathrm O (\dt^2) = z(t_{n+1}) +\mathrm O (\dt^2).
\end{align*}
This property is useful to ensure a fast convergence of the numerical algorithm solving the nonlinear constraints $(C_{q})$.

The numerical flow associated with~\eqref{eq:Verletconstswitched} is denoted in the sequel as
\begin{equation}
\label{eq:fwd_flow}
\Phi^n \ : \ 
\left\{ \begin{array}{ccc}
  \dps \manq_{\xi,v_\xi}\left(z(t_n), \frac{z(t_{n+1})-z(t_n)}{\dt}\right) 
  & \to & \dps \manq_{\xi,v_\xi}\left(z(t_{n+1}), \frac{z(t_{n+2})-z(t_{n+1})}{\dt}\right) \\[10pt]
  (q^n,p^{n+1/4}) & \mapsto & (q^{n+1},p^{n+3/4}) 
\end{array}\right.
\end{equation}
It can be proven that $\Phi^n$ is a symplectic map. The proof is indeed exactly the same as for the symplecticity of the classical RATTLE scheme, see \cite[Sections~VII.1.3]{HairerLubichWanner06} for an explicit computation for symplectic Euler and \cite[Sections~VII.1.4]{HairerLubichWanner06} for an extension to RATTLE.
As a consequence, $\Phi^n$ transports the phase space measure 
$\sigma_{\manq_{\xi,v_\xi}\left(z(t_n), \frac{z(t_{n+1})-z(t_n)}{\dt}\right)}$ to
the phase space measure 
$\sigma_{\manq_{\xi,v_\xi}\left(z(t_{n+1}), \frac{z(t_{n+2})-z(t_{n+1})}{\dt}\right)}$.

\subsubsection{Comments on the fluctuation-dissipation part~\eqref{eq:flucdissconstjarz1}-\eqref{eq:flucdissconstjarz2}}
\label{sec:numflucdissconstjarz}

In practice, \eqref{eq:flucdissconstjarz1} may be rewritten 
in a form more suited to numerical computations. 
Of course, similar considerations hold for~\eqref{eq:flucdissconstjarz2}.
Since $\gamma_P(q) = P_M(q) \gamma P_M(q)^T$ and $\sigma_P(q)=P_M(q)\sigma$, 
\eqref{eq:flucdissconstjarz1} is equivalent to:
\begin{equation}\label{eq:flucdissconstjarz1_bis}
\left \{ 
\begin{aligned}
& p^{n+1/4} = p^n  - \frac{\dt}{4} \gamma M^{-1}  
\pare{ p^n + p^{n+1/4} - 2\nabla\xi G_{M}^{-1}(q^n) 
\frac{ z(t_{n+1})- z(t_{n})}{\dt} } \\
& \phantom{p^{n+1/4} =} + \sqrt{\frac{\dt}{2}} \, \sigma \, {\mathcal G}^n 
+ \nabla\xi(q^n)^T \, \lambda^{n+1/4},\\
& \nabla\xi(q^n)^T M^{-1}p^{n+1/4} = \frac{ z(t_{n+1})- z(t_{n})}{\dt},  \qquad\qquad(C_p)
\end{aligned} \right.
\end{equation}
where the Lagrange multiplier $\lambda^{n+1/4}$ is associated with the constraint $(C_p)$. The equivalence between~\eqref{eq:flucdissconstjarz1} and~\eqref{eq:flucdissconstjarz1_bis} can be checked by multiplying~\eqref{eq:flucdissconstjarz1_bis} by $P_M(q^n)$ and  using~\eqref{eq:const_satisf}.

The Lagrange multiplier $\lambda^{n+1/4}$ in~\eqref{eq:flucdissconstjarz1_bis} is obtained by multiplying the above equation by 
$\nabla\xi(q^n)^T M^{-1} \left(\op{Id}+\frac\dt4 \gamma M^{-1} \right)^{-1}$, and 
solving the following linear system:
\[
\begin{split}
& \frac{ z(t_{n+1})- z(t_{n})}{\dt}  = \nabla\xi(q^n)^T M^{-1} \left(\op{Id}+\frac\dt4 \gamma M^{-1} \right)^{-1} \left(\op{Id}-\frac\dt4 \gamma M^{-1} \right) \, p^n  \\
& \quad + \nabla\xi(q^n)^T M^{-1} \left(\op{Id}+\frac\dt4 \gamma M^{-1} \right)^{-1} \pare{ \gamma M^{-1} \nabla\xi(q^n) G^{-1}_{M}(q^n) 
\frac{ z(t_{n+1})- z(t_{n})}{2} +    \sqrt{\frac{\dt}{2}} \,  \sigma \, {\mathcal G}^n } \\
& \quad +\nabla\xi(q^n)^T M^{-1} \left(\op{Id}+\frac\dt4 \gamma M^{-1} \right)^{-1} \nabla\xi(q^n) \, \lambda^{n+1/4}. 
\end{split}
\]
This system is well posed. Indeed, the matrix 
$\nabla\xi(q^n)^T M^{-1} \left(\op{Id}+\frac\dt4 \gamma M^{-1} \right)^{-1} \nabla\xi(q^n)$
can be rewritten as $\nabla \xi(q)^{T} S\nabla \xi(q)$ with $S=M^{-1} \left(\op{Id}+\frac{\dt}{4} \,\gamma M^{-1}\right)^{-1}$. Both $M$ and $\gamma$ are symmetric and non-negative, so that $S$ is symmetric, positive and invertible. Finally, the invertibility of  $\nabla \xi(q)^{T} S\nabla \xi(q)$ follows from the invertibility of $G_M(q)$. 

In the special case when $\gamma$ and $M$ are equal up to a multiplicative constant, the numerical integration can be simplified using the explicit formula~\eqref{eq:OUexpl_simpl} and the method described below~\eqref{eq:OUexpl_simpl}, which still holds for the tangential part of the momentum. See Section~\ref{sec:num_res_jarz} below for further precisions.

\subsubsection{Discretization of the backward process~\eqref{eq:Langevinconstjarzback}}

The splitting scheme for the backward Langevin dynamics with time-evolving constraints~\eqref{eq:Langevinconstjarzback} reads as follows: Denote $n'=N_T - n$, take initial conditions $(q^{{\rm b},0},p^{{\rm b},0})$ distributed according to $\mu_{\manq_{\xi,v_\xi}\left(z(t_{N_T}),\frac{z(t_{N_T+1})-z(t_{N_T})}{\dt}\right)}$ and iterate on $0 \leq n' \leq N_T-1$,
\begin{align}
&  \begin{cases}
\dps    p^{{\rm b},n'+1/4} =   p^{{\rm b},n'} -\frac{\dt}{4} \gamma_P(q^{{\rm b},n'}) M^{-1}(p^{{\rm b},n'+1/4}+p^{{\rm b},n'}) \\
\dps \phantom{ p^{{\rm b},n'+1/4} =} + \sqrt{ \frac{\dt}{2}} \sigma_P(q^{{\rm b},n'}) {\mathcal G}^{{\rm b},n'}, \end{cases}\label{eq:flucdissconstjarzbck1}\\[6pt]
é&\begin{cases}
 \dps  p^{{\rm b},n'+1/2} = p^{{\rm b},n'+1/4} + \displaystyle{\frac{\dt}{2} \nabla V (q^{{\rm b},n'})} + \nabla \xi(q^{{\rm b},n'}) \lambda^{{\rm b},n'+1/2}, &\\[6pt]
 \dps  q^{{\rm b},n'+1} = q^{{\rm b},n'} - \dt \ M^{-1} p^{{\rm b},n'+1/2},  \\[6pt]
 \dps   \xi(q^{{\rm b},n'+1})  = z(t_{N_T-n'-1}), & (C_q) \\[6pt]
 \dps  p^{{\rm b},n'+3/4} = p^{{\rm b},n'+1/2} + \displaystyle{\frac{\dt}{2} \nabla V (q^{{\rm b},n'+1})} +\nabla \xi(q^{{\rm b},n'+1}) \lambda^{{\rm b},n'+3/4},&\\[6pt]
  \dps   \nabla \xi (q^{{\rm b},n'+1})^{T} M^{-1} p^{{\rm b},n'+3/4} = \frac{ \dps  z(t_{N_T - n' })- z(t_{N_T - n' - 1 })}{\dps \dt}, &(C_p)
\end{cases}\label{eq:Verletconstswitchedbck}\\
&   \begin{cases}  \dps    p^{{\rm b},n'+1} = p^{{\rm b},n'+3/4} -\frac{\dt}{4} \gamma_P(q^{{\rm b},n'+1}) M^{-1}(p^{{\rm b},n'+3/4}+p^{{\rm b},n'+1}) \\
\dps \phantom{p^{{\rm b},n'+1} =} + \sqrt{\frac{\dt}{2}} \sigma_P(q^{{\rm b},n'+1}) {\mathcal G}^{{\rm b},n'+1/2}, \end{cases}\label{eq:flucdissconstjarzbck2}
\end{align}
where $({\mathcal G}^{{\rm b},n'})$ and $({\mathcal G}^{{\rm b},n'+1/2})$ are sequences of i.i.d. Gaussian random variables of mean $0$ and covariance matrix~$\mathrm{Id}_{3N}$.
The numerical flow associated with~\eqref{eq:Verletconstswitchedbck} is denoted
\[
\Phi^{{\rm b},n'} \ : \ 
\left\{ \begin{array}{ccc}
  \manq_{\xi,v_\xi} \pare{ \dps z(t_{n}), \frac{z(t_{n+1})-z(t_{n})}{\dt}  }
  & \to & \dps \manq_{\xi,v_\xi} \pare{ z(t_{n-1}), \frac{z(t_{n})-z(t_{n-1})}{\dt}  } \\[10pt]
  (q^{{\rm b},n'},p^{{\rm b},n'+1/4}) & \mapsto & (q^{{\rm b},n'+1},p^{{\rm b},n'+3/4}) 
\end{array}\right.
\]
where we recall $n'= N_T-n$. Assuming that the flow $\Phi^n$ given by~\eqref{eq:fwd_flow} and $\Phi^{{\rm b},n'}$ are both well-defined, the following reversibility property is easily checked (extending the symmetry property of the standard RATTLE scheme, see for instance~\cite[Section~VII.1.4]{HairerLubichWanner06}):
\begin{equation}
  \label{eq:rev_flow_jarz}
  \Phi^{{\rm b},N_T-n} \circ \Phi^{n-1} = \op{Id}.
\end{equation}

\subsubsection{Work discretization and free energy computations}
\label{sec:work_discretization}

The work~\eqref{eq:work} can be approximated using the Lagrange multipliers
in~\eqref{eq:Verletconstswitched}:
\begin{equation}
\label{eq:lang_num_work_multi_bis}
\begin{cases}
\W^0=0,\\
\dps \W^{n+1} = \W^{n} + \left(\frac{ z(t_{n+1})-z(t_{n}) }{ \dt }\right)^T 
\pare{ \lambda^{n+1/2} + \lambda^{n+3/4} },
\end{cases}
\end{equation}
for $n=0 \ldots N_T-1$. 
The (formal) consistency of the work discretization~\eqref{eq:lang_num_work_multi_bis} in the time continuous limit is a direct consequence of the work expression~\eqref{eq:work}.

An estimator of the free energy profile is then obtained by using $K$ independent realizations of the switching process, computing the work $\W^{N_T,k}$ for each realization $k \in \{1, \ldots, K\}$ 
(with the numerical trajectories obtained from the numerical scheme~\eqref{eq:flucdissconstjarz1}-\eqref{eq:Verletconstswitched}-\eqref{eq:flucdissconstjarz2} and i.i.d. initial conditions sampled according to~$\mu_{\manq_{\xi,v_\xi}\left(z(t_0),\frac{z(t_1)-z(t_0)}{\dt}\right)}$), and approximating~\eqref{eq:identity_to_approximate}, rewritten up to an unimportant additive constant (independent of~$T$), as
\[
F(z(T)) \simeq  -\frac1\beta \ln \E\left(\rme^{-\beta \left[ \W^{N_T} + C^{N_T}(q^{N_T})   \right]  }\right),
\]
with empirical averages such as 
\[
-\frac1\beta \ln \left ( 
\frac1K \sum_{k=1}^K \exp\left[-\beta \left( \W^{N_T,k}+C^{N_T}(q^{N_T,k}) \right) \right] \right).
\]
In the above, the discretization $C^{n}(q)$ of the corrector~\eqref{eq:corr} is
\begin{equation}
  \label{eq:corr_disc}
  C^n(q) = \frac{1}{2\beta} \ln \Big( \det G_M(q) \Big) - 
  \frac{1}{2} \pare{ \Frac{z(t_{n+1})-z(t_{n} )}{\dt} }^T G_M^{-1}(q) \pare{ \Frac{z(t_{n+1})-z(t_{n} )}{\dt} }.
\end{equation}
We refer to Chapter~4 in \cite{LelRouStoBook} for more background on free energy estimators
for nonequilibrium dynamics. 
In particular,
it is possible to compute a work associated with the backward switching from the
Lagrange multipliers in~\eqref{eq:Langevinconstjarzback}, and to resort to 
bridge estimators (see Section~4.2.3 in~\cite{LelRouStoBook}).

However, using approximations such as~\eqref{eq:lang_num_work_multi_bis} 
in the Jarzynski-Crooks identity introduces a time discretization error. 
We show in the next section how to eliminate this error.

\subsubsection{Discrete Jarzynski-Crooks identity}
\label{sec:discrete_jarz}

It turns out that a discrete version of the Jarzynski-Crooks identity~\eqref{eq:Jarz_Crooks} can be obtained. This enables the estimation of free energy differences using nonequilibrium simulation \emph{without time discretization error}. The discrete equality~\eqref{eq:time_reversal_on_paths_discrete} below may be seen as an extension of the corresponding equality obtained for transitions associated with time-dependent Hamiltonians and performed with Metropolis-Hastings dynamics (see~\cite{Crooks98} and Remark~4.5 in~\cite{LelRouStoBook}).

For this purpose, we consider a discretization of the work $\W_{0,T}$ using the interpretation~\eqref{eq:workexpr3} of the work as the energy variation of the Hamiltonian part of the Langevin dynamics. This leads to the following definition of the work at the discrete level:
\begin{equation}
\label{eq:lang_num_work_multi}
\begin{cases}
\W^0=0,\\
\W^{n+1} = \W^{n} + H(q^{n+1},p^{n+3/4}) - H(q^{n},p^{n+1/4}),
\end{cases}
\end{equation}
for $n=0 \ldots N_T-1$.
This work discretization leads to a Jarzynski-Crooks identity without time discretization error.

\begin{theorem}[Discrete Jarzynski-Crooks fluctuation identity]
\label{eq:jarz_disc}
  Consider the distribution $\mu_{\manq_{\xi,v_\xi}(z(t),\dot{z}(t))}$ 
  and its normalization $Z_{z(t),\dot{z}(t)}$ defined in~\eqref{eq:mu_v}.
  Denote by $\{q^n,p^n\}_{0 \leq n \leq N_T}$ the solution of the forward 
  discretized Langevin dynamics~\eqref{eq:flucdissconstjarz1}-\eqref{eq:Verletconstswitched}-\eqref{eq:flucdissconstjarz2} with initial conditions
  distributed according to
  \begin{equation}\label{eq:ICf_discrete}
  (q^0,p^0) \sim \mu_{\manq_{\xi,v_\xi} \pare{ z(t_0), \frac{z(t_1)-z(t_0)}{\dt}  } }(dq\, dp ),
  \end{equation}
  and by $\{q^{ {\rm b},n'} ,p^{ {\rm b},n'} \}_{0 \leq n' \leq N_T }$ the solution of the discretized 
  backward Langevin dynamics~\eqref{eq:flucdissconstjarzbck1}-\eqref{eq:Verletconstswitchedbck}-\eqref{eq:flucdissconstjarzbck2} distributed according to
  \begin{equation}\label{eq:ICb_discrete}
  (q^{ { \rm b},0},p^{ {\rm b},0}) \sim \mu_{\manq_{\xi,v_\xi} \pare{ z(t_{N_T}), \frac{z(t_{N_T+1})-z(t_{N_T} )}{\dt}  } }(dq\, dp ).
  \end{equation}
  Then, the following Jarzynski-Crooks identity holds on~$[0,N_T]$: for any bounded 
  discrete path functional $\ph_{[0,N_T]}$, 
  \begin{equation}
    \label{eq:Jarz_Crooks_discrete}
    \frac{Z_{z(N_T), \frac{z(t_{N_T+1})-z(t_{N_T} )}{\dt}   }}{Z_{z(t_0), \frac{z(t_1)-z(t_0)}{\dt}  }  } 
    = \frac{\E \pare{ \ph_{[0,N_T]}\left(\{q^n,p^n\}_{0 \leq n \leq N_T}\right) \,  
      {\rm e}^{-\beta\W^{N_T} } }}{\E \pare{ \ph^{\rm r}_{[0,N_T]}\left(\{q^{ {\rm b} , {n'} },p^{ {\rm b}, {n'} } \}_{0 \leq n' \leq N_T }\right) }},
  \end{equation}
  where $\W^n$ is computed according to~\eqref{eq:lang_num_work_multi}, and 
  $(\, \cdot \,)^{\rm r}$ denotes the composition with 
  the operation of time reversal of paths:
  \begin{equation}
    \label{eq:time_reversal_on_paths_discrete}
    \ph_{[0,N_T]}^{\rm r}\Big(\{ q^{ {\rm b},{n'}} ,p^{{\rm b}, {n'} } \}_{0 \leq n' \leq N_T}\Big) = 
    \ph_{[0,N_T]}\Big(\{ q^{ {\rm b},{N_T-n}} ,p^{{\rm b}, {N_T-n} } \}_{0 \leq n \leq N_T}\Big).
  \end{equation} 
\end{theorem}

The (formal) consistency of the work discretization~\eqref{eq:lang_num_work_multi} in the time continuous limit is a direct consequence of the work expression~\eqref{eq:workexpr3}.
Free energy estimators based on the identity~\eqref{eq:Jarz_Crooks_discrete} are obtained as described in Section~\ref{sec:work_discretization}. Let us emphasize once again that there is no error related to the finiteness of the time-step $\dt$ in this estimator, and that the only source of approximation is due to the statistical error.

\begin{proof}
  With a slight abuse of notation, we denote in the same way the random variables $(q^n,p^n)$, $(q^{{\rm b},n},p^{{\rm b},n})$, etc. in~\eqref{eq:flucdissconstjarz1}-\eqref{eq:Verletconstswitched}-\eqref{eq:flucdissconstjarz2} or~\eqref{eq:flucdissconstjarzbck1}-\eqref{eq:Verletconstswitchedbck}-\eqref{eq:flucdissconstjarzbck2}, and the integration variables in the definition of probability distributions. We divide the proof into three steps.

\smallskip

\noindent
Step 1: The phase space conservation of $\Phi^n$ and $\Phi^{{\rm b},n'}$ and the reversibility property~\eqref{eq:rev_flow_jarz} imply
\begin{equation}
\label{eq:detailed_H_jarz}
\begin{aligned}
&  \delta_{\Phi^n(q^n,p^{n+1/4}) }( dq^{n+1} \, d p^{n+3/4}  ) \, \sigma_{ \manq_{\xi,v_\xi}\left(z(t_{n}), \frac{z(t_{n+1})-z(t_{n})}{\dt}\right) } ( dq^n \, d p^{n+1/4}     )  \\ 
& \quad = \delta_{\Phi^{ {\rm b} ,N_T-n-1}(q^{n+1},p^{n+3/4})}( dq^{n} \, d p^{n+1/4}  ) \, \sigma_{ \manq_{\xi,v_\xi}\left(z(t_{n+1}), \frac{z(t_{n+2})-z(t_{n+1})}{\dt}\right) } ( dq^{n+1} \, d p^{n+3/4} ) . 
\end{aligned}
\end{equation}
\smallskip

\noindent
Step 2: The probability distribution of $p^{n+1/4}$ given $(q^n,p^n)$ in the discretization of the fluctuation-dissipation part~\eqref{eq:flucdissconstjarz1} is denoted $K^{\rm OU}(q^n,p^n,d p^{n+1/4})$. The scheme~\eqref{eq:flucdissconstjarz1} is a mid-point discretization of an Ornstein-Uhlenbeck process, which can be rewritten by decomposing the orthogonal and tangential updates of the momentum:
\begin{equation}
  \label{eq:5}
  \begin{cases}
    p^{n+1/4}_{\parallel} = \dps p^{n}_{\parallel} -\frac{\dt}{4} \gamma_P(q^n) M^{-1}\left( p^{n+1/4}_{\parallel}+p^{n}_{\parallel} \right) + \sqrt{ \frac{\dt}{2}} \sigma_P(q^n) {\mathcal G}^n, \\
    p^{n+1/4}_{\perp} = p^{n}_{\perp},
  \end{cases}
\end{equation}
where $p_{\parallel} = P_M(q^n) p$, and $p_{\perp} = (\op{Id} - P_M(q^n)) p$.
The Markov chain induced by the parallel part of the momentum is the same as the one induced by the scheme~\eqref{eq:flucdiss1} (or~\eqref{eq:flucdiss2}) defined in Section~\ref{sec:consnum}. The latter verifies a detailed balance equation (both in the plain sense and 
up to momentum reversal) with respect to the stationary measure $\kappa_{T^\ast_q \manq(z)}^{M^{-1}}(dp)$ defined by~\eqref{eq:kinetic_part} (see Sections~2.3.2 and~3.3.5 in \cite{LelRouStoBook}). We recall that this measure is defined as the kinetic probability distribution in the momentum variable of the canonical distribution $\mu_{T^*\manq(z)}(dq \, dp)$ on the tangential space, conditioned by a given $q\in \manq(z)$. Adding the (invariant) orthogonal part of the momentum, the following detailed balance condition is satisfied:
\begin{align}
&  \exp\left(-\frac\beta2 (p^n)^T M^{-1}p^n\right) K^{\rm OU}(q^n,p^n,d p^{n+1/4}) \, \sigma^{M^{-1}}_{\manq_{v_\xi(q^n,\cdot)}\left(\frac{z(t_{n+1})-z(t_{n} )}{\dt}\right)} (d p ^n ) \label{eq:detailed_OU_jarz} \\
& \quad = \exp\left(-\frac\beta2 (p^{n+1/4})^T M^{-1}p^{n+1/4}\right) K^{\rm OU}(q^n,p^{n+1/4},d p^n) \, \sigma^{M^{-1}}_{\manq_{v_\xi(q^n,\cdot)}\left(\frac{z(t_{n+1})-z(t_{n} )}{\dt}\right)} (d p ^{n+1/4} ). \nonumber
\end{align}

\smallskip

\noindent
Step $3$: Denote by $K^{\rm f}(q^n,p^n;dq^{n+1},dp^{n+1/4},dp^{n+3/4},dp^{n+1})$ the probability distribution of the variables $(q^{n+1},p^{n+1/4},p^{n+3/4},p^{n+1})$ given the variables $(q^n,p^n)$ in the scheme~\eqref{eq:flucdissconstjarz1}-\eqref{eq:Verletconstswitched}-\eqref{eq:flucdissconstjarz2}; and by $K^{\rm b}(q^{{\rm b},n'},p^{{\rm b},n'};dq^{{\rm b},n'+1},dp^{{\rm b},n'+1/4},dp^{{\rm b},n'+3/4},dp^{{\rm b},n'+1})$ the probability distribution of the variables $(q^{{\rm b},n'+1},p^{{\rm b},n'+1/4},p^{{\rm b},n'+3/4},p^{{\rm b},n'+1})$ given the variables $(q^{{\rm b},n'},p^{{\rm b},n'})$ in the scheme~\eqref{eq:flucdissconstjarzbck1}-\eqref{eq:Verletconstswitchedbck}-\eqref{eq:flucdissconstjarzbck2}. The splitting structure yields:
\begin{eqnarray*}
  &&  K^{{\rm f},n}(q^n,p^n;dq^{n+1}\,dp^{n+1/4}\,dp^{n+3/4}\,dp^{n+1}) \\
  && \quad = K^{{\rm OU}}(q^{n+1},p^{n+3/4}, dp^{n+1}) \delta_{\Phi^n(q^n,p^{n+1/4}) }( dq^{n+1} \, d p^{n+3/4}  ) K^{\rm OU}(q^n,p^{n}, dp^{n+1/4}),
\end{eqnarray*}
as well as
\begin{eqnarray*}
  && K^{{\rm b},n'}(q^{{\rm b},n'},p^{{\rm b},n'};dq^{{\rm b},n'+1}\,dp^{{\rm b},n'+1/4}\,dp^{{\rm b},n'+3/4}\,dp^{{\rm b},n'+1}) \\
  && \qquad \quad = K^{{\rm OU}}(q^{{\rm b},n'+1},p^{{\rm b},n'+3/4}, dp^{{\rm b},n'+1}) \delta_{\Phi^{{\rm b},n'}(q^{{\rm b},n'},p^{{\rm b},n'+1/4}) }( dq^{{\rm b},n'+1} \, d p^{{\rm b},n'+3/4}  ) \\
  && \qquad \qquad \times K^{\rm OU}(q^{{\rm b},n'},p^{{\rm b},n'}, dp^{{\rm b},n'+1/4}).
\end{eqnarray*}
Combining the detailed balance conditions~\eqref{eq:detailed_H_jarz} and~\eqref{eq:detailed_OU_jarz} of Steps~$1$ and~$2$, and using the decomposition~\eqref{eq:surfacemeas} of phase space measures, it follows
\begin{align*}
&\rme^{-\beta \pare{H(q^{n+1},p^{n+3/4}) - H(q^n,p^{n+1/4})}}   K^{{\rm f},n}(q^n,p^n;dq^{n+1}\,dp^{n+1/4}\,dp^{n+3/4}\,dp^{n+1}) \\
& \quad \times \rme^{-\beta H(q^n,p^{n})} \sigma_{ \manq_{\xi,v_\xi}\left(z(t_{n}), \frac{z(t_{n+1})-z(t_{n})}{\dt} \right) } ( dq^n \, d p^{n}     ) \\
&= K^{{\rm b},N_T-n-1}(q^{n+1},p^{n+1};dq^{n}\,dp^{n+3/4}\,dp^{n+1/4}\,dp^{n}) \\
& \quad \times \rme^{-\beta H(q^{n+1},p^{n+1})} \sigma_{ \manq_{\xi,v_\xi}\left(z(t_{n+1}), \frac{z(t_{n+2})-z(t_{n+1})}{\dt}\right) } ( dq^{n+1} \, d p^{n+1}),
\end{align*}
which can be seen as the Jarzynski-Crooks identity over one time-step. Iterating the argument, it is easy to obtain:
\begin{align*}
&\rme^{-\beta \mathcal{W}^{N_T}} K^{{\rm f},0}(q^0,p^0;dq^{1}\,dp^{1/4}\,dp^{3/4}\,dp^{1})\ldots
K^{{\rm f},N_T-1}(q^{N_T-1},p^{N_T-1};dq^{N_T}\,dp^{N_T-3/4}\,dp^{N_T-1/4}\,dp^{N_T}) \\
& \quad \times \rme^{-\beta H(q^0,p^{0})} \sigma_{ \manq_{\xi,v_\xi}\left(z(t_{0}), \frac{z(t_{1})-z(t_{0})}{\dt} \right) } ( dq^0 \, d p^{0}) \\
&= K^{{\rm b},N_T-1}(q^{1},p^{1};dq^{0}\,dp^{3/4}\,dp^{1/4}\,dp^{0})\ldots
K^{{\rm b},0}(q^{N_T},p^{N_T};dq^{N_T-1}\,dp^{N_T-1/4}\,dp^{N_T-3/4}\,dp^{N_T-1}) \\
& \quad \times \rme^{-\beta H(q^{N_T},p^{N_T})} \sigma_{ \manq_{\xi,v_\xi}\left(z(t_{N_T}), \frac{z(t_{N_T+1})-z(t_{N_T})}{\dt}\right) } ( dq^{N_T} \, d p^{N_T}),
\end{align*}
which yields~\eqref{eq:Jarz_Crooks_discrete}.
\end{proof}

\subsection{The overdamped limit}
\label{sec:ovd_limit_disc}

\subsubsection{An exact free energy estimator for the overdamped Langevin dynamics}

The splitting scheme~\eqref{eq:flucdissconstjarz1}-\eqref{eq:Verletconstswitched}-\eqref{eq:flucdissconstjarz2} can be used in the overdamped regime, using the method of Proposition~\ref{p:langtooverd}, {\it i.e.} by choosing 
 \begin{equation}
   \label{eq:scaling_num_ovd_const_rappel}
   \frac{\dt}{4} \gamma = M = \frac{\dt}{2} \op{Id},
 \end{equation}
which implies $\gamma_P = 2 P_M^T P_M$ and $\sigma_P = \frac{2}{\sqrt{\beta}} P_M$.
For this choice of parameters, the continuous limit of the numerical scheme is the following variant of the stochastic differential equation~\eqref{eq:overdampconst}:
\begin{equation}
  \label{eq:overdampconst_jarz}
  \begin{cases}
      d q_t = \dps -\nabla V(q_t) \, dt + \sqrt{\frac{2}{\beta}} \, d W_t +  \nabla \xi (q_t) \, 
    d \lambda_t^{\rm od} , \\
\xi(q_t)=z(t),
  \end{cases}
\end{equation}
where $\lambda_t^{\rm od}$ is an adapted stochastic process such that $\xi(q_t)=z(t)$. We then obtain the following Jarzynski-Crooks relation for discretized overdamped dynamics, 
\emph{without time discretization error}.

\begin{proposition}
\label{p:langtooverd_jarz}
Suppose that the relation~\eqref{eq:scaling_num_ovd_const_rappel} is satisfied. With a slight abuse of notation, the mass matrix and the friction matrix 
are rewritten as $M \, \mathrm{Id}$ and $\gamma \, \mathrm{Id}$ with $M,\gamma \in \mathbb{R}$.
Then the splitting scheme~\eqref{eq:flucdissconstjarz1}-\eqref{eq:Verletconstswitched}-\eqref{eq:flucdissconstjarz2} yields the following Euler discretization of the overdamped Langevin constrained dynamics~\eqref{eq:overdampconst_jarz}:
\begin{equation}
  \label{eq:Eulerconst_jarz}
  \begin{cases}
      \dps q^{n+1} = q^{n} - \dt  \nabla V (q^{n}) 
    + \sqrt{\frac{2\dt}{\beta}} \, {\mathcal G}^n + \nabla \xi (q^n) \, \lambda^{n+1}_{\rm od},  \\[6pt]
    \xi(q^{n+1}) = z(t_{n+1}),
  \end{cases}
\end{equation}
where $({\mathcal G}^{n})_{n\geq 0}$ are independent and identically distributed 
centered and normalized Gaussian
variables, and $(\lambda^n_{\rm od})_{n \geq 1}$ are the Lagrange multipliers
associated with the constraints $(\xi(q^{n}) = z(t_{n}))_{0 \leq n \leq N_T}$. In the same way, the backward process~\eqref{eq:flucdissconstjarzbck1}-\eqref{eq:Verletconstswitchedbck}-\eqref{eq:flucdissconstjarzbck2} yields the following Euler scheme:
\begin{equation}
  \label{eq:Eulerconst_jarz_bck}
  \begin{cases}
      \dps q^{{\rm b},n'+1} = q^{{\rm b},n'} - \dt  \nabla V (q^{{\rm b},n'}) 
    + \sqrt{\frac{2\dt}{\beta}} \, {\mathcal G}^{{\rm b},n'} + \nabla \xi (q^{{\rm b},n'}) \, \lambda^{{\rm b},n'+1}_{\rm od},  \\[6pt]
    \xi(q^{{\rm b},n'+1}) = z(t_{N_T-n'-1}).
  \end{cases}
\end{equation}
Consider the work update
\begin{equation}
\label{eq:lang_num_work_multi_od}
\begin{cases}
\W^0=0,\\
\dps \W^{n+1} = \W^{n} + V(q^{n+1}) - V(q^n) + \frac{1}{\dt}\pare{ \abs{p^{n+3/4}}^2 - \abs{p^{n+1/4}}^2 },
\end{cases}
\end{equation}
for $n=0 \ldots N_T-1$, where
\begin{equation*}
  \begin{cases}
  \dps 2  p^{n+1/4} = \sqrt{\frac{2\dt}{\beta}} \, P(q^n) {\mathcal G}^n + \nabla \xi(q^n) G^{-1}(q^n)\pare{   z(t_{n+1})- z(t_{n})}  , \\
  \dps 2 \lambda^{n+1/2} = \lambda^{n+1}_{\rm od} - G^{-1}(q^n) \pare{   z(t_{n+1})- z(t_{n})} + \sqrt{\frac{2\dt}{\beta}} G^{-1}(q^n)\nabla \xi(q^n)^T{\mathcal G}^n,
  \end{cases}
\end{equation*}
with $G = \nabla\xi^T\nabla\xi$, and the scheme~\eqref{eq:Eulerconst_jarz} is rewritten as:
\begin{equation}\label{eq:p_od}
\begin{cases}
 \dps  p^{n+1/2} = p^{n+1/4} - \displaystyle{\frac{\dt}{2} \nabla V (q^{n})} + \nabla \xi(q^n) \lambda^{n+1/2}, \\[6pt]
 \dps  q^{n+1} = q^{n} + 2 p^{n+1/2},  \\[6pt]
 \dps   \xi(q^{n+1})  = z(t_{n+1}), \quad (C_q) \\[6pt]
 \dps  p^{n+3/4} = p^{n+1/2} - \displaystyle{\frac{\dt}{2} \nabla V (q^{n+1})} +\nabla \xi(q^{n+1}) \lambda^{n+3/4},\\[6pt]
  \dps   \nabla \xi (q^{n+1})^{T} p^{n+3/4} = \frac{   z(t_{n+2})- z(t_{n+1})}{2}, \quad (C_p)
\end{cases}
\end{equation}
Then the Jarzynski-Crooks relation~\eqref{eq:Jarz_Crooks_discrete} holds under the assumptions~\eqref{eq:ICf_discrete} and~\eqref{eq:ICb_discrete} on the initial conditions of the schemes~\eqref{eq:Eulerconst_jarz} and~\eqref{eq:Eulerconst_jarz_bck} respectively.
\end{proposition}

The proof is a direct consequence of the reformulation of~\eqref{eq:Eulerconst_jarz} into~\eqref{eq:p_od}, and a direct application of Theorem~\ref{eq:jarz_disc} with  the parameters~\eqref{eq:scaling_num_ovd_const_rappel}.

Note that the free energy estimator
\begin{equation}
  \label{eq:exact_ovd_estimator}
  F(z(T)) = -\frac1\beta \ln \E\left(\rme^{-\beta \left[ \W^{N_T} + C^{N_T}(q^{N_T})   \right]  }\right),
\end{equation}
based on the work~\eqref{eq:lang_num_work_multi_od} and the corrector 
\begin{equation}
  \label{eq:corr_od_disc}
   C^n(q) = \frac{1}{2\beta} \ln \Big( \det G(q) \Big) - 
  \frac{\dt}{4} \pare{ \Frac{z(t_{n+1})-z(t_{n} )}{\dt} }^T G^{-1}(q) \pare{ \Frac{z(t_{n+1})-z(t_{n} )}{\dt} }
\end{equation}
is exact (there is no time discretization error).
This free energy estimator can be seen as a variant of the estimator proposed in~\cite{lelievre-rousset-stoltz-07-a}, which was derived directly for the scheme~\eqref{eq:Eulerconst_jarz}, and reads (up to an unimportant additive constant):
\begin{equation}
  \label{eq:previous_estimator_jarz_ovd}
  F(z(T)) \simeq  -\frac1\beta \ln \E\left(\rme^{-\beta \left[ \widetilde{\W}^{N_T} + \widetilde{C}(q^{N_T})   \right]  }\right),
\end{equation}
where the work is defined as 
\begin{equation}
  \label{eq:work_od_jcp}
  \begin{cases}
    \widetilde{\W}^0=0,\\
    \dps \widetilde{\W}^{n+1}-\widetilde{\W}^{n}  = \pare{ \frac{\dps z(t_{n+1})-z(t_{n})}{\dps \dt } }^T  \tilde{\lambda}^{n+1}_{\rm od},
  \end{cases}
\end{equation}
with
\begin{equation*}
  \tilde{\lambda}^{n+1}_{\rm od} = 2\lambda^{n+1/2}= \lambda^{n+1}_{\rm od}  - G^{-1}(q^n)(z(t_{n+1})-z(t_{n})) + \sqrt{\frac{2 \dt}{\beta}}  G^{-1}(q^n)  \nabla \xi (q^n) ^T {\mathcal G}^n,
\end{equation*}
and the modified corrector is defined without the kinetic energy term:
\begin{equation}\label{eq:corr_jcp}
  \widetilde{C}(q) = \frac{1}{2\beta} \ln \Big( \det G(q) \Big).
\end{equation}
There is a bias due to the time discretization error in the estimator~\eqref{eq:previous_estimator_jarz_ovd}, which can be removed upon following the procedure described in Proposition~\ref{p:langtooverd_jarz}.

\subsubsection{Consistency analysis of three free energy estimators}
\label{sec:time_continuous_limit_ovd}

In this section, we would like to discuss the consistency of three free energy estimators introduced above: \eqref{eq:previous_estimator_jarz_ovd}-\eqref{eq:work_od_jcp} (based on the direct discretization of the overdamped dynamics proposed in~\cite{lelievre-rousset-stoltz-07-a}), \eqref{eq:exact_ovd_estimator}-\eqref{eq:lang_num_work_multi_bis} (which uses the Lagrange multipliers to approximate the work) and
\eqref{eq:exact_ovd_estimator}-\eqref{eq:lang_num_work_multi_od} (based on the 
discrete Jarzynski equality).

The limiting continuous-in-time version of the Jarzynski relation is:
\begin{equation}\label{eq:FE_jcp}
F(z(T))=-\frac1\beta \ln \E\left(\mathrm{e}^{-\beta \W^{\rm od}_{0,T}
\pare{ \set{q_t}_{0 \leq t \leq T}}}\right),
\end{equation}
where the work for the overdamped dynamics~\eqref{eq:overdampconst_jarz} reads (see~\cite{lelievre-rousset-stoltz-07-a}):
  \begin{equation}\label{eq:work_jcp}
    \W^{\rm od}_{0,T}\pare{ \set{q_t}_{0 \leq t \leq T}} = \int_0^T z'(t)^T \, d \widetilde{\lambda}^{\rm od}_t,
  \end{equation}
with
$$ d \widetilde{\lambda}^{\rm od}_t = d \lambda_t^{\rm od} - G^{-1}(q_t) z'(t) \, dt + \sqrt{\frac2\beta}  G^{-1}(q_t) \nabla \xi (q_t) ^T \, d W_t.$$
The consistency of~\eqref{eq:previous_estimator_jarz_ovd}-\eqref{eq:work_od_jcp} with~\eqref{eq:FE_jcp}-\eqref{eq:work_jcp} was already proven in~\cite{lelievre-rousset-stoltz-07-a}.

Concerning the consistency of  $C^{n}$ with $\widetilde{C}$ (see~\eqref{eq:exact_ovd_estimator} and~\eqref{eq:previous_estimator_jarz_ovd}), note that in the overdamped scaling ($M=\frac{\Delta t}{2} \op{Id}$), the difference
\[
\widetilde{C}(q)-C^{n}(q) = \frac{\dt}{4} \pare{ \frac{z(t_{n+1})-z(t_{n} )}{\dt} }^T G^{-1}(q) \pare{ \frac{z(t_{n+1})-z(t_{n} )}{\dt} } =  \mathrm{O}( \dt)
\]
vanishes when $\dt \to 0$. This difference can therefore be neglected when analyzing the consistency of the scheme in the continuous-in-time limit. We henceforth concentrate on the consistency of the works~\eqref{eq:lang_num_work_multi_bis} and~\eqref{eq:lang_num_work_multi_od} with~\eqref{eq:work_jcp}.

In the sequel, we denote the anticipating stochastic integration of the integrand $Y_t$ with respect to $d X_t$ by 
\[
Y_t\, \dot{}\, dX_t = 2Y_t \, \circ dX_t - Y_t\, . d X_t, 
\]
where $\, \circ \,$ is the Stratonovitch symmetric integration, and $\, . \,$ the It\^o integration. The symbol $\leadsto$ denotes the formal time continuous limit.

\medskip

\paragraph{Consistency of~\eqref{eq:lang_num_work_multi_bis}.}

Let us justify the consistency of the work expression~\eqref{eq:lang_num_work_multi_bis} with~\eqref{eq:work_jcp}. Remark that the Lagrange multipliers in \eqref{eq:p_od} verify:
\begin{equation}
2\lambda^{n+1/2} = G^{-1}(q^n) \Big[ \nabla \xi (q^n)^T \pare{q^{n+1}-q^{n}} -(z(t_{n+1})-z(t_{n} ))+ \dt \nabla \xi (q^n)^T \nabla V (q^n) \Big] \label{eq:lag_od_jarz_1} 
\end{equation}
and
\begin{equation}
2\lambda^{n+1/2} =  \lambda^{n+1}_{\rm od} - G^{-1}(q^n)(z(t_{n+1})-z(t_{n}))  +  \sqrt{\frac{2\dt}{\beta}}  G^{-1}(q^n)  \nabla \xi (q^n) ^T {\mathcal G}^n,   \label{eq:lag_od_jarz_2} 
\end{equation}
as well as
\begin{align}
2 \lambda^{n+3/4} &= G^{-1}(q^{n+1}) \Big[ \nabla \xi (q^{n+1})^T \pare{q^{n}-q^{n+1}} + (z(t_{n+2})-z(t_{n+1})) + \dt \nabla \xi (q^{n+1})^T \nabla V (q^{n+1})\Big]. \label{eq:lag_od_jarz_3}
 \end{align}
The expressions~\eqref{eq:lag_od_jarz_1} and~\eqref{eq:lag_od_jarz_3} yield
\begin{align}
  \label{eq:lag_od_1_bis}
  2 \lambda^{n+1/2} \leadsto G^{-1}(q_t)\pare{\nabla \xi(q_t)^T \, . dq_t - z'(t) \, dt + \nabla \xi(q_t)^T \nabla V(q_t) \, dt  },
\end{align}
as well as
\[
  2 \lambda^{n+3/4} \leadsto G^{-1}(q_t)\pare{-\nabla \xi(q_t)^T \, \dot{} dq_t + z'(t) \, dt + \nabla \xi(q_t)^T \nabla V(q_t) \, dt  }.
\]
Moreover the constraints imply that
\begin{equation}
  \label{eq:limit_dxi}
d\xi(q_t) = z'(t) \, dt = \nabla \xi^T(q_t) \circ dq_t = \frac12 \left( \nabla \xi^T(q_t) \, .  dq_t  + \nabla \xi^T(q_t) \, \dot{} dq_t \right),
\end{equation}
so that $\lambda^{n+1/2}$ and $\lambda^{n+3/4}$ yield the same time continuous limit, that is to say
\begin{align}
  \label{eq:lag_od_2_bis}
  2 \lambda^{n+3/4} \leadsto G^{-1}(q_t)\pare{\nabla \xi(q_t)^T  \, . dq_t - z'(t) \, dt + \nabla \xi(q_t)^T \nabla V(q_t) \, dt  }.
\end{align}
Eventually, \eqref{eq:lag_od_jarz_2} implies
\begin{equation}
  \label{eq:lag_final}
  2 \lambda^{n+1/2} \leadsto d \widetilde{\lambda}_t^{\rm od},
\end{equation}
the same holding true for $2 \lambda^{n+3/4}$. The work expression~\eqref{eq:lang_num_work_multi_bis} is thus formally consistent with~\eqref{eq:work_jcp}.
This concludes the proof of the consistency of \eqref{eq:exact_ovd_estimator}-\eqref{eq:lang_num_work_multi_bis} with~\eqref{eq:FE_jcp}-\eqref{eq:work_jcp}.

\medskip

\paragraph{Consistency of~\eqref{eq:lang_num_work_multi_od}.}

We now prove the consistency of the work expression~\eqref{eq:lang_num_work_multi_od} with~\eqref{eq:work_jcp}. Define 
\[
f^n = - \displaystyle{\frac{\dt}{2} \nabla V (q^{n})} + \nabla \xi(q^n) \lambda^{n+1/2},
\qquad 
f^{n+1} = - \displaystyle{\frac{\dt}{2} \nabla V (q^{n+1})} +\nabla \xi(q^{n+1}) \lambda^{n+3/4}. 
\]
The expression~\eqref{eq:lang_num_work_multi_od} yields using~\eqref{eq:p_od}:
\begin{align}
\W^{n+1}-\W^{n} &= V(q^{n+1})- V(q^n) +\frac{1}{\dt}\pare{\abs{p^{n+1/2} + f^{n+1}}^2 - \abs{p^{n+1/2}- f^{n}}^2  } \nonumber \\
& = V(q^{n+1})- V(q^n) + \frac{1}{\dt}\pare{f^n + f^{n+1}} \cdot \pare{q^{n+1}- q^n - f^n + f^{n+1}} \nonumber \\
& = V(q^{n+1})- V(q^n) - \frac12 (\nabla V (q^{n})+\nabla V (q^{n+1}))\cdot(q^{n+1}-q^n) \label{eq:deltaW1} \\
& \quad + I^n + \frac{1}{\dt}\pare{f^n + f^{n+1}} \cdot \pare{f^{n+1}-f^n}, \label{eq:deltaW}
\end{align}
where
\[
I^n =\frac{1}{\dt}\pare{ \nabla \xi(q^n) \lambda^{n+1/2} + \nabla \xi(q^{n+1}) \lambda^{n+3/4}} \cdot \pare{q^{n+1}- q^n}.
\]
First, since $V(q^{n+1})- V(q^n) \leadsto \nabla V (q_t) \, \circ dq_t$ and
$ \frac12(\nabla V (q^{n}) + \nabla V (q^{n+1}) ) \cdot (q^{n+1}- q^n) \leadsto - \nabla V (q_t) \, \circ dq_t$, the limit of the terms in~\eqref{eq:deltaW1} is zero.
Second, using the expressions~\eqref{eq:lag_od_jarz_1} and \eqref{eq:lag_od_jarz_3}, 
and similarly to~\eqref{eq:lag_od_1_bis} and~\eqref{eq:lag_od_2_bis}, it holds
\begin{equation}\label{eq:diff_f_consist}
\begin{aligned}
  & f^{n}-f^{n+1} = (\nabla \xi G^{-1})(q^n) \Big[ \nabla \xi (q^n)^T \pare{q^{n+1}-q^{n}} -(z(t_{n+1})-z(t_{n} ))+ \dt \nabla \xi (q^n)^T \nabla V (q^n) \Big]  \\
  & - (\nabla \xi G^{-1})(q^{n+1}) \Big[ \nabla \xi (q^{n+1})^T \pare{q^{n}-q^{n+1}} + (z(t_{n+2})-z(t_{n+1})) + \dt \nabla \xi (q^{n+1})^T \nabla V (q^{n+1})\Big]  \\
&+ \frac\dt2 (\nabla V(q^{n+1}) - \nabla V(q^n)) = \mathrm{o}(\dt) 
\end{aligned}
\end{equation}
since $( \nabla \xi G^{-1} \nabla \xi (q^n) + \nabla \xi G^{-1} \nabla \xi (q^{n+1}))^T
\pare{q^{n+1}-q^{n}} \leadsto 2\nabla \xi G^{-1}(q_t) z'(t)\,dt$ by~\eqref{eq:limit_dxi}.
Expanding in higher order powers of~$\dt$, it can be checked that 
there exists two functions~$a$ and~$b$ such that
\[
\frac{f^{n}-f^{n+1}}{\dt} \leadsto a(t,q_t) \, dt + b(t,q_t) . dq_t.
\]
Therefore, since (in the limit $\dt \to 0$) the martingale part of $f^n + f^{n+1}$ arises only from the term $\nabla \xi G^{-1} \nabla \xi^T (q_t) . dq_t$, one obtains
\begin{align*}
\pare{f^n + f^{n+1}} \cdot \frac{f^{n+1}-f^n}{\dt} &\leadsto d\left\langle \int_0^. \nabla \xi G^{-1} \nabla \xi^T (q_t) . dq_t , \int_0^. b(t,q_t) . dq_t \right\rangle_t \\
&=\frac2\beta d\left\langle \int_0^. \nabla \xi G^{-1} \nabla \xi^T (q_t) P(q_t). dW_t , \int_0^. b(t,q_t) P(q_t). dW_t \right\rangle_t = 0
\end{align*}
since $\nabla \xi^T (q_t) P(q_t) = 0$. In conclusion, the second term in~\eqref{eq:deltaW}
has a zero contribution to the continuous-in-time limit.

As a consequence, the formal time continuous limit of $\W^{n+1}-\W^{n}$ is the same as the one of $I^n$. Computations similar to the one performed above yield
\[
J^n = \frac{1}{2\dt} (q^{n+1}-q^n)^T(\nabla \xi(q^{n+1}) - \nabla \xi(q^{n}) ) ( \lambda^{n+1/2} -\lambda^{n+3/4}) = \mathrm{o}(\dt).
\]
Indeed, $\lambda^{n+1/2} -\lambda^{n+3/4} = \mathrm{o}(\dt)$ as in~\eqref{eq:diff_f_consist},
while $(q^{n+1}-q^n)^T(\nabla \xi(q^{n+1}) - \nabla \xi(q^{n}) ) = \mathrm{O}(\dt)$.
The formal time continuous limit of $I^n$ is therefore the same as the limit of
\begin{align*}
 I^n+J^n = \frac{1}{2\dt} (q^{n+1}-q^n)^T\pare{\nabla \xi(q^n)+\nabla \xi(q^{n+1})}\pare{ \lambda^{n+1/2} +  \lambda^{n+3/4}}.
\end{align*}
Since~\eqref{eq:limit_dxi} implies
\[
\frac{\pare{\nabla \xi(q^n)+\nabla \xi(q^{n+1})}^T}{2} (q^{n+1}-q^n) \leadsto z'(t) \, dt, 
\]
we get in the end that the formal time continuous limit of $I^n$ and $\W^{n+1}-\W^{n}$ is the same as:
 \[
z'(t_n)^T \pare{ \lambda^{n+1/2} +  \lambda^{n+3/4}} \leadsto z'(t)^T d \widetilde{\lambda}_t^{\rm od},
\]
where we have used~\eqref{eq:lag_final}.
This concludes the proof of the consistency of \eqref{eq:exact_ovd_estimator}-\eqref{eq:lang_num_work_multi_od}  with~\eqref{eq:FE_jcp}-\eqref{eq:work_jcp}.


\subsection{Numerical illustration}
\label{sec:num_res_jarz}

We present some free energy profiles obtained with nonequilibrium switching dynamics
for the model system and the  parameters described in Section~\ref{sec:simple_example_WCA}.
The switching schedule reads 
\[
z(t) = z_{\rm min} + (z_{\rm max}-z_{\rm min})\frac{t}{T}
\]
with $z_{\rm min} = -0.1$ and $z_{\rm max} = 1.1$. The time-step is $\dt = 0.01$. 
The initial conditions
are obtained by first subsampling a constrained dynamics with $\xi(q) = z_{\rm min}$ 
and $v_\xi(q,p) = 0$, with a time spacing $T_{\rm sample} = 1$;
and then adding the required component $\nabla \xi(q) G_M^{-1} \dot{z}(0)$ to the momentum
variable (with $\dot{z}(0) = (z_{\rm max}-z_{\rm min})/T$).

In the specific case at hand, the corrector term~\eqref{eq:corr} is constant, and 
free energies differences are equal to differences of rigid free energies.
The dynamics used to integrate the nonequilibrium dynamics is based on a splitting strategy,
analogous to~\eqref{eq:flucdissconstjarz1}-\eqref{eq:Verletconstswitched}-\eqref{eq:flucdissconstjarz2}, except that the midpoint 
integration of the Ornstein-Uhlenbeck part is replaced by an exact integration for
the unconstrained dynamics, followed by a projection. This can be done here
since we choose a friction matrix of the form $\gamma \, \mathrm{Id}$ (recall also that 
$M = \mathrm{Id}$). More precisely, the corresponding scheme is obtained by 
replacing~\eqref{eq:flucdissconstjarz1} (and likewise for~\eqref{eq:flucdissconstjarz2}) with
\[
\widetilde{p}^{n+1/4} = \alpha p^n + \sqrt{\frac{1-\alpha^2}{\beta}} \, {\mathcal G}^n,
\]
where $\alpha= {\rm e}^{- \gamma \dt}$, and setting $p^{n+1/4} = \widetilde{p}^{n+1/4} + \lambda^{n+1/4} \nabla \xi(q^n)$ with
$\lambda^{n+1/4}$ chosen such that 
\[
\nabla \xi(q^n)^T M^{-1} p^{n+1/4} = \frac{z(t_{n+1})-z(t_n)}{\dt}.
\]

\begin{figure}
\centering
\includegraphics[width=7.3cm]{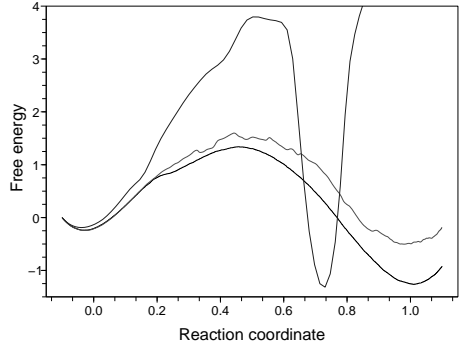}
\caption{\label{fig:noneq_RC} 
  Free energy profiles. The top curve corresponds to
  $T=1$ with $M=10^5$, while the two other curves were obtained for
  $T=10$ with $M=10^4$ and $T=100$ with $M=10^3$ (smoothest curve).
}
\end{figure} 

Figure~\ref{fig:noneq_RC} presents estimates obtained with $M$ 
independent realizations of the switching dynamics
for different switching times $T$, using the estimator presented in 
Section~\ref{sec:work_discretization} with the work 
discretization~\eqref{eq:lang_num_work_multi_bis}.
In all cases, the product $MT$ is kept constant. 
The free energy profile becomes closer to the reference curve as $T$ is increased, and the
profile obtained for $T=100$ is in excellent agreement with the
result obtained with thermodynamic integration.
When the switching time is small, more realizations should be considered to reduce 
the statistical errors and obtain estimates in
better agreement with the reference profile.
The fact that the variance is very large when the switching time~$T$ is too small is a well-known
drawback of estimators based on the Jarzynski-Crooks identity, see the review in 
Sections~4.1.4 and~4.1.5 in~\cite{LelRouStoBook}.
Roughly speaking, the difficulty is related to the fact that
the free energy difference is obtained as an average of 
$\exp(- \beta {\mathcal W})$, which requires a very good
sampling of the small values of the work~${\mathcal W}$. 
As $T$ decreases, the width of the work distribution increases
and the low tail part is more and more difficult to sample.
Improved estimates can be obtained with 
estimators based on combinations of forward and backward trajectories,
see for instance~\cite{MA08} and Section~4.2 in~\cite{LelRouStoBook}.


\section{Appendix: Some technical results}
\label{sec:appendix}

We give in this appendix two technical lemmas, used in the proof of 
Proposition~\ref{p:meanforce}. 
The first lemma can also be used to prove the divergence formula~\eqref{eq:divsympl}.

\begin{lemma}
  \label{l:Hnull} 
  For any $a \in \set{1,\dots,2\nc}$ :
  \begin{equation}
    \label{eq:Hnull}
   \sum_{b=1}^{2\nc} \poisson{ \abs{ {\rm det} \, \Gamma}^{1/2 } (\sgram^{-1})_{a,b},\Xi_b} = 0,
  \end{equation}
where $\Gamma$ is defined in \eqref{eq:sgram}.
\end{lemma}

\begin{proof}
 The proof relies on the following computation rules for any family 
 of invertible square matrices $\theta \mapsto A_{\theta}$:
 \begin{equation}\label{eq:diffdet}
\frac{d}{d \theta} \,  \Big ( \ln \abs{ {\rm det} A_{\theta} }  \Big) = 
         {\rm tr}\left(A_{\theta}^{-1} \frac{d}{d \theta} A_{\theta}\right),
 \end{equation}
 and
 \begin{equation}\label{eq:diffmat}
 A_{\theta} \, \frac{d}{d \theta}\big(A_{\theta}^{-1}\big)
 =-\Big( \frac{d}{d \theta} A_{\theta}\Big) A_{\theta}^{-1}.
 \end{equation}
 Fix $a \in \{ 1,\dots,2m \}$.
 First, using~\eqref{eq:diffmat} with $A_\theta$ replaced by $\Gamma$ and $\frac{d}{d \theta}$
 replaced by $\poisson{\cdot,\Xi_c}$, we obtain
 \[
 \sum_{b,c=1}^{2\nc} \sgram_{a,b}\poisson{(\sgram^{-1})_{b,c},\Xi_c} = 
 - \sum_{b,c=1}^{2\nc} \poisson{\sgram_{a,b},\Xi_c} (\sgram^{-1})_{b,c},
 \]
so that by the skew-symmetry of $\Gamma^{-1}$ and $\Gamma$, 
 \begin{align*}
   \sum_{b,c=1}^{2\nc} \sgram_{a,b}\poisson{(\sgram^{-1})_{b,c},\Xi_c}
   &=  \sum_{b,c=1}^{2\nc} -\frac12 \Big( \poisson{\Gamma_{a,b},\Xi_c} 
   + \poisson{\Gamma_{c,a},\Xi_b} \Big) (\sgram^{-1})_{b,c}.
 \end{align*}
 Jacobi's identity for Poisson brackets and~\eqref{eq:diffdet} then yield
 \begin{align*}
   \sum_{b,c=1}^{2\nc} \sgram_{a,b}\poisson{(\sgram^{-1})_{b,c},\Xi_c} & =
   \frac12 \sum_{b,c=1}^{2\nc} \poisson{\poisson{\Xi_b,\Xi_c},\Xi_a} (\sgram^{-1})_{b,c} \\
   & = -\frac12 \sum_{b,c=1}^{2\nc} \poisson{\Gamma_{c,b},\Xi_a} (\sgram^{-1})_{b,c} \\
   & =   -\frac{1}{2} \poisson{\ln \abs{  \det \sgram},\Xi_a }
   = -\abs{\det \sgram}^{-1/2}\poisson{\abs{ \det \sgram}^{1/2},\Xi_a}  \\
   & =  -\sum_{b,c=1}^{2\nc} \abs{\det \sgram}^{-1/2}\sgram_{a,b}(\sgram^{-1})_{b,c} 
   \poisson{ \abs{\det \sgram}^{1/2},\Xi_c}
 \end{align*}
 since $\sgram_{a,b}(\sgram^{-1})_{b,c} = \delta_{a,c}$ where $\delta_{i,j}$ is the Kronecker symbol.
Finally, the left hand and right hand sides of the last equality can be factorized as
\[
\sum_{b,c=1}^{2\nc} 
\abs{\det \sgram}^{-1/2}\sgram_{a,b}\poisson{ \abs{\det \sgram}^{1/2}(\sgram^{-1})_{b,c},\Xi_c} = 0.
\]
Since $\abs{\det \sgram} > 0$ and $\sgram$ is invertible, it follows
\[
\sum_{c=1}^{2\nc} \poisson{ \abs{\det \sgram}^{1/2}(\sgram^{-1})_{b,c},\Xi_c} = 0
\]
for all $b = 1,\dots,2m$, which is~\eqref{eq:Hnull}.
\end{proof}

\begin{lemma}\label{l:consint}
 For any compactly supported smooth test function $\ph$ on $\R^{6N}$:
 \[
 \nabla_{\zeta}  \pare{ \int_{\manq_{\Xi}(\zeta)} \ph \, d\sigma_{\manq_{\Xi}(\zeta)} } = \int_{\manq_{\Xi}(\zeta)} \sgram^{-1}\poisson{\Xi,\ph} \, d\sigma_{{\manq_{\Xi}(\zeta)}},
 \]
where the phase space $\manq_{\Xi}(\zeta)$ is defined in \eqref{eq:phaseXi}, and the Gram matrix $\Gamma$ in \eqref{eq:sgram}.
 \end{lemma}

\begin{proof}
  Consider a test function $\phi:\R^{2\nc} \to \R$. An integration by parts and the co-area formula~\eqref{eq:coareasympl} give:
  \begin{align*}
    I &:=\int_{\R^{2\nc}} \phi(\zeta) \nabla_{\zeta}\pare{\int_{\manq_{\Xi}(\zeta)} \ph(q,p) \, \sigma_{\manq_\Xi(\zeta)}(dq\, dp) }\, d\zeta \\
    &= -\int_{\R^{2\nc}} \nabla_{\zeta}\phi(\zeta) \pare{ \int_{\manq_{\Xi}(\zeta)} \ph(q,p) \, \sigma_{\manq_\Xi(\zeta)}(dq\,dp) } \, d\zeta\\
    &=  -\int_{\R^{6N}}  \sgram^{-1} \poisson{\Xi, \phi \circ \Xi }       \ph \,  \op{det}(\sgram)^{1/2}  \,    dq \, dp,
  \end{align*}
  where in the last line the following chain rule has been used:
  $$
  \poisson{\Xi, \phi \circ \Xi }(q,p) = \poisson{\Xi,\Xi}(q,p) \nabla_{\zeta}\phi(\Xi(q,p))= \Gamma(q,p) \nabla_{\zeta}\phi(\Xi(q,p)).
  $$
  Now an integration by parts with respect to $d q\, dp$, together with~\eqref{eq:Hnull}, leads to
  \begin{align*}
    I&= \sum_{b=1}^{2m}\int_{\R^{6N}}  \phi \circ \Xi \poisson{\Xi_{b}, \op{det}(\sgram)^{1/2}\sgram^{-1}_{.,b} \ph}  \,  dq\, dp \\
    &= \int_{\R^{6N}} \phi \circ \Xi \, \sgram^{-1} \poisson{\Xi,\ph} \op{det}(\sgram)^{1/2}\, dq\, dp \\
    &=\int_{\R^{2\nc}} \phi(\zeta) \pare{ \int_{\manq_{\Xi}(\zeta)}     \sgram^{-1} \poisson{\Xi,\ph}  \, d\sigma_{\manq_\Xi(\zeta)} } d \zeta,
  \end{align*}
  which gives the result.
\end{proof}
 

\end{document}